\documentclass[11pt]{amsart}
\usepackage{geometry}                
\usepackage[latin1]{inputenc}   
\usepackage{mathpazo}  
\usepackage[english]{babel}
\usepackage[usenames,dvipsnames]{color}
\usepackage{url}
\usepackage{amsmath,amssymb,latexsym,mathrsfs,amssymb,amsthm}
\usepackage{enumerate}
\usepackage{mathtools} 
\usepackage[all]{xy}
\usepackage{setspace}
\usepackage[colorlinks=true,linkcolor=PineGreen, citecolor=DarkOrchid, urlcolor=BurntOrange, pagebackref]{hyperref}

\usepackage{graphicx,tikz, tkz-euclide}
\usepackage{float, pgfplots, overpic}
\usetikzlibrary{arrows,calc, decorations.markings, positioning,fixedpointarithmetic}
\usetkzobj{all}

\newtheorem{theorem}{Theorem}[section]
\newtheorem{lemma}[theorem]{Lemma}
\newtheorem{prop}[theorem]{Proposition}
\newtheorem{coro}[theorem]{Corollary}

\theoremstyle{definition}
\newtheorem{definition}[theorem]{Definition}

\theoremstyle{remark}
\newtheorem{remark}[theorem]{Remark}
\newtheorem{rhp}{Riemann-Hilbert Problem}
\newcommand{\rhref}[1]{Riemann-Hilbert Problem~\ref{#1}}

\let\Re=\undefined\DeclareMathOperator{\Re}{Re}
\let\Im=\undefined\DeclareMathOperator{\Im}{Im}
\DeclareMathOperator{\res}{Res}

\DeclareMathOperator{\diag}{diag}

\newcommand{\D}{\ensuremath{\,\mathrm{d}}}
\newcommand{\I}{\ensuremath{\mathrm{i}}}
\newcommand{\E}{\ensuremath{\,\mathrm{e}}}
\newcommand{\defeq}{\vcentcolon=}
\newcommand{\eqdef}{=\vcentcolon}

\makeatletter
\renewcommand*\env@matrix[1][\arraystretch]{%
  \edef\arraystretch{#1}%
  \hskip -\arraycolsep
  \let\@ifnextchar\new@ifnextchar
  \array{*\c@MaxMatrixCols c}}
\makeatother

\let\originalleft\left
\let\originalright\right
\renewcommand{\left}{\mathopen{}\mathclose\bgroup\originalleft}
\renewcommand{\right}{\aftergroup\egroup\originalright}

\title{A Robust Inverse Scattering Transform for the Focusing Nonlinear Schr\"odinger Equation}

\author{Deniz Bilman}
 \author{Peter D.~Miller}

\address{
University of Michigan, Ann Arbor, MI 48109, USA.
}
\email{bilman@umich.edu}
\email{millerpd@umich.edu}

\date{\today}

\pgfplotsset{compat=1.10}
\begin{document}
 \begin{abstract}
We propose a modification of the standard inverse scattering transform for the
focusing nonlinear Schr\"odinger equation (also other equations by natural generalization) formulated with nonzero boundary conditions at infinity. The
purpose is to deal with arbitrary-order poles and potentially severe spectral singularities in a
simple and unified way. As an application, we use the modified transform to place the Peregrine
solution and related higher-order ``rogue wave'' solutions in an inverse-scattering context for the first time.
This allows one to directly study properties of these solutions such as their dynamical or structural stability, or their asymptotic behavior in the limit of high order.  The modified transform method
also allows rogue waves to be generated on top of other structures by elementary Darboux
transformations, rather than the generalized Darboux transformations in the literature or other related limit processes.
\end{abstract}

\maketitle
\section{Introduction}
Consider the Cauchy initial-value problem for the focusing cubic nonlinear Schr\"odinger (NLS) equation in one space dimension written in the form
\begin{equation}
\I\frac{\partial\psi}{\partial t} +\frac{1}{2}\frac{\partial^2\psi}{\partial x^2} + (|\psi|^2-1)\psi=0,\quad x\in\mathbb{R},\quad t>0.
\label{eq:NLS}
\end{equation}
Here we have written the equation in a ``rotating frame'' by replacing the traditional nonlinear term $|\psi|^2\psi$ with $(|\psi|^2-1)\psi$ which ensures that $\psi(x,t)\equiv 1$ is an exact solution, to which we demand solutions decay for large $|x|$ by imposing the boundary conditions
\begin{equation}
\lim_{x\to\pm\infty}\psi(x,t)=1,\quad t>0.
\label{eq:NZBC}
\end{equation}
The formulation of the Cauchy problem is complete (modulo technical details of the sense in which \eqref{eq:NLS} is satisfied and in which the side conditions are achieved, usually encoded in the choice of an appropriate space) by specifying an initial condition
\begin{equation}
\psi(x,0)=\psi_0(x),\quad x\in\mathbb{R}.
\label{eq:IC}
\end{equation}
The equation \eqref{eq:NLS} frequently arises in an asymptotic limit as an amplitude equation for weakly-nonlinear, monochromatic, strongly dispersive waves in conservative systems.  In this setting, the exact solution $\psi(x,t)\equiv 1$ corresponds in the physical system to a uniform traveling wavetrain, nearly sinusoidal due to weak nonlinearity.  Such nearly sinusoidal waves were studied by G. G. Stokes in the setting of surface water waves and the corresponding solution $\psi(x,t)\equiv 1$ of \eqref{eq:NLS} is sometimes called a \emph{Stokes wave}.  The Cauchy problem \eqref{eq:NLS}--\eqref{eq:IC} is therefore a mathematical formulation of the question of what becomes of a spatially-localized perturbation of the Stokes wave.  We will also follow convention and frequently refer to the solution $\psi=\psi_\mathrm{bg}(x,t)\equiv 1$ as the \emph{background}.

One can check that the solution of the Cauchy problem \eqref{eq:NLS}--\eqref{eq:IC} with initial data
\begin{equation}
\psi_0(x):= 1-4\frac{1-2\I t_0}{1+4(x-x_0)^2+4t_0^2}
\end{equation}
(here $(x_0,t_0)\in\mathbb{R}^2$ is a pair of parameters) is the \emph{Peregrine breather} \cite{Peregrine1983}
\begin{equation}
\psi(x,t)=\psi_\mathrm{P}(x,t):=1-4\frac{1+2\I (t-t_0)}{1+4(x-x_0)^2+4(t-t_0)^2}.
\label{eq:Peregrine}
\end{equation}
This solution may be viewed as a model for a ``rogue wave'':  it converges uniformly to the background solution $\psi_\mathrm{bg}(x,t)\equiv 1$ as $t\to\pm\infty$, but has a peak at $x=x_0$ that occurs at the time $t=t_0$ in the sense that $\max_{(x,t)\in\mathbb{R}^2}|\psi(x,t)|=|\psi(x_0,t_0)|=3$.  See Figure~\ref{f:Peregrine}.
\begin{figure}
    \includegraphics[scale=0.42]{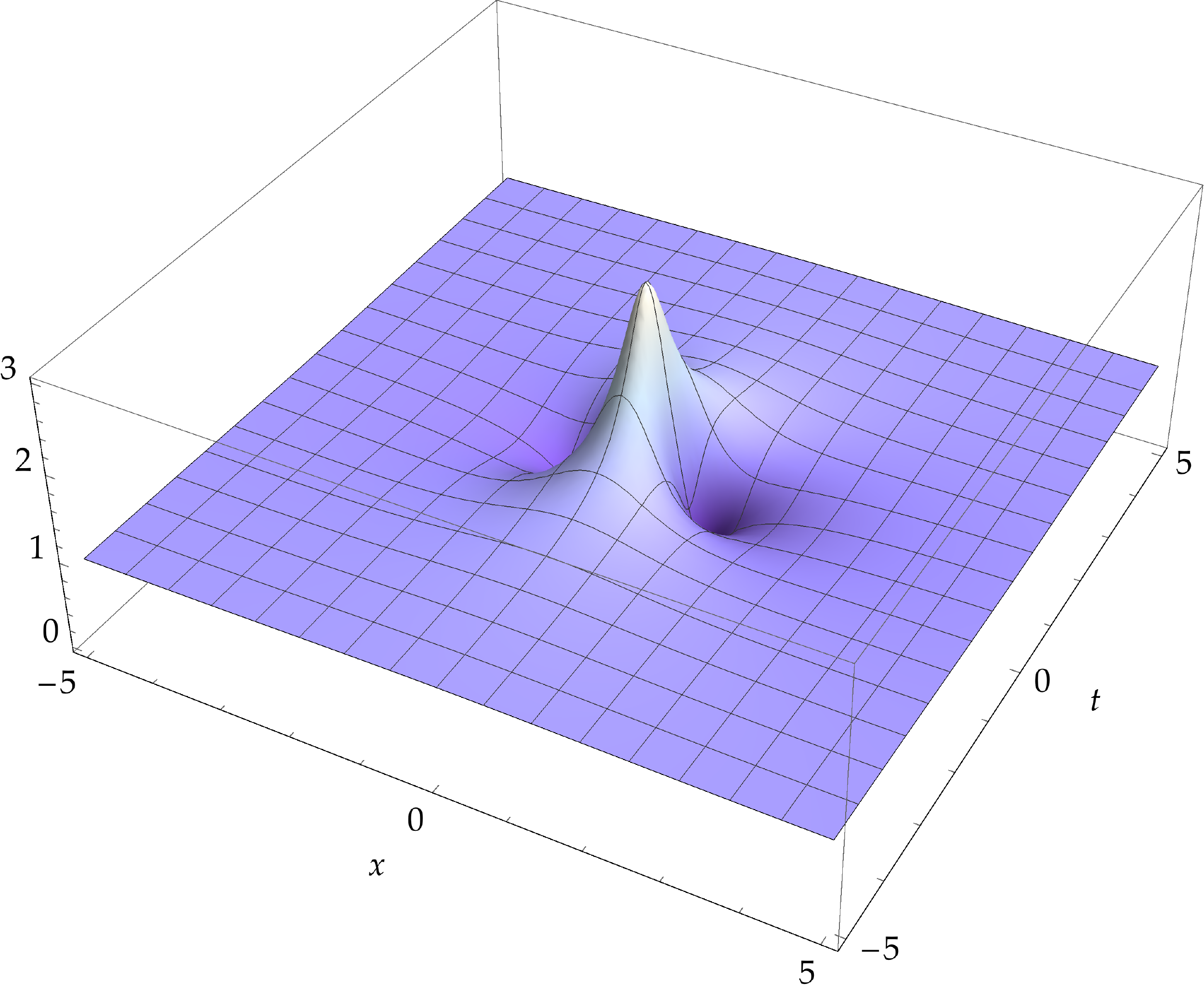}
    \caption{The modulus $|\psi_\mathrm{P}(x,t)|$ of the Peregrine breather solution $\psi_\mathrm{P}(x,t)$ \eqref{eq:Peregrine} of the focusing NLS equation \eqref{eq:NLS} with its peak located at $(x_0,t_0)=(0,0)$.}
    \label{f:Peregrine}
\end{figure}
In studies of surface water waves, ``rogue'', ``extreme'', or ``freak'' waves are commonly defined as waves that are spatially and temporally localized (hence coming out of nowhere and disappearing again without a trace) disturbances of amplitudes exceeding that of the background by an order of magnitude.  Rogue waves are viewed as rare events (and the background amplitude is usually defined in the sense of a statistical average) but they can cause damage to ships and fixed structures such as oil drilling platforms and for these reasons they have attracted substantial scientific attention in recent years.
See \cite{GuoLZLY2017-book,Kharif2009} for a review of these notions.  Rogue waves have been produced in wave-tank experiments \cite{Chabchoub2011}.  A recent application of the focusing NLS equation \eqref{eq:NLS} to model ocean surface rogue waves is the statistical, data-driven computational tool described in \cite{Cousins2016a} (see also this \href{http://news.mit.edu/2016/prediction-tool-rogue-waves-0225}{online article}) which aims to give a warning of a few minutes for the appearance of a rogue wave.  

Of course, as the focusing NLS equation \eqref{eq:NLS} is a universal asymptotic model for weakly-nonlinear waves, analogues of the Peregrine solution and of rogue waves more generally are expected to occur in a myriad of physical systems.  One landmark in this direction was the experimental observation \cite{Kibler2010,Hammani2011} of extreme light waves well-modeled by the Peregrine solution in nonlinear optical fibers.  Some more recent experiments involving the generation and observation of Peregrine-like rogue waves in nonlinear optics are described in \cite{Chabchoub2014,Suret2016,Walczak2015}.  The optical experiments in particular are very interesting because they suggest some sort of robustness (if not proper dynamical stability) of Peregrine rogue waves in the sense that the experiments are reliable and repeatable (and in some cases the reported results are actually averages over many propagating pulses).

Another context in which the Peregrine breather solution \eqref{eq:Peregrine} reliably appears is in the examination of solutions of the focusing NLS equation \eqref{eq:NLS} with \emph{general} initial data but in a certain semiclassical limit.  The semiclassical limit makes the characteristic length scale and duration of the Peregrine solution small compared to the length scale in the initial data and the length of time over which the solution is observed.  Thus it becomes relatively unimportant exactly which kind of boundary conditions are imposed at infinity.  In \cite{Bertola2013} the semiclassical scaling for \eqref{eq:NLS} was studied under the assumption that the boundary condition \eqref{eq:NZBC} was replaced with 
\begin{equation}
  \lim_{x\to\pm\infty}\psi(x,t)=0,\quad t>0,
  \label{eq:ZBC}
\end{equation}
and it was shown that for a large class of (implicitly-specified) initial data the solution of the Cauchy problem develops a universal wave pattern in which numerous Peregrine breathers appear at predictable (via the pole distribution of the \emph{tritronqu\'ee} Painlev\'e-I transcendent) locations $(x_0,t_0)$ within a wedge-shaped domain in the $(x,t)$-plane opening from its vertex in the positive $t$-direction.  There are many Peregrine breathers centered at  points $(x_0,t_0)$ that, in the semiclassical scaling, are very far from one another compared to the characteristic scales of the breathers themselves.  All of these breathers rest upon a common nonzero background field that is asymptotically constant on length scales that contain arbitrarily many breathers, although on larger spatial scales this background has to decay to zero to match the mathematically-imposed boundary conditions.  This kind of pattern of waves might seem to be a strictly theoretical phenomenon, and yet it has recently also been observed in optical experiments \cite{Tikan2017}.

There are many papers that describe generalizations of the elementary Peregrine solution \eqref{eq:Peregrine}, making use of the integrable structure behind the focusing NLS equation \eqref{eq:NLS} to construct ``higher-order'' versions of $\psi_\mathrm{P}(x,t)$ by mostly algebraic means such as the Hirota method or Darboux-like transformations leading to determinantal or Wronskian formulae for exact solutions with interesting properties.  An interesting point that was observed from the beginning is that it is not possible to generate $\psi_\mathrm{P}(x,t)$ or its higher-order analogues from the standard Darboux transformation method.  Instead, one must take certain limits within the standard Darboux transformation method as was done in \cite{AkhmedievAS09} and put into the framework of a generalized Darboux scheme in \cite{Guo2012a}.  See also \cite{Ankiewicz2010,Dubard2010b,Gaillard2013,Gaillard2014a,HeZWPF13}.  The fact that limits are required at some level mirrors what is also required to obtain these solutions via the inverse-scattering method as it is currently understood.

Despite all of this activity, and all of the knowledge available to date on the Peregrine solution and its generalizations mentioned above, \emph{none of these solutions has been obtained directly as the solution of a Cauchy problem of the form \eqref{eq:NLS}--\eqref{eq:IC} via an inverse-scattering transform}.  This makes it challenging to address questions that would seem natural from the point of view of soliton theory:
\begin{itemize}
\item
Are rogue wave solutions stable to localized perturbations?  If not, what kinds of perturbations excite instabilities, and what is the long-term nonlinear saturation of the instabilities?
\item 
Is it possible to make a prediction based on the computation of some relevant scattering data as to how many Peregrine-like peaks will be generated from a localized perturbation of the background?  In other words, what kind of initial conditions generate rogue waves?
\item 
How do rogue waves interact with other coherent structures, such as the time-periodic Kuznetsov-Ma soliton \cite{Kuznetsov1977,Ma1979b} and its Galilean-boosted generalization sometimes called the Tajiri-Watanabe soliton \cite{Tajiri1998a}?  How do they interact with more general waves that are not in the realm of ``exact solutions''?
\item 
Since it is possible to define mathematically a ``rogue wave of order $N$'' for arbitrary $N$, what can be said quantitatively about the asymptotics for large $N$?
\end{itemize}
In this paper, we aim to resolve this difficulty by generalizing the inverse-scattering solution for the Cauchy problem \eqref{eq:NLS}--\eqref{eq:IC} currently in the literature to capture rogue waves.  As will be seen, rogue waves correspond to a particular type of \emph{spectral singularity}; the method we propose handles these and other types of spectral singularities, as well as poles of arbitrary order, in a universal framework.  Moreover, the method also makes taking limits \cite{AkhmedievAS09} or related constructions such as the generalized form of the Darboux transformation \cite{Guo2012a} unnecessary in generating high-order rogue waves from the background field, and it yields an analytical setting in which the large $N$ asymptotics of $N^\mathrm{th}$ order rogue wave solutions can be effectively studied.

Before we proceed, we give a remark on notational conventions used in this paper.
\begin{remark}[Notational convention]
  We denote by $a^*$ the complex conjugate of a complex number $a$. When used with a matrix, $\mathbf{A}^*$ denotes the element-wise complex conjugate, without transposition. We use the ``dagger'' symbol $\mathbf{A}^\dagger$ to denote the conjugate transpose (Hermitian) of a matrix $\mathbf{A}$. In this paper, $\bar{a}$ is \underline{not} used used to denote the complex conjugate of a quantity $a$. With the exception of the Pauli matrices defined by
  \begin{equation}
    \sigma_1 \defeq \begin{bmatrix} 0 & 1 \\ 1 & 0 \end{bmatrix}\,,\quad\sigma_2 \defeq \begin{bmatrix} 0 & -\I \\ \I & 1 \end{bmatrix}\,,\quad\sigma_3 \defeq \begin{bmatrix} 1 & 0 \\ 0 & - 1 \end{bmatrix}\,,
    \label{eq:Pauli}
\end{equation}
and the identity matrix $\mathbb{I}$, we use boldface capital letters to denote matrices and boldface lower-case letters to denote column vectors.
\end{remark}
\subsection{Rogue waves and the known inverse-scattering transform for the Cauchy problem \eqref{eq:NLS}--\eqref{eq:IC}.}
To show what goes wrong when rogue waves are considered in the inverse-scattering transform, we begin by reviewing the method that is already in the literature.  The existing method was introduced by Ma \cite{Ma1979b}, further elucidated by Faddeev and Takhtajan \cite{FaddeevBook}, and substantially developed by Biondini and Kova\v{c}i\v{c} \cite{Biondini2014}, where the inverse problem was formulated as a matrix Riemann-Hilbert problem.  This method has subsequently been used to prove a remarkable result \cite{BiondiniM2017}, namely that under some conditions the solution of the Cauchy problem \eqref{eq:NLS}--\eqref{eq:IC} takes on a universal form as $t\to\infty$ that is independent of the initial condition \eqref{eq:IC}.  

The focusing NLS equation in the form \eqref{eq:NLS} is the $\lambda$-independent compatibility condition for the simultaneous linear equations of a Lax pair~\cite{Shabat1972}
\begin{equation}
\mathbf{u}_x = \mathbf{X}(\lambda;x,t)\mathbf{u},\quad \mathbf{X}(\lambda;x,t):=\begin{bmatrix}-\I\lambda & \psi\\-\psi^* & \I\lambda\end{bmatrix},
\label{eq:Lax-x}
\end{equation}
and
\begin{equation}
\mathbf{u}_t = \mathbf{T}(\lambda;x,t)\mathbf{u},\quad\mathbf{T}(\lambda;x,t):=\begin{bmatrix}-\I\lambda^2 + \I\tfrac{1}{2}\left(|\psi|^2 -1 \right) & \lambda\psi + \I\tfrac{1}{2}\psi_x\\
-\lambda\psi^*+\I\tfrac{1}{2}\psi_x^* & \I\lambda^2 -\I\tfrac{1}{2}\left(|\psi|^2 -1 \right)\end{bmatrix},
\label{eq:Lax-t}
\end{equation}
governing an auxiliary vector $\mathbf{u}$ (or more generally, a column of a fundamental solution matrix) depending on $(x,t)\in\mathbb{R}^2$ and the complex spectral parameter $\lambda\in\mathbb{C}$.
The existence of simultaneous solutions of this Lax pair constitutes the basis of of the inverse-scattering transform (IST) method to solve the Cauchy initial-value problem \eqref{eq:NLS}--\eqref{eq:IC}. 

\subsubsection{Review of the inverse-scattering transform}
\label{s:IST-NZBC}
We now describe the method as it was formulated in \cite{BiondiniM2017}, in which the inverse problem is set in the complex $\lambda$-plane rather than in terms of a uniformization map to the Riemann sphere \cite{Biondini2014}.
\paragraph{\underline{\smash{Direct transform}}}
The Zakharov-Shabat problem \eqref{eq:Lax-x} with the background field $\psi=\psi_\mathrm{bg}(x,t)\equiv 1$ is a constant-coefficient first-order system and hence can be solved by diagonalization (or the eigenvalue method) to obtain a fundamental a matrix of solutions: \begin{equation}
  \mathbf{U}_\mathrm{bg}(\lambda;x) \defeq n(\lambda)\begin{bmatrix} 1 & \I(\lambda-\rho(\lambda)) \\ \I(\lambda-\rho(\lambda)) & 1 \end{bmatrix}\E^{-\I\rho(\lambda)x\sigma_3}\eqdef \mathbf{E}(\lambda)\E^{-\I\rho(\lambda)x\sigma_3}\,,
  \label{eq:NZBC-asymptotic-eigenfunction}
\end{equation}
where $\rho(\lambda)$ is determined from the equation $\rho^2 =\lambda^2 + 1$ (the characteristic equation for eigenvalues $\I\rho$ of the constant coefficient matrix in \eqref{eq:Lax-x} with $\psi\equiv 1$). To be concrete, we suppose that $\rho(\lambda)$ is the function analytic for complex $\lambda$ with the exception of a vertical branch cut $\Sigma_\mathrm{c}$ between the branch points $\lambda=\pm \I$, whose square coincides with $\lambda^2+1$ and that satisfies $\rho(\lambda)=\lambda+ O(\lambda^{-1})$ as $\lambda\to\infty$. The scalar complex-valued function $n(\lambda)$ in \eqref{eq:NZBC-asymptotic-eigenfunction} is well defined as the function analytic for $\lambda\in\mathbb{C}\setminus\Sigma_\text{c}$ satisfying
\begin{equation}
  n(\lambda)^2 = \frac{\lambda + \rho(\lambda)}{2\rho(\lambda)},\quad \lim_{\lambda\to\infty} n(\lambda)=1\,.
  \label{eq:f-def}
\end{equation}
The normalization factor $n(\lambda)$ ensures that $\det(\mathbf{U}_\mathrm{bg}(\lambda;x))=1$. Note that each matrix element of the unimodular matrix $\mathbf{E}(\lambda)$ (and hence those of the fundamental matrix $\mathbf{U}_\mathrm{bg}(\lambda;x)$) has (mild) singularities at the (branch) points $\lambda=\pm \I$, where all four matrix elements of blow up as $|\lambda\mp \I|^{-1/4}$. 
We also observe that $\mathbf{U}_\mathrm{bg}(\lambda;x)$ and $\mathbf{E}(\lambda)$ are well defined as analytic functions for $\lambda\in\mathbb{C}\setminus\Sigma_\mathrm{c}$, but we may consider them as taking two different values on $\Sigma_\mathrm{c}\setminus\{-\I,\I\}$,
corresponding to evaluating $\rho(\lambda)$ and $n(\lambda)$ as boundary values taken on one side or the other.

Roughly speaking, the \emph{continuous spectrum} $\Gamma$ for the direct problem consists of those values of $\lambda$ for which $\rho(\lambda)$ is real.  Since $\rho(\lambda)$ takes distinct real boundary values on the branch cut $\Sigma_\mathrm{c}$, it will be more convenient to distinguish these boundary values by defining $\Gamma$ as the union of the real intervals $(-\infty,0)$ and $(0,+\infty)$ with the boundary of the slit domain $\mathbb{C}\setminus\Sigma_\mathrm{c}$, which is topologically a circle. So, as a set of points in the $\lambda$-plane, we might write $\Gamma:=\mathbb{R}\cup\Sigma_\mathrm{c}$, but with each $\lambda\in\Sigma_\mathrm{c}\setminus\{-\I,\I\}$ we actually associate two distinct points of $\Gamma$, corresponding to the two boundary values of $\rho(\lambda)$.  From here onwards, when we define or discuss quantities for $\lambda\in\Gamma\cap(\Sigma_\mathrm{c}\setminus\{-\I,\I\})$, we are doing so with two distinct boundary values of $\rho(\lambda)$ and $n(\lambda)$ corresponding to distinct points on the boundary of the slit domain $\mathbb{C}\setminus\Sigma_\mathrm{c}$.

Suppose that $t$ is fixed and $\Delta\psi(x,t)\defeq \psi(x,t) - 1 $ is an absolutely integrable complex-valued function of $x\in\mathbb{R}$. 
For $\lambda\in\Gamma\setminus\{-\I,\I\}$, the \emph{Jost matrix solutions} $\mathbf{U}=\mathbf{J}^\pm(\lambda;x,t)$ of \eqref{eq:Lax-x} are defined uniquely\footnote{modulo the interpretation of $\lambda\in\Gamma\cap(\Sigma_\mathrm{c}\setminus\{-\I,\I\})$  discussed above} by the boundary conditions
\begin{equation}
    \mathbf{J}^{\pm}(\lambda;x,t)\E^{\I\rho(\lambda)x\sigma_3} = \mathbf{E}(\lambda)+o(1),\quad x\to\pm\infty,
\label{eq:Jost-BC}
\end{equation}
and through the renormalization $\mathbf{K}^{\pm}(\lambda;x,t)\defeq\mathbf{J}^{\pm}(\lambda;x,t)\E^{\I\rho(\lambda)x\sigma_3}$, they can be obtained from the unique solutions of the Volterra integral equations
\begin{equation}
  \mathbf{K}^{\pm}(\lambda;x,t)=\mathbf{E}(\lambda)+\int\limits_{\pm \infty}^{x}\mathbf{E}(\lambda)\E^{-\I\rho(\lambda)(x-y)\sigma_3}\mathbf{E}(\lambda)^{-1}\mathbf{\Delta\Psi}(y,t)\mathbf{K}^{\pm}(\lambda;y,t)\E^{\I\rho(\lambda)(x-y)\sigma_3}\, \D y\,,~\rho(\lambda)\in\mathbb{R}\,,
  \label{eq:Volterra-NZBC}
\end{equation}
where
\begin{equation}
    \mathbf{\Delta\Psi}(x,t)\defeq \begin{bmatrix}0 & \Delta\psi(x,t) \\ - \Delta\psi(x,t)^* & 0 \end{bmatrix}\,.
\end{equation}
Because $\E^{\I\rho(\lambda)x\sigma_3}$ is a diagonal matrix, the following are consequences of standard analysis of the iterates that (see, for example, \cite{Biondini2014} and the references therein).
\begin{itemize}
  \item The first column $\mathbf{j}^{-,1}(\lambda;x,t)$ of $\mathbf{J}^{-}(\lambda;x,t)$ and the second column $\mathbf{j}^{+,2}(\lambda;x,t)$ of $\mathbf{J}^{+}(\lambda;x,t)$ are the boundary values of  vector-valued functions of $\lambda$ analytic in the domain $\mathbb{C}^+\setminus \Sigma_\mathrm{c}$ (corresponding to $\Im\{\rho(\lambda)\}>0$).  For $\lambda\in\Sigma_\mathrm{c}\setminus\{-\I,\I\}$ this is only relevant for $\Im\{\lambda\}\ge 0$, in which case the statement refers to the columns of the Jost matrices associated with both boundary values of $\rho(\lambda)$ and $\mathbf{E}(\lambda)$ as are needed to traverse the boundary of the slit domain $\mathbb{C}^+\setminus\Sigma_\mathrm{c}$.
  \item The first column $\mathbf{j}^{+,1}(\lambda;x,t)$ of $\mathbf{J}^{+}(\lambda;x,t)$ and the second column $\mathbf{j}^{-,2}(\lambda;x,t)$ of $\mathbf{J}^{-}(\lambda;x,t)$ are the boundary values of   vector-valued functions of $\lambda$ analytic in the domain $\mathbb{C}^-\setminus \Sigma_\mathrm{c}$ (corresponding to $\Im\{\rho(\lambda)\}<0$).  For $\lambda\in\Sigma_\mathrm{c}\setminus\{-\I,\I\}$ this is only relevant for $\Im\{\lambda\}\le 0$, with a similar caveat.
\end{itemize}
Following an analogous argument given in \cite{Demontis2013} for the defocusing problem, one can show that
if also $x^2\Delta\psi(x)\in L^1(\mathbb{R})$, then $\mathbf{j}^{-,1}$ and $\mathbf{j}^{+,2}$ are $O((\lambda-\I)^{-1/4})$ as $\lambda\to \I$ while $\mathbf{j}^{+,1}$ and $\mathbf{j}^{-,2}$ are $O((\lambda+\I)^{-1/4})$ as $\lambda\to -\I$. 
Moreover, $\det(\mathbf{J}^{\pm}(\lambda;x,t))=1$ for $\lambda\in \Gamma\setminus\{-\I,\I\}$, and $\mathbf{J}^\pm(\lambda;x,t)$ are both fundamental matrices of solutions of \eqref{eq:Lax-x} with $\psi=\psi(x,t)$ for $\lambda\in \Gamma\setminus\{-\I,\I\}$.  Thus, they satisfy the \emph{scattering relation} $\mathbf{J}^{+}(\lambda;x,t) = \mathbf{J}^{-}(\lambda;x,t)\mathbf{S}(\lambda;t)$, where $\mathbf{S}(\lambda;t)$ is called the \emph{scattering matrix} defined for $\lambda\in\Gamma\setminus\{-\I,\I\}$ (in the generalized sense of $\Gamma$ described above), and $\det(\mathbf{S}(\lambda;t))=1$. 
We may write the scattering matrix in the following form
\begin{equation}
    \mathbf{S}(\lambda;t)\eqdef \begin{bmatrix} \bar{a}(\lambda;t) & \bar{b}(\lambda;t) \\ - b(\lambda;t) & a(\lambda;t)\end{bmatrix}
\end{equation}
with
\begin{equation}
    \begin{aligned}
        a(\lambda;t) &\defeq \det \left(\begin{bmatrix} \mathbf{j}^{-,1}(\lambda;x,t); & \mathbf{j}^{+,2}(\lambda;x,t)\end{bmatrix}\right)\,,\quad
        \bar{a}(\lambda;t) \defeq \det\left(\begin{bmatrix} \mathbf{j}^{+,1}(\lambda;x,t); & \mathbf{j}^{-,2}(\lambda;x,t)\end{bmatrix}\right)\,,\\
        b(\lambda;t) &\defeq \det\left(\begin{bmatrix} \mathbf{j}^{+,1}(\lambda;x,t); & \mathbf{j}^{-,1}(\lambda;x,t)\end{bmatrix}\right)\,,\quad
        \bar{b}(\lambda;t) \defeq \det\left(\begin{bmatrix} \mathbf{j}^{+,2}(\lambda;x,t); & \mathbf{j}^{-,2}(\lambda;x,t)\end{bmatrix}\right)\,,
    \end{aligned}
    \label{eq:a-b-NZBC}
\end{equation}
and the fact that these determinants are independent of $x$ can be seen another way as a consequence of Abel's theorem. 

These Wronskian formulae show that $a(\lambda;t)$ and $\bar{a}(\lambda;t)$ extend analytically to the domains $\Im\{\rho(\lambda)\}>0$ and $\Im\{\rho(\lambda)\}<0$ respectively, but $b(\lambda;t)$ and $\bar{b}(\lambda;t)$ do not necessarily enjoy analytic continuation from $\Gamma$ in any direction. Moreover, $a(\lambda;t),\,\bar{a}(\lambda;t)\to 1$ as $\lambda\to\infty$ in $\mathbb{C}^+,\,\mathbb{C}^-$ respectively \cite{Biondini2014}.  In fact, we have the following estimate for the rate at which $a(\lambda)\to 1$.
\begin{lemma}
Suppose that $\Delta\psi(\cdot,t):=\psi(\cdot,t)-1\in L^1(\mathbb{R})$, and $\Delta\psi_x(\cdot,t)\in L^1(\mathbb{R})$.  Then $a(\lambda;t)-1=O(\lambda^{-1})$ uniformly in $\mathbb{C}^+$, and moreover, $|a(\lambda;t)|>\tfrac{1}{2}$ provided that $|\lambda|\ge r[\Delta\psi]$, where $r[\Delta\psi]$ is defined by \eqref{eq:lambda-lower-bound}.
\label{lemma:radius}
\end{lemma}
We give the proof in Appendix~\ref{A:lemma-radius-proof}.  If $(1+x^2)\Delta\psi(x)\in L^1(\mathbb{R})$, then it follows from the estimates of the Jost solutions for $\lambda\approx\I$ and the Wronskian formula \eqref{eq:a-b-NZBC} for $a$ that $a(\lambda)=O((\lambda-\I)^{-1/2})$ as $\lambda\to\I$.  In the special case that $\lambda=\I$ is a ``virtual level'' \cite[Appendix B]{Biondini2014}, one can show under a slightly stronger decay assumption on $\Delta\psi$ that $a(\lambda)=O(1)$ as $\lambda\to\I$ (this is the case for the background potential, for instance, for which $a(\lambda)\equiv 1$).    
The ratio 
\begin{equation}
R(\lambda;t)\defeq b(\lambda;t)/a(\lambda;t)\,,\quad \lambda\in\Gamma\setminus\{-\I,\I\}\,,
\label{eq:reflection-coeff}
\end{equation}
(for $\Gamma$ interpreted in the generalized sense)
is called the \emph{reflection coefficient} and it does not necessarily have any analytic extension from $\Gamma$. We also define $\bar{R}(\lambda;t):=\bar{b}(\lambda;t)/\bar{a}(\lambda;t)$ for $\lambda\in\Gamma$.

The coefficient matrix $\mathbf{X}(\lambda;x,t)$ in \eqref{eq:Lax-x} satisfies the Schwarz symmetry 
\begin{equation}
\sigma_2 \mathbf{X}(\lambda^*;x,t)^*\sigma_2=\mathbf{X}(\lambda;x,t).
\label{eq:Schwarz-X}
\end{equation}
This implies that if $\mathbf{u}(\lambda;x,t)$ is a solution of \eqref{eq:Lax-x} for $\lambda\in\Gamma\setminus\{-\I,\I\}$, then so is $\I\sigma_2 \mathbf{u}(\lambda^*,x,t)^*$. Since the Jost matrices $\mathbf{J}^{\pm}(\lambda;x,t)$ are uniquely determined by their boundary conditions \eqref{eq:Jost-BC} for $\lambda\in\Gamma\setminus\{-\I,\I\}$, it follows that 
\begin{equation}
\label{eq:Jost-sym}
\mathbf{J}^{\pm}(\lambda;x,t)=\sigma_2\mathbf{J}^{\pm}(\lambda^{*};x,t)^{*}\sigma_2,\quad\lambda\in\Gamma\setminus\{-\I,\I\},
\end{equation}
where we used the property that $\rho(\lambda)=\rho(\lambda^*)^*$ along with $\mathbf{E}(\lambda)=\sigma_2\mathbf{E}(\lambda^*)^*\sigma_2$. Thus the scattering matrix also enjoys the same Schwarz symmetry on its domain, which implies that $\bar{a}(\lambda;t)=a(\lambda^*;t)^*$ and $\bar{b}(\lambda;t)=b(\lambda^*;t)^*$ for $\lambda\in\Gamma\setminus\{-\I, \I\}$. Moreover, the analytic continuations admitted by the Jost column vector solutions of \eqref{eq:Lax-x} also satisfy the following symmetries in the indicated domains:
\begin{align}
\mathbf{j}^{-,1}(\lambda;x,t)&=\I\sigma_2 \mathbf{j}^{-,2}(\lambda^*;x,t)^*,\quad \lambda\in \mathbb{C}^+\setminus\Sigma_\mathrm{c}\,,\label{eq:Jost-minus-col-sym}\\
\mathbf{j}^{+,1}(\lambda;x,t)&=\I\sigma_2 \mathbf{j}^{+,2}(\lambda^*;x,t)^*,\quad \lambda\in \mathbb{C}^-\setminus\Sigma_\mathrm{c}\label{eq:Jost-plus-col-sym}\,.
\end{align}
These relations extend the symmetry $\bar{a}(\lambda;t)=a(\lambda^*;t)^*$ to the values of $\lambda\in\mathbb{C}^{-}\setminus\Sigma_\mathrm{c}$. 

Assuming for convenience of exposition\footnote{This assumption can be removed; indeed the robust transform we shall introduce later renders it obsolete.} that $a(\lambda;t)\neq 0$ for all $\lambda\in\Gamma\setminus\{-\I,\I\}$ and that its analytic extension to $\mathbb{C}^+\setminus\Sigma_\mathrm{c}$ has only finitely many simple zeros $\xi_1(t),\dots,\xi_N(t)$, the \emph{scattering data} associated with $\psi(\cdot,t)$ is
\begin{equation}
    \mathcal{S}(\psi(\cdot,t)) \defeq \left\{ R(\lambda;t)=b(\lambda;t)/a(\lambda;t)~\text{for}~\lambda\in\Gamma\setminus\{-\I,\I\}\,,~ \{(\xi_j(t),\gamma_j(t)) \}_{j=1}^{N} \right\},
    \label{eq:traditional-scattering-data}
\end{equation}
where the nonzero proportionality constants $\gamma_1(t),\dots,\gamma_N(t)$ are defined in terms of the analytic Jost columns by (cf., \eqref{eq:a-b-NZBC})
\begin{equation}
\mathbf{j}^{-,1}(\xi_j(t);x,t)=\gamma_j(t)\mathbf{j}^{+,2}(\xi_j(t);x,t).
\label{eq:proportionality-constants}
\end{equation}

\paragraph{\underline{\smash{Time dependence}}} 
Because $\psi(x,t)$ is a solution of \eqref{eq:NLS} by assumption, it is easy to verify that the unimodular matrices $\mathbf{J}^{\pm}(\lambda;x,t)\E^{-\I\rho(\lambda)\lambda t\sigma_3}$ are both fundamental matrices of \emph{simultaneous} solutions of the Lax pair \eqref{eq:Lax-x}--\eqref{eq:Lax-t} for $\lambda\in\Gamma\setminus\{-\I,\I\}$ and from this follows the time evolution of the scattering matrix:
\begin{equation}
    \mathbf{S}(\lambda;t)=\E^{-\I\rho(\lambda)\lambda t\sigma_3}\mathbf{S}(\lambda;0)\E^{\I\rho(\lambda)\lambda t\sigma_3}\,,\quad \mathbf{S}(\lambda;0)\defeq \mathbf{J}^{-}(\lambda;x,0)^{-1}\mathbf{J}^{+}(\lambda;x,0)\,.
    \label{eq:S-mat-evolution}
\end{equation}
From this formula, it follows that $a(\lambda;t)=a(\lambda;0)=:a(\lambda)$, and hence the number of zeros $N$ is also constant as are the zeros themselves:  $\xi_j(t)=\xi_j(0)=:\xi_j$, $j=1,\dots,N$.  It also follows that $R(\lambda;t)=R(\lambda;0)\E^{2\I\rho(\lambda)\lambda t}$ holds for $\lambda\in\Gamma\setminus\{-\I,\I\}$.  Finally, it can be shown that $\gamma_j(t)=\gamma_j(0)\E^{2\I\rho(\xi_j)\xi_j t}$, $j=1,\dots,N$.  Therefore, the scattering data $\mathcal{S}(\psi(\cdot,t))$ enjoys explicit and elementary dependence on $t$ when $\psi(\cdot,t)$ evolves in time according to the Cauchy problem \eqref{eq:NLS}--\eqref{eq:IC}.

\paragraph{\underline{\smash{Inverse transform}}}
The \emph{Beals-Coifman} simultaneous solution of \eqref{eq:Lax-x}--\eqref{eq:Lax-t} is the sectionally meromorphic matrix-valued function
\begin{equation}
    \mathbf{U}^{\mathrm{BC}}(\lambda;x,t) \defeq
    \begin{cases}
        \begin{bmatrix} a(\lambda)^{-1}\mathbf{j}^{-,1}(\lambda;x,t)\E^{-\I\rho(\lambda)\lambda t}\,; & \mathbf{j}^{+,2}(\lambda;x,t)\E^{\I\rho(\lambda)\lambda t}\end{bmatrix}\,,
        &\lambda\in\mathbb{C}^+\setminus\Sigma_\mathrm{c}\,,\vspace{0.25em}\\
        \begin{bmatrix} \mathbf{j}^{+,1}(\lambda;x,t)\E^{-\I\rho(\lambda)\lambda t}\,; & \bar{a}(\lambda)^{-1}\mathbf{j}^{-,2}(\lambda;x,t)\E^{\I\rho(\lambda)\lambda t}\end{bmatrix}\,,
        &\lambda\in\mathbb{C}^-\setminus\Sigma_\mathrm{c}\,,
    \end{cases}
    \label{eq:U-mat-NZBC}
\end{equation}
and it is convenient to introduce the related function\footnote{By \eqref{eq:Jost-minus-col-sym}, \eqref{eq:Jost-plus-col-sym},  for $\lambda\in\mathbb{C}^{+}\setminus\Sigma_\mathrm{c}$ with $a(\lambda)\neq 0$, the matrix function $\mathbf{M}^\mathrm{BC}(\lambda;x,t)$ is uniquely determined by fact that the columns of the associated matrix $\mathbf{U}^\mathrm{BC}(\lambda;x,t)$ satisfy the differential equation \eqref{eq:Lax-x} and that boundary conditions 
\begin{equation}
  \begin{aligned}
  \mathbf{M}^{\mathrm{BC}}(\lambda;x,t)&\to\mathbb{I}\,,~\text{as}~x\to +\infty\,,\\
  \mathbf{M}^{\mathrm{BC}}(\lambda;x,t)&~\text{is bounded as}~ x\to-\infty\,,
  \end{aligned}
  \label{eq:BC-asymptotics}
\end{equation}
hold.  These conditions can be combined into an integral equation for $\mathbf{M}^\mathrm{BC}(\lambda;x,t)$ of 
Fredholm type (see, for example, \cite{BealsCoifman} for details).}
\begin{equation}
\mathbf{M}^\mathrm{BC}(\lambda;x,t):=\mathbf{U}^\mathrm{BC}(\lambda;x,t)\E^{\I\rho(\lambda)(x+\lambda t)\sigma_3},\quad \lambda\in\mathbb{C}\setminus\Gamma.
\label{eq:M-mat-NZBC}
\end{equation}
$\mathbf{M}^\mathrm{BC}(\lambda;x,t)$ has simple poles at the points $\{\xi_1,\dots,\xi_N\}$ and $\{\xi_1^*,\dots,\xi_N^*\}$, the zeros of $a(\lambda)$ and $\bar{a}(\lambda):=a(\lambda^*)^*$ respectively, and no other singularities in its domain of definition.  It satisfies the normalization condition 
\begin{equation}
\lim_{\lambda\to\infty}\mathbf{M}^{\mathrm{BC}}(\lambda;x,t)=\mathbb{I},
\label{eq:MBC-norm}
\end{equation}
and $\det(\mathbf{M}(\lambda;x,t))=1$ holds identically. It also satisfies the Schwarz symmetry condition
\begin{equation}
\mathbf{M}^\mathrm{BC}(\lambda^*;x,t)=\sigma_2\mathbf{M}^\mathrm{BC}(\lambda;x,t)^*\sigma_2,\quad \lambda\in\mathbb{C}\setminus\Gamma.
\label{eq:MBC-Schwarz}
\end{equation}
The boundary values taken by $\mathbf{M}^\mathrm{BC}(\lambda;x,t)$ on $\Gamma=\mathbb{R}\cup \Sigma_\mathrm{c}$ from either side are related by the following \emph{jump conditions}.  For $\lambda\in\mathbb{R}$, $\lambda\neq 0$, $\mathbf{M}^\mathrm{BC}(\cdot;x,t)$ has well-defined nontangential boundary values $\mathbf{M}^\mathrm{BC}_\pm(\lambda;x,t):=\lim_{\epsilon\downarrow 0}\mathbf{M}^\mathrm{BC}(\lambda\pm i\epsilon;x,t)$, and from the scattering relation $\mathbf{J}^+=\mathbf{J}^-\mathbf{S}$ and the definition of $\mathbf{M}^\mathrm{BC}(\lambda;x,t)$ we have
\begin{equation}
    \mathbf{M}^{\mathrm{BC}}_+(\lambda;x,t) =
     \mathbf{M}^{\mathrm{BC}}_-(\lambda;x,t)\E^{-\I\rho(\lambda)(x+\lambda t)\sigma_3}\mathbf{V}^\mathbb{R}(\lambda) \E^{\I\rho(\lambda)(x+\lambda t)\sigma_3}, \quad\lambda\in\mathbb{R}\setminus\{0\},
    \label{eq:traditional-jump-real}
\end{equation}
where
\begin{equation}
\mathbf{V}^\mathbb{R}(\lambda):=\begin{bmatrix}1+|R(\lambda;0)|^2 & R(\lambda;0)^* \\ R(\lambda;0) & 1 \end{bmatrix},\quad \lambda\in\mathbb{R}\setminus\{0\}.
\label{eq:VR}
\end{equation}
For $\lambda=\I\Im\{\lambda\}$ with $0<\Im\{\lambda\}<1$, $\mathbf{M}^\mathrm{BC}(\cdot;x,t)$ takes well-defined boundary values given by $\mathbf{M}^\mathrm{BC}_\pm(\lambda;x,t):=\lim_{\epsilon\downarrow 0}\mathbf{M}^\mathrm{BC}(\lambda\pm\epsilon;x,t)$ related by
\begin{equation}
\mathbf{M}^\mathrm{BC}_+(\lambda;x,t)=\mathbf{M}^\mathrm{BC}_-(\lambda;x,t)
\E^{-\I\rho_-(\lambda)(x+\lambda t)\sigma_3}\mathbf{V}^\downarrow(\lambda)\E^{\I\rho_+(\lambda)(x+\lambda t)\sigma_3},\quad \lambda=\I\Im\{\lambda\},\quad
0<\Im\{\lambda\}<1,
\label{eq:traditional-jump-Sigmac-up}
\end{equation}
where
\begin{equation}
\mathbf{V}^\downarrow(\lambda):=\begin{bmatrix}
\I\bar{R}_-(\lambda;0) & -\I\\-\I(1+R_-(\lambda;0)\bar{R}_-(\lambda;0)) & \I R_-(\lambda;0)
\end{bmatrix},\quad \lambda=\I\Im\{\lambda\},\quad
0<\Im\{\lambda\}<1.
\end{equation}
Here, $\rho_\pm(\lambda):=\lim_{\epsilon\downarrow 0}\rho(\lambda\pm\epsilon)$, and $R_-(\lambda;0)$ and $\bar{R}_-(\lambda;0)$ refer to the elements of the scattering matrix constructed from Jost solutions defined using the boundary value $\rho_-(\lambda)$ as well as the corresponding boundary value $\mathbf{E}_-(\lambda)$ taken by $\mathbf{E}(\cdot)$.  Finally, for $\lambda=\I\Im\{\lambda\}$ with $-1<\Im\{\lambda\}<0$, $\mathbf{M}^\mathrm{BC}(\cdot;x,t)$ takes boundary values $\mathbf{M}^\mathrm{BC}_\pm(\lambda;x,t):=\lim_{\epsilon\downarrow 0}\mathbf{M}^\mathrm{BC}(\lambda\mp\epsilon;x,t)$ related by
\begin{equation}
\mathbf{M}_+^\mathrm{BC}(\lambda;x,t)=\mathbf{M}_-^\mathrm{BC}(\lambda;x,t)
\E^{-\I\rho_-(\lambda)(x+\lambda t)\sigma_3}\mathbf{V}^\uparrow(\lambda)\E^{\I\rho_+(\lambda)(x+\lambda t)\sigma_3},\;\;\lambda=\I\Im\{\lambda\},\;\; -1<\Im\{\lambda\}<0,
\label{eq:traditional-jump-Sigmac-down}
\end{equation}
where
\begin{equation}
\mathbf{V}^\uparrow(\lambda):=\mathbf{V}^\downarrow(\lambda^*)^\dagger,\quad \lambda=\I\Im\{\lambda\},\quad -1<\Im\{\lambda\}<0.
\end{equation}
The simple poles of $\mathbf{M}^\mathrm{BC}(\lambda;x,t)$ are characterized by the following \emph{residue conditions}:
\begin{equation}
\begin{split}
\mathop{\mathrm{Res}}_{\lambda=\xi_j}\mathbf{M}^\mathrm{BC}(\lambda;x,t)&=
\lim_{\lambda\to\xi_j}\mathbf{M}^\mathrm{BC}(\lambda;x,t)\E^{-\I\rho(\xi_j)(x+\xi_jt)\sigma_3}\mathbf{N}_j\E^{\I\rho(\xi_j)(x+\xi_jt)\sigma_3},\quad j=1,\dots,N,\\
\mathop{\mathrm{Res}}_{\lambda=\xi_j^*}\mathbf{M}^\mathrm{BC}(\lambda;x,t)&=
\lim_{\lambda\to\xi_j^*}\mathbf{M}^\mathrm{BC}(\lambda;x,t)\E^{-\I\rho(\xi_j^*)(x+\xi_j^*t)\sigma_3}
\sigma_2\mathbf{N}_j^*\sigma_2\E^{\I\rho(\xi_j^*)(x+\xi_j^*t)\sigma_3},\quad j=1,\dots,N,
\end{split}
\label{eq:residues}
\end{equation}
where 
\begin{equation}
\mathbf{N}_j:=\begin{bmatrix}0 & 0\\a'(\xi_j)^{-1}\gamma_j(0) & 0\end{bmatrix},\quad j=1,\dots,N.
\end{equation}

The analyticity of $\mathbf{M}^\mathrm{BC}(\cdot;x,t)$ in the domain $\mathbb{C}\setminus(\Gamma\cup\{\xi_1,\dots,\xi_N,\xi_1^*,\dots,\xi_N^*\})$, the normalization condition $\mathbf{M}^\mathrm{BC}(\lambda;x,t)\to\mathbb{I}$ as $\lambda\to\infty$, the jump conditions \eqref{eq:traditional-jump-real}, \eqref{eq:traditional-jump-Sigmac-up}, and \eqref{eq:traditional-jump-Sigmac-down}, and residue conditions \eqref{eq:residues} are all conditions on $\mathbf{M}^\mathrm{BC}(\lambda;x,t)$ involving the scattering data at $t=0$ that can be computed from the initial condition \eqref{eq:IC} alone via the direct transform.  Moreover, these Riemann-Hilbert conditions are nearly enough to determine $\mathbf{M}^\mathrm{BC}(\lambda;x,t)$.  The only missing piece of information concerns the nature of the boundary values taken by $\mathbf{M}^\mathrm{BC}(\cdot;x,t)$ on $\Gamma$.  The direct transform shows that these boundary values are continuous functions of $\lambda$ with the exception of the points $\{-\I,\I\}$.  These points are singularities of $\mathbf{M}^\mathrm{BC}(\lambda;x,t)$, as must be the case because $\mathbf{V}^\downarrow(\lambda)$ does not tend to the identity as $\lambda\to\I$.  Moreover, in the simplest case of the background potential $\psi_\mathrm{bg}(x,0)\equiv 1$, the Beals-Coifman matrix is $\mathbf{M}_\mathrm{bg}^\mathrm{BC}(\lambda;x,t)=\mathbf{E}(\lambda)$ which exhibits $-1/4$ power singularities at $\lambda=\pm\I$.  We therefore supplement the Riemann-Hilbert conditions with an additional growth condition to account for this singularity.  Taking into account the behavior in the background case, we impose the following growth condition:
\begin{equation}
\mathbf{M}^\mathrm{BC}(\lambda;x,t)=O((\lambda\mp\I)^{-1/4}),\quad \lambda\to \pm\I.
\label{eq:growth-condition}
\end{equation}
Indeed, under the assumption that $(1+x^2)\Delta\psi(x)\in L^1(\mathbb{R})$ and that $a(\I)\neq 0$ if $\lambda=\I$ is a ``virtual level'' \cite[Appendix B]{Biondini2014}, this condition correctly\footnote{It seems to us that it might be possible for $\lambda=\I$ to be a virtual level and also to have $a(\I)=0$ for some special potentials $\Delta\psi$ with $(1+x^2)\Delta\psi(x)\in L^1(\mathbb{R})$, in which case \eqref{eq:growth-condition} is violated.  The robust transform we introduce in Section~\ref{s:New-IST} will sidestep this difficulty as well.} estimates the growth rate of $\mathbf{M}^\mathrm{BC}(\lambda;x,t)$ as $\lambda\to\pm\I$.
Another way to formulate \eqref{eq:growth-condition} is simply to require that the boundary values of $\mathbf{M}^\mathrm{BC}(\cdot;x,t)$ lie in $L^2(\Gamma)$.  The growth condition is not explicitly stated in \cite{BiondiniM2017} but it is necessary to ensure uniqueness of the inverse problem.  
\paragraph{\underline{\smash{Summary of the IST}}}
The IST method to solve the Cauchy initial value problem \eqref{eq:NLS}--\eqref{eq:IC} consists of the following steps. Given suitable initial data $\psi_0(x)=\psi(x;0)$, for $\lambda\in\Gamma\setminus\{-\I,\I\}$ one calculates the Jost matrices $\mathbf{J}^{\pm}(\lambda;x,0)$ and hence $\mathbf{S}(\lambda;0)$.  Assuming finitely many simple zeros $\xi_1,\dots,\xi_N$ of $a(\cdot)=a(\cdot;0)$ in $\mathbb{C}^+\setminus\Sigma_\mathrm{c}$, 
one computes also the proportionality constants $\gamma_1(0),\dots,\gamma_N(0)$ using \eqref{eq:proportionality-constants}.  This completes the calculation of the scattering data at $t=0$ and constitutes the direct transform.

Then one allows the scattering data to evolve explicitly in time $t>0$, and uses it to formulate the inverse problem, namely seeking $\mathbf{M}^\mathrm{BC}(\lambda;x,t)$ that satisfies the Riemann-Hilbert conditions augmented with \eqref{eq:growth-condition}.  Under some technical conditions that are not of concern here, this Riemann-Hilbert problem has a unique solution, and thus $\mathbf{M}^\mathrm{BC}(\lambda;x,t)$ is constructed for all $(x,t)\in\mathbb{R}^2$ from the scattering data obtained from the initial condition \eqref{eq:IC}.  Finally,
from $\mathbf{M}^\mathrm{BC}(\lambda;x,t)$ one obtains a function $\psi(x,t)$ satisfying \eqref{eq:NLS} by the limit
\begin{equation}
  \psi(x,t)=2\I\lim_{\lambda\to\infty} \lambda M^{\mathrm{BC}}_{12}(\lambda;x,t)\,.
\end{equation}
Under some additional assumptions that in particular exclude rogue waves, $\psi(x,t)$ is then the solution of the Cauchy initial-value problem \eqref{eq:NLS}--\eqref{eq:IC}.

\subsubsection{Spectral singularities}
\label{s:spectral-singularities}
\emph{Spectral singularities} are points $\lambda$ in the continuous spectrum $\Gamma$ at which $\mathbf{M}^\mathrm{BC}(\cdot;x,t)$ fails to have a well-defined non-tangential limit from one side or the other.  There are potentially two kinds of spectral singularities in this problem.  

One type of spectral singularity is a zero of the scattering coefficient $a(\cdot)$ ($\bar{a}(\cdot)$ respectively) at a point $\lambda\neq\I$ ($\lambda\neq-\I$ respectively) on the boundary of the domain $\mathbb{C}^+\setminus\Sigma_\mathrm{c}$ ($\mathbb{C}^-\setminus\Sigma_\mathrm{c}$ respectively).  Since the Jost solutions $\mathbf{j}^{\mp,1}(\cdot;x,t)$ and $\mathbf{j}^{\pm,2}(\cdot;x,t)$ are not only analytic in $\mathbb{C}^\pm\setminus\Sigma_\mathrm{c}$ but also continuous up to the boundary except at $\lambda=\pm\I$, from \eqref{eq:U-mat-NZBC} and \eqref{eq:M-mat-NZBC} one sees that such zeros of $a(\cdot)$ or $\bar{a}(\cdot)$ are indeed spectral singularities and they are the only possible ones that can occur for $\lambda\neq \pm\I$.  From one point of view, zeros of $a$ or $\bar{a}$ in $\Gamma\setminus\{-\I,\I\}$ are to be expected because two initial conditions without such spectral singularities but with different values of $n$, the number of zeros of $a$ in $\mathbb{C}^+\setminus\Sigma_\mathrm{c}$ weighted by multiplicity, can be connected by a suitable homotopy in which there must be at least one point at which a zero of finite multiplicity is ``born'' from $\Gamma$.  An even more disturbing situation was described by Zhou \cite{Zhou1989d} in the context of the simpler Cauchy problem for \eqref{eq:NLS} with zero boundary conditions \eqref{eq:ZBC}.  For the latter problem (see \cite{Shabat1972,Ablowitz1974b,BealsCoifman} for details of the IST solution in this setting) it was shown by Beals and Coifman \cite{BealsCoifman} that for an open dense subset of initial conditions $\psi_0\in L^1(\mathbb{R})$ there are no spectral singularities and the analogue of the number $n$ is finite.  On the other hand, as shown by Zhou \cite[Example 3.3.16]{Zhou1989d}, the complement of this dense open subset contains Schwartz-class functions $\psi_0$ for which $n=\infty$ and infinitely many zeros of $a$ accumulate at certain points in the continuous spectrum $\Gamma$ (for zero boundary conditions, $\Gamma=\mathbb{R}$).  Moreover, there exist Schwartz-class $\psi_0$ for which there are infinitely many of these accumulation points, which themselves accumulate from within $\Gamma$ at particularly severe spectral singularities.  In order to deal with spectral singularities in the zero boundary condition setting, Zhou \cite{Zhou1989d} and Deift and Zhou \cite{Deift1991b} developed a method based on combining the standard Beals-Coifman matrix $\mathbf{M}^\mathrm{BC}(\lambda;x,t)$ for values of $\lambda$ suitably large that the latter has no singularities with another simultaneous solution of the Lax pair \eqref{eq:Lax-x}--\eqref{eq:Lax-t} for smaller $\lambda$.  This other solution matrix was constructed as the product of the Beals-Coifman matrix taken for a sufficiently ``cut-off'' version of $\psi_0$ that it generates no singularities and evaluated for $x$ equal to the cut-off point $x=L$ with a transfer matrix solution of \eqref{eq:Lax-x} to obtain the solution at a general $x\in\mathbb{R}$ from that at $x=L$.  This approach leads to a Riemann-Hilbert problem for a sectionally holomorphic matrix (no poles at all) that takes continuous boundary values on $\Gamma=\mathbb{R}$ as well as on a circle of large radius where the two matrix solutions are related by a computable jump condition.  To complete this brief discussion of spectral singularities caused by zeros of $a$ or $\bar{a}$, we simply make the point that while the sort of severe spectral singularities exhibited by Schwartz-class $\psi_0$ in the zero boundary conditions case have not yet been observed in the presence of non-zero boundary conditions, there is no reason to assume that they cannot occur for certain $\psi_0$ with $\Delta\psi_0:=\psi_0-1$ Schwartz-class.

The other spectral singularities in the IST solution of the Cauchy problem \eqref{eq:NLS}--\eqref{eq:IC} with nonzero boundary conditions are those at $\lambda=\pm\I$.  These are present due to the nonzero background $\psi\equiv 1$, of which the initial condition $\psi_0$ is a localized perturbation.  We have already pointed out that the matrix $\mathbf{U}_\mathrm{bg}^\mathrm{BC}(\lambda;x,t):=\mathbf{E}(\lambda)\E^{-\I\rho(\lambda)(x+\lambda t)\sigma_3}$, which is the Beals-Coifman simultaneous solution matrix of \eqref{eq:Lax-x}--\eqref{eq:Lax-t} for $\psi_0(x)\equiv 1$ blows up like $(\lambda\mp\I)^{-1/4}$ as $\lambda\to\pm\I$.  These singularities generally propagate into the Jost matrix solutions $\mathbf{J}^\pm(\lambda;x,t)$ through the integral equations \eqref{eq:Volterra-NZBC}, which involve $\mathbf{E}(\lambda)$, although from \eqref{eq:Volterra-NZBC} alone it is generally (i.e., in absence of an extra condition like $(1+x^2)\Delta\psi(x)\in L^1(\mathbb{R})$) difficult to assess how the severity of the singularities induced at $\lambda=\pm\I$ relates to $\Delta\psi(x,t)$.  We wish to point out that the cut-off potential approach of \cite{Zhou1989e} and \cite{Deift1991b} does not sidestep the issue, because \emph{these spectral singularities are features of the background potential itself} and the cut-off approach simply replaces one potential by another one that is sufficiently close in $L^1(\mathbb{R})$ to the background.

The reason that the background potential $\psi_0(x)\equiv 1$ produces spectral singularities at $\lambda=\pm\I$ in the Beals-Coifman solution $\mathbf{U}_\mathrm{bg}^\mathrm{BC}(\lambda;x,t)$ is simple:  the constant coefficient matrices in the Lax pair \eqref{eq:Lax-x}--\eqref{eq:Lax-t} fail to be diagonalizable at $\lambda=\pm\I$.  Therefore, the eigenvalue method, on which the construction of the matrix $\mathbf{U}_\mathrm{bg}^\mathrm{BC}(\lambda;x,t)$ is based, cannot produce a fundamental solution matrix when $\lambda=\pm\I$.  The spectral singularities thus appear upon normalizing the eigenvectors so that
$\det(\mathbf{U}_\mathrm{bg}^\mathrm{BC}(\lambda;x,t))=1$.  Note however that the essence of the difficulty lies in the collapse in the dimension of the eigenvector span, and this difficulty remains if one chooses to work with non-unimodular solution matrices or attempts to regularize the branch points at $\lambda=\pm\I$ by mapping the slit domain $\mathbb{C}\setminus\Sigma_\mathrm{c}$ to the Riemann sphere via the Joukowski mapping \cite{Biondini2014}. 

\subsubsection{Attempting to capture rogue waves in the IST}
The spectral singularities at $\lambda=\pm\I$ are the reason why the IST as described in Section~\ref{s:IST-NZBC} does not capture rogue waves.  Here the example of the Peregrine solution $\psi_\mathrm{P}(x,t)$ \eqref{eq:Peregrine} is particularly instructive.  Since $\psi_\mathrm{P}(\cdot,t)-1\in L^1(\mathbb{R})$ for all $t\in\mathbb{R}$, the direct transform applies.  Rather than compute the Jost solutions via the Volterra equations \eqref{eq:Volterra-NZBC}, it is more convenient to appeal to algebraic methods such as the generalized Darboux transformation \cite{Guo2012a} to construct a basis of simultaneous solutions of \eqref{eq:Lax-x}--\eqref{eq:Lax-t} for $\psi=\psi_\mathrm{P}(x,t)$ and $\lambda\in\Gamma\setminus\{-\I,\I\}$.  The result (see Corollary~\ref{c:RogueWaveScattering}) is that $\mathbf{S}=\mathbf{S}_\mathrm{P}(\lambda;0)\equiv\mathbb{I}$ for all $\lambda\in\Gamma\setminus\{-\I,\I\}$.  Therefore, there are no poles of $\mathbf{M}_\mathrm{P}^\mathrm{BC}(\lambda;x,t)$ in $\mathbb{C}\setminus\Gamma$, and the reflection coefficient is
$R=R_\mathrm{P}(\lambda;0)\equiv 0$.  This implies that \emph{the Peregrine solution $\psi=\psi_\mathrm{P}(x,t)$ is indistinguishable from the background solution $\psi_\mathrm{bg}(x,t)\equiv 1$ at the level of the scattering data \eqref{eq:traditional-scattering-data}}, i.e., the direct transform $\psi\mapsto\mathcal{S}(\psi)$ is not injective.  Indeed, $\mathcal{S}(\psi_\mathrm{P})=\mathcal{S}(\psi_\mathrm{bg})$.  Although they share the same scattering data, the matrices $\mathbf{U}^\mathrm{BC}_\mathrm{P}(\lambda;x,t)$ and $\mathbf{U}^\mathrm{BC}_\mathrm{bg}(\lambda;x,t)$ do not coincide; they are distinguished by the rate at which they blow up as $\lambda\to\pm\I$.  The difficulty lies in working out how to take this asymptotic behavior of the matrix $\mathbf{M}^\mathrm{BC}(\lambda;x,t)$ as $\lambda\to \pm\I$ into account in the formulation of the IST.  If one begins with the initial condition $\psi_0(x)=\psi_\mathrm{P}(x,0)$ and enforces the growth condition \eqref{eq:growth-condition} in the inverse problem, the IST returns instead the background solution $\psi_\mathrm{bg}(x,t)\equiv 1$.  However, if one simply replaces the growth condition \eqref{eq:growth-condition} with the estimate $\mathbf{M}^\mathrm{BC}(\lambda;x,t)=O((\lambda\mp\I)^{-3/4})$ as $\lambda\to\pm\I$ known to be correct for the matrix $\mathbf{M}^\mathrm{BC}_\mathrm{P}(\lambda;x,t)$ (see Remark~\ref{r:Peregrine-growth}), then one loses the uniqueness of the solution of the inverse problem as $\mathbf{M}^\mathrm{BC}_\mathrm{P}(\lambda;x,t)$ and $\mathbf{M}^\mathrm{BC}_\mathrm{bg}(\lambda;x,t)$ are no longer distinguished by the Riemann-Hilbert conditions augmented with the updated (weakened, really) growth condition.  All of these remarks apply to the higher-order rogue waves \cite{Ankiewicz2010} (see Figures~\ref{f:Peregrine2} and \ref{f:8th-rogue-wave}) as well, although as the order increases, so does the rate at which $\mathbf{M}^\mathrm{BC}(\lambda;x,t)$ blows up as $\lambda\to\pm\I$.  Note that neither the Peregrine solution nor its higher-order analogues satisfy the condition $(1+x^2)\Delta\psi(x)\in L^1(\mathbb{R})$ that, with some additional technical assumptions related to virtual levels at $\pm\I$, guarantees the growth estimate on $\mathbf{M}^\mathrm{BC}(\lambda;x,t)$ in the form \eqref{eq:growth-condition}.

None of this is to say that exact rogue wave solutions cannot be obtained from the IST summarized in Section~\ref{s:IST-NZBC} by a limiting process.  Indeed, it is well known that they can be found by first writing down the reflectionless multi-soliton solutions in which the parameters $\{\xi_1,\dots,\xi_N\}$ and $\{\gamma_1(0),\dots,\gamma_N(0)\}$ are retained as variables, and then subsequently taking a suitable limit in which certain $\xi_j$ tend to $\I$ (and $\gamma_j$ is suitably scaled).  However, the above discussion makes clear the fact that these solutions cannot be obtained directly.  Moreover, the limiting techniques are useful for exact (reflectionless) solutions, but they are not as useful if $R(\lambda;0)\not\equiv 0$.

The purpose of this paper is to introduce a new version of the IST for the Cauchy problem \eqref{eq:NLS}--\eqref{eq:IC} that is sufficiently robust to capture in an elementary way the missing information about the rate of growth of the Beals-Coifman matrix as $\lambda\to\pm\I$.  In fact, as the reader will see, it is never necessary to consider this rate of growth at all, as the information is encoded instead in the analytic columns of the Jost solution matrices \emph{at a finite distance from these spectral singularities}.

\subsection{Outline of the paper}
In Section~\ref{s:New-IST} we formulate the robust IST for the Cauchy problem \eqref{eq:NLS}--\eqref{eq:IC}.  Then in Section~\ref{s:Darboux} we show how the robust IST makes the application of iterated Darboux transformations easier than in the traditional approach, and in particular how the robust IST builds in automatically the missing generalized eigenvectors that have to be extracted by differentiation with respect to $\lambda$ in the generalized Darboux transformation method \cite{Guo2012a}.  Thus rogue waves of arbitrary order can be calculated from the robust IST using only standard Darboux methods.  As a more analytical application, in Section~\ref{s:Linearization} we begin to consider the question of dynamical stability/instability of rogue waves by linearizing the robust IST about a solution of the Cauchy problem \eqref{eq:NLS}--\eqref{eq:IC}.  We present the solution of the linearized problem as a kind of spectral transform involving squared eigenfunctions, which allows us to draw some preliminary conclusions about the nature of instabilities and how they essentially arise from the background itself.  Some of the more technical calculations are relegated to the appendices of the paper.

\subsection{Acknowledgements}
The authors are grateful to Liming Ling for useful discussions.  Both authors were partially supported by the National Science Foundation under grant DMS-1513054.  In addition, D. B. was supported by an AMS-Simons Travel Grant.

\section{The Robust IST}
\label{s:New-IST}
As pointed out in Section~\ref{s:spectral-singularities}, the ``cut-off'' approach of Deift and Zhou \cite{Deift1991b,Zhou1989d} does not apply in the case of nonzero boundary conditions.  However, we may take from it the following key ideas:
\begin{itemize}
\item Simultaneous solutions of the Lax pair \eqref{eq:Lax-x}--\eqref{eq:Lax-t} adapted to the boundary conditions at hand (the Beals-Coifman solutions $\mathbf{U}^\mathrm{BC}(\lambda;x,t)$) are essential to consider for large $\lambda$ in order to build in the correct asymptotic normalization condition for the Riemann-Hilbert problem of inverse scattering.
\item On the other hand, other simultaneous solutions of \eqref{eq:Lax-x}--\eqref{eq:Lax-t} can be useful for bounded $\lambda$.  In particular, solutions of \emph{initial-value problems} for \eqref{eq:Lax-x}--\eqref{eq:Lax-t} have fantastic analytic properties with respect to $\lambda$ although they fail to capture any information about boundary conditions at infinity.
\end{itemize}
To define what we mean by a solution of an initial-value problem for \eqref{eq:Lax-x}--\eqref{eq:Lax-t} and demonstrate its analytic features, we begin with the following proposition.
\begin{prop}
    Fix $L\in\mathbb{R}$. Suppose that $\psi(x,t)$ is a bounded classical solution of the focusing NLS equation \eqref{eq:NLS} defined for $(x,t)$ in a simply-connected domain $\Omega\subseteq\mathbb{R}^2$ that contains the point $(L,0)$. Then for each $\lambda\in\mathbb{C}$ there exists a unique simultaneous fundamental solution matrix $\mathbf{U}=\mathbf{U}^\mathrm{in}(\lambda;x,t)$, $(x,t)\in\Omega$, of the Lax pair equations \eqref{eq:Lax-x}--\eqref{eq:Lax-t} together with the initial condition $\mathbf{U}^\mathrm{in}(\lambda;L,0)=\mathbb{I}$.  Moreover, $\mathbf{U}^\mathrm{in}(\lambda;x,t)$ is an entire function of $\lambda$ for each $(x,t)\in\Omega$, and $\det(\mathbf{U}^\mathrm{in}(\lambda;x,t))\equiv 1$.
    \label{p:Uin-def}
\end{prop}
\begin{proof}
    The proof is based on a standard Picard iteration argument for the system of simultaneous linear differential equations \eqref{eq:Lax-x}--\eqref{eq:Lax-t} in $\Omega\subset\mathbb{R}^2$ (see, for example, \cite[Chapter 1]{CoddingtonLevinson} for a more detailed construction). Let $(x,t)\in\Omega$ be fixed. Because $\psi(x,t)$ is a solution of the focusing NLS equation \eqref{eq:NLS} in $\Omega$, the linear differential equations \eqref{eq:Lax-x}--\eqref{eq:Lax-t} are compatible, i.e., the coefficient matrices $\mathbf{X}(\lambda;x,t)$ and $\mathbf{T}(\lambda;x,t)$
satisfy the zero-curvature condition
\begin{equation}
  \mathbf{X}_t-\mathbf{T}_x +[\mathbf{X},\mathbf{T}]=0
\end{equation}
and therefore since $\Omega$ is simply connected,
we can setup a Picard iteration to integrate the simultaneous equations \eqref{eq:Lax-x}--\eqref{eq:Lax-t} along an arbitrary smooth path $\Pi$ in $\Omega$ from $(L,0)$ to $(x,t)$ and be guaranteed that the result of the iteration will be independent of the choice of path. 
Let the path $\Pi$ be parametrized by $x=X(u)$, $t=T(u)$, $0\le u\le 1$, such that $(X(0),T(0))=(L,0)$ and $(X(1),T(1))=(x,t)$.  The integral equation governing the restriction $\mathbf{W}^\Pi(\lambda;u):=\mathbf{U}^\mathrm{in}(\lambda;X(u),T(u))$ of $\mathbf{U}^\mathrm{in}(\lambda;x,t)$ to the path $\Pi$ is then
\begin{equation}
\mathbf{W}^\Pi(\lambda;u)=\mathbb{I}+\int_0^u\left[\mathbf{X}(\lambda;X(v),T(v))X'(v)+\mathbf{T}(\lambda;X(v),T(v))T'(v)\right]\mathbf{W}^\Pi(\lambda;v)\,\D v,\quad 0<u\le 1.
\end{equation}
By assumption on the classical nature of the solution $\psi(x,t)$ of \eqref{eq:NLS}, the expression in square brackets in the integrand is continuous in $v$, and this guarantees uniform convergence of the Picard iterates (possibly first chopping the interval $[0,1]$ into finitely-many subintervals and re-starting the iteration after each).  The expression in square brackets in the integrand is also a (quadratic) polynomial in $\lambda$, so by an argument based on Morera's Theorem and Fubini's Theorem, the iterates are all entire functions.  As the iterates converge uniformly on compact subsets of $\lambda\in\mathbb{C}$, the unique solution of the initial-value problem is entire in $\lambda$.  By compatibility, $\mathbf{U}^\mathrm{in}(\lambda;x,t):=\mathbf{W}^\Pi(\lambda;1)$ is well defined regardless of the choice of path $\Pi$.  The fact that $\det(\mathbf{U}^\mathrm{in}(\lambda;x,t))=1$ then follows from $\det(\mathbf{U}^\mathrm{in}(L,0))=\det(\mathbb{I})=1$ using Abel's Theorem, because $\mathrm{tr}(\mathbf{X}(\lambda;x,t))=\mathrm{tr}(\mathbf{T}(\lambda;x,t))=0$ holds for $\lambda\in\mathbb{C}$ and $(x,t)\in\Omega$.
\end{proof}
Note that the Schwarz symmetry of the coefficient matrices $\mathbf{X}(\lambda;x,t)$ and $\mathbf{T}(\lambda;x,t)$ (see \eqref{eq:Schwarz-X} for $\mathbf{X}$; the analogous result holds also for $\mathbf{T}$) of the simultaneous differential equations \eqref{eq:Lax-x} and \eqref{eq:Lax-t} implies that $\sigma_2\mathbf{U}^{\mathrm{in}}(\lambda^*;x,t)^*\sigma_2$ is also an entire simultaneous solution matrix for the system \eqref{eq:Lax-x}--\eqref{eq:Lax-t}. Since this solution matrix coincides with $\mathbf{U}^\mathrm{in}(\lambda;L,0)=\mathbb{I}$ for $(x,t)=(L,0)$, 
by uniqueness we conclude that
\begin{equation}
\mathbf{U}^{\mathrm{in}}(\lambda;x,t)=\sigma_2\mathbf{U}^{\mathrm{in}}(\lambda^*;x,t)^*\sigma_2.
\label{eq:U-in-sym}
\end{equation}
\begin{remark}
\label{r:Uin-renormalize}
In practice, $\mathbf{U}^{\mathrm{in}}(\lambda;x,t)$ can also be obtained from \emph{any} matrix $\mathbf{U}(\lambda;x,t)$ of simultaneous solutions of \eqref{eq:Lax-x}--\eqref{eq:Lax-t} defined and fundamental for $\lambda\in\mathbb{C}$ with the possible exception of isolated points or even curves and defining for some chosen $L\in\mathbb{R}$
\begin{equation}
\mathbf{U}^{\mathrm{in}}(\lambda;x,t)\defeq \mathbf{U}(\lambda;x,t)\mathbf{U}(\lambda;L,0)^{-1}.
\label{eq:Uin-def}
\end{equation}
Regardless of whether $\mathbf{U}(\lambda;x,t)$ is analytic anywhere, by uniqueness of the solution of the initial-value problem solved in the proof of Proposition~\ref{p:Uin-def} this formula produces 
the desired object for all $\lambda$ with the exception of removable singularities at the aforementioned isolated points or curves.
\end{remark}

Recall from Lemma~\ref{lemma:radius} that if $\Delta\psi_0$ and $\Delta\psi_0'$ are absolutely integrable on $\mathbb{R}$, then $|a(\lambda)|>\tfrac{1}{2}$ holds for $|\lambda|\ge r[\Delta\psi]>1$ and $\Im\{\lambda\}\ge 0$.    Let $\Sigma_0$ denote the circle of radius $r=r[\Delta\psi]$ centered at the origin in the $\lambda$-plane, and let $D_0$ denote the open disk whose boundary is $\Sigma_0$.  Since $r>1$, the branch cut $\Sigma_\mathrm{c}$ is contained in $D_0$.  We also define the following related domains:
\begin{equation}
  D_{\pm}\defeq \lbrace \lambda\in\mathbb{C}\colon |\lambda|>r[\Delta\psi]~\text{and}~\Im\{\lambda\}\gtrless 0 \rbrace,
\end{equation}
and the following contours which divide the $\lambda$-plane into disjoint regions $D_0$, $D_+$, and $D_-$:
\begin{equation}
\begin{aligned}
  \Sigma_{+} &\defeq \lbrace \lambda\in \Sigma_0 \colon \Im\{\lambda\}\geq 0 \rbrace,\\
  \Sigma_{-} &\defeq \lbrace \lambda\in \Sigma_0 \colon \Im\{\lambda\}\leq 0 \rbrace,\\
  \Sigma_{\mathrm{L}} &\defeq \lbrace \lambda\in \mathbb{R} \colon \lambda\in(-\infty,-r[\Delta\psi]] \rbrace,\\
  \Sigma_{\mathrm{R}} &\defeq \lbrace \lambda\in \mathbb{R} \colon \lambda\in[r[\Delta\psi],+\infty) \rbrace.
\end{aligned}
\end{equation}
Note that $\Sigma_0=\Sigma_+ \cup \Sigma_-$. We orient both semicircular arcs $\Sigma_{\pm}$ in the direction from $\lambda=- r[\Delta\psi]$ to $\lambda= r[\Delta\psi]$, $\Sigma_{\mathrm{L}}$ from $\lambda=-\infty$ to $\lambda=- r[\Delta\psi]$, and $\Sigma_{\mathrm{R}}$ from $\lambda=r[\Delta\psi]$ to $\lambda=+\infty$. See Figure~\ref{fig:contours-ZBC}.   We denote by $\Sigma$ the union of all of these contours.
\begin{figure}[h]
\begin{tikzpicture}[scale=0.8]
\def\yL{5}
\def\xL{8}
\def\R{\yL-2}

\coordinate (o) at (0,0);
\coordinate (i) at (0,1);
\coordinate (mi) at (0,-1);
\coordinate (leftcircle) at (-\R,0);
\coordinate (yaxisend) at (0,-\yL);
\draw[help lines,opacity=0.65,<->] (-\xL,0) -- (\xL,0) coordinate (xaxis);
\draw[help lines,opacity=0.65,<->] (0,-\yL) -- (0,\yL) coordinate (yaxis);

\begin{scope}[very thick,decoration={markings,
mark=at position 0.3 with {\arrow[line width =2pt]{>}},
mark=at position 0.7 with {\arrow[line width =2pt]{>}}
}
]
\path[draw,PineGreen,line width =2.4pt,line cap=round , postaction=decorate]
(180:\R) arc(180:0:\R);
\path[draw,PineGreen,line width =2.4pt,line cap=round , postaction=decorate]
(180:\R) arc(180:360:\R);
\end{scope}
\path[fill,PineGreen, opacity=0.1] (\R,0) arc(0:360:\R);

\begin{scope}[very thick,decoration={
  	 markings,
	 mark=at position 0.5 with {\arrow[line width =2pt]{>}},
	 }
]
\path[draw,CadetBlue,line cap=round, line width =2.4pt,postaction=decorate] (\R,0)--(\xL-0.2,0);
\end{scope}

\begin{scope}[very thick,decoration={
  	 markings,
	 mark=at position 0.5 with {\arrow[line width =2pt]{>}}
	 }
]
\path[draw,CadetBlue,line cap=round,line width =2.4pt,postaction=decorate] (180:\xL-0.2)--(180:\R);
\end{scope}

\coordinate (sp) at (150:\R+0.1);
\coordinate (sptext) at (150:\R+2);
\node[yshift=6pt] at (sptext) {\color{PineGreen}$\Sigma_+$};
\draw [->,thin ] (sptext) -- (sp);

\coordinate (sm) at (210:\R+0.1);
\coordinate (smtext) at (210:\R+2);
\node[yshift=-6pt] at (smtext) {\color{PineGreen}$\Sigma_-$};
\draw [->,thin ] (smtext) -- (sm);

\node[below] at (xaxis) {$\Re \lambda$};
\node[above] at (yaxis) {$\Im \lambda$};
\node[above right] at (20:\R+2) {$D_{+}$};
\node[below right] at (-20:\R+2) {$D_{-}$};

\node[below] at (0:\R+3) {$\color{CadetBlue}\Sigma_{\mathrm{R}}$};
\node[below] at (180:\R+3) {$\color{CadetBlue}\Sigma_{\mathrm{L}}$};

\node[above right] at (75:\R-0.1) {$\color{PineGreen}\Sigma_0=\Sigma_{+}\cup\Sigma_{-}$};
\node[above] at (150:\R-1.5) {$D_0$};
\node[below left] at (270:\R)  {$-\I r$};
\node at (270:\R) {\color{PineGreen}\textbullet};
\node[above left] at (90:\R)  {$\I r$};
\node at (90:\R) {\color{PineGreen}\textbullet};
\node[below right] at (0:\R)  {$r$};
\node at (0:\R) {\color{PineGreen}\textbullet};
\node[below left] at (180:\R)  {$-r$};
\node at (180:\R) {\color{PineGreen}\textbullet};
\end{tikzpicture}
\caption{Definitions of the regions $D_0$, $D_{\pm}$, and the contours $\Sigma_0$, $\Sigma_{\mathrm{L}}$, and $\Sigma_\mathrm{R}$.}
\label{fig:contours-ZBC}
\end{figure}
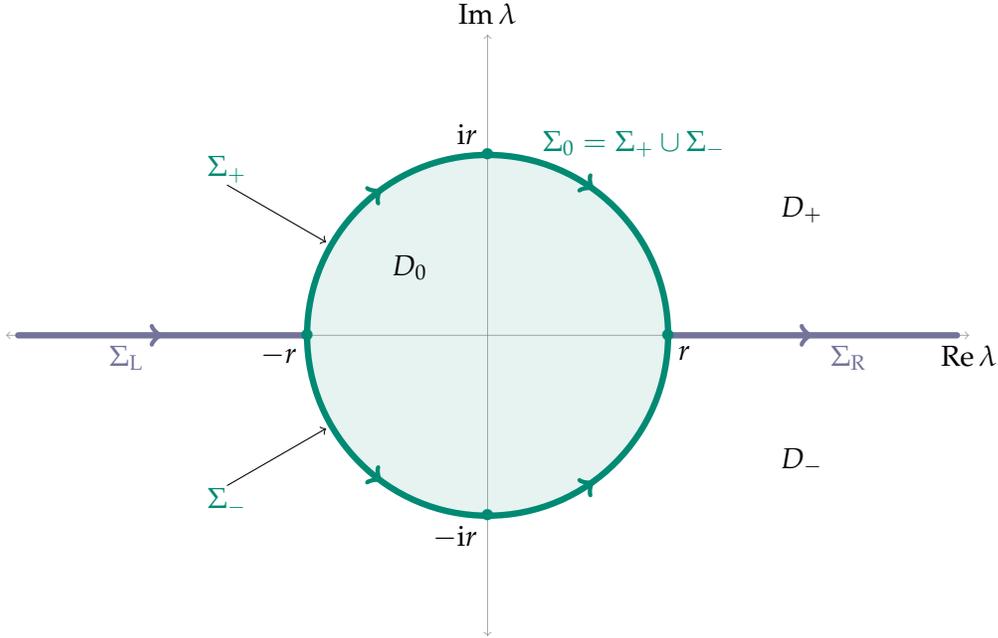

We now define for $\lambda\in\mathbb{C}\setminus\Sigma$ a simultaneous fundamental solution matrix for the Lax pair \eqref{eq:Lax-x}--\eqref{eq:Lax-t} as follows:
\begin{equation}
\mathbf{U}(\lambda;x,t):=\begin{cases}
\mathbf{U}^\mathrm{BC}(\lambda;x,t),&\quad\lambda\in D_+\cup D_-\\
\mathbf{U}^\mathrm{in}(\lambda;x,t),&\quad\lambda\in D_0,
\end{cases}
\label{eq:N-mat}
\end{equation}
where $\mathbf{U}^\mathrm{BC}(\lambda;x,t)$ is defined by \eqref{eq:U-mat-NZBC} and $\mathbf{U}^\mathrm{in}(\lambda;x,t)$ is defined in Proposition~\ref{p:Uin-def}.

Assuming only that the initial condition $\psi_0$ in \eqref{eq:IC} satisfies $\Delta\psi_0:=\psi_0-1\in L^1(\mathbb{R})$ and $\Delta\psi_0'\in L^1(\mathbb{R})$ and generates a classical solution $\psi(x,t)$ of the Cauchy problem \eqref{eq:NLS}--\eqref{eq:IC}, Lemma~\ref{lemma:radius} and the formula \eqref{eq:U-mat-NZBC} show that $\mathbf{U}(\lambda;x,t)$ is analytic in $D_+$ and $D_-$ and takes continuous boundary values from these domains on $\Sigma$.  Similarly, Proposition~\ref{p:Uin-def} guarantees that 
$\mathbf{U}(\lambda;x,t)$ is analytic in $D_0$ and takes continuous boundary values on $\Sigma_0=\partial D_0$.  Furthermore, $\det(\mathbf{U}(\lambda;x,t))=1$ holds for all $\lambda\in\mathbb{C}\setminus\Sigma$.  Combining \eqref{eq:M-mat-NZBC}--\eqref{eq:MBC-Schwarz} and \eqref{eq:U-in-sym}, we see that
\begin{equation}
\mathbf{U}(\lambda^*;x,t)=\sigma_2\mathbf{U}(\lambda;x,t)^*\sigma_2,\quad\lambda\in\mathbb{C}\setminus\Sigma,\quad (x,t)\in\mathbb{R}^2.
\label{eq:U-sym}
\end{equation}

Now define the related matrix $\mathbf{M}(\lambda;x,t)$ by
\begin{equation}
\mathbf{M}(\lambda;x,t) \defeq \mathbf{U}(\lambda;x,t)\E^{\I\rho(\lambda)(x+\lambda t)\sigma_3}\,,
\quad\lambda\in\mathbb{C}\setminus(\Sigma\cup\Sigma_\mathrm{c}),\quad (x,t)\in \mathbb{R}^2,
\label{eq:M-mat-NZBC-new}
\end{equation}
where $\mathbf{U}(\lambda;x,t)$ is defined by \eqref{eq:N-mat}.  
The exponential factor $\E^{\I\rho(\lambda)(x+\lambda t)\sigma_3}$ introduces a new jump discontinuity across $\Sigma_\mathrm{c}$, but otherwise $\mathbf{M}(\lambda;x,t)$ is analytic where it is defined.  Furthermore, from \eqref{eq:M-mat-NZBC} and \eqref{eq:MBC-norm} we see that
for each fixed $(x,t)\in\mathbb{R}^2$, $\mathbf{M}(\lambda;x,t)\to\mathbb{I}$ as $\lambda\to\infty$ in $D_\pm$, uniformly with respect to direction.  It is easy to see that $\mathbf{M}(\lambda;x,t)$ takes continuous boundary values on $\Sigma\cup\Sigma_\mathrm{c}$, \emph{including at the endpoints $\lambda=\pm\I$ of $\Sigma_\mathrm{c}$} (since $\rho(\pm\I)=0$ unambiguously and $\mathbf{U}^\mathrm{in}(\lambda;x,t)$ is analytic at $\lambda=\pm\I$).

We now compute the jump conditions relating these boundary values.  It will be convenient to oppositely orient the parts of $\Sigma_\mathrm{c}$ in the upper and lower half-planes, so we define
\begin{equation}
\Sigma_{\mathrm{c}}^{+} \defeq \{\lambda\in\Sigma_\mathrm{c}\colon \Im\{ \lambda \}\geq 0 \}\,,\quad \Sigma_{\mathrm{c}}^{-} \defeq \{\lambda\in\Sigma_\mathrm{c}\colon \Im\{ \lambda \}\leq 0 \}\,,
\end{equation}
and orient both contours $\Sigma_\mathrm{c}^{\pm}$ in the direction toward the origin $\lambda=0$.  This makes $\Sigma\cup\Sigma_\mathrm{c}$ a Schwarz-symmetric contour (taking orientation into account). See Figure~\ref{fig:contours-NZBC}.  We adopt the usual convention that $\mathbf{M}_+(\lambda;x,t)$ (respectively $\mathbf{M}_-(\lambda;x,t)$) denotes the boundary value taken at a regular (non-self-intersection) point of $\Sigma\cup\Sigma_\mathrm{c}$ from the left (respectively right) side by orientation.
\begin{figure}
\begin{tikzpicture}[scale=0.8]
\def\yL{5}
\def\xL{8}
\def\R{\yL-2}

\coordinate (o) at (0,0);
\coordinate (i) at (0,1);
\coordinate (mi) at (0,-1);
\coordinate (leftcircle) at (-\R,0);
\coordinate (yaxisend) at (0,-\yL);
\draw[help lines,opacity=0.65,<->] (-\xL,0) -- (\xL,0) coordinate (xaxis);
\draw[help lines,opacity=0.65,<->] (0,-\yL) -- (0,\yL) coordinate (yaxis);

\begin{scope}[very thick,decoration={markings,
mark=at position 0.3 with {\arrow[line width =2pt]{>}},
mark=at position 0.7 with {\arrow[line width =2pt]{>}}
}
]
\path[draw,PineGreen,line width =2.4pt,line cap=round , postaction=decorate]
(180:\R) arc(180:0:\R);
\path[draw,PineGreen,line width =2.4pt,line cap=round , postaction=decorate]
(180:\R) arc(180:360:\R);
\end{scope}
\path[fill,PineGreen, opacity=0.1] (\R,0) arc(0:360:\R);

\begin{scope}[very thick,decoration={
  	 markings,
	 mark=at position 0.5 with {\arrow[line width =2pt]{>}},
	 }
]
\path[draw,CadetBlue,line cap=round, line width =2.4pt,postaction=decorate] (\R,0)--(\xL-0.2,0);
\end{scope}

\begin{scope}[very thick,decoration={
  	 markings,
	 mark=at position 0.5 with {\arrow[line width =2pt]{>}}
	 }
]
\path[draw,CadetBlue,line cap=round,line width =2.4pt,postaction=decorate] (180:\xL-0.2)--(180:\R);
\end{scope}
\draw [->,BurntOrange,line width =2.4pt,decorate,line cap=round,decoration={snake, amplitude=.5mm, segment length=3.5mm,post length=2mm}]
(mi) -- (o);
\draw [->,BurntOrange,line width =2.4pt,decorate,line cap=round,decoration={snake, amplitude=.5mm, segment length=3.5mm,post length=2mm}]
(i) -- (o);

\foreach \Point in {(i), (mi)}{
  \node at \Point {\color{BurntOrange}\textbullet};
}

\coordinate (sp) at (150:\R+0.1);
\coordinate (sptext) at (150:\R+2);
\node[yshift=6pt] at (sptext) {\color{PineGreen}$\Sigma_+$};
\draw [->,thin ] (sptext) -- (sp);

\coordinate (sm) at (210:\R+0.1);
\coordinate (smtext) at (210:\R+2);
\node[yshift=-6pt] at (smtext) {\color{PineGreen}$\Sigma_-$};
\draw [->,thin ] (smtext) -- (sm);

\node[below] at (xaxis) {$\Re \lambda$};
\node[above] at (yaxis) {$\Im \lambda$};
\node[above right] at (20:\R+2) {$D_{+}$};
\node[below right] at (-20:\R+2) {$D_{-}$};

\node[below] at (0:\R+3) {$\color{CadetBlue}\Sigma_{\mathrm{R}}$};
\node[below] at (180:\R+3) {$\color{CadetBlue}\Sigma_{\mathrm{L}}$};

\node[left] at (i) {\small{$\I$}};
\node[left] at (mi) {\small{$-\I$}};
\node[above right] at (75:\R-0.1) {$\color{PineGreen}\Sigma_0=\Sigma_{+}\cup\Sigma_{-}$};
\node[above] at (150:\R-1.5) {$D_0$};
\node[right] at (0, 0.65) {\color{BurntOrange} $\Sigma_{\mathrm{c}}^+$};
\node[right] at (0, -0.65) {\color{BurntOrange} $\Sigma_{\mathrm{c}}^-$};
\node[below left] at (270:\R)  {$-\I r$};
\node at (270:\R) {\color{PineGreen}\textbullet};
\node[above left] at (90:\R)  {$\I r$};
\node at (90:\R) {\color{PineGreen}\textbullet};
\node[below right] at (0:\R)  {$r$};
\node at (0:\R) {\color{PineGreen}\textbullet};
\node[below left] at (180:\R)  {$-r$};
\node at (180:\R) {\color{PineGreen}\textbullet};
\end{tikzpicture}
\caption{Definitions of the regions $D_0$, $D_{\pm}$, and the contours $\Sigma_0$, $\Sigma_{\mathrm{L}}$, and $\Sigma_\mathrm{R}$. $\Sigma_\mathrm{c}=\Sigma_\mathrm{c}^{+}\cup \Sigma_\mathrm{c}^{-}$ is the branch cut of the function $\rho(\lambda)$.}
\label{fig:contours-NZBC}
\end{figure}
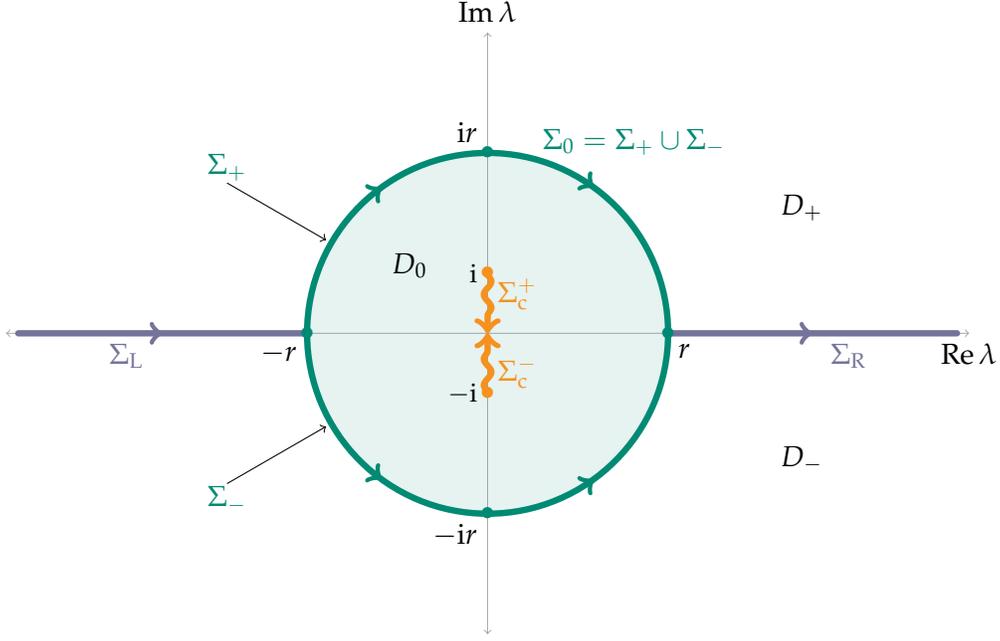
\begin{prop} The continuous boundary values taken by $\mathbf{M}(\lambda;x,t)$ satisfy the jump conditions
\begin{equation}
  \mathbf{M}_+(\lambda;x,t) = \begin{cases}\mathbf{M}_-(\lambda;x,t)\E^{-\I\rho(\lambda)(x+\lambda t)\sigma_3}\mathbf{V}(\lambda)\E^{\I\rho(\lambda)(x+\lambda t)\sigma_3}\,,&\lambda\in\Sigma\,,\\
  \mathbf{M}_-(\lambda;x,t)\E^{2\I\rho_+(\lambda)(x+\lambda t)\sigma_3}\,,&\lambda\in\Sigma_\mathrm{c}\,,
  \end{cases}
  \label{eq:M-mat-jump}
\end{equation}
  with
  \begin{equation}
    \mathbf{V}(\lambda)\defeq \begin{cases}
    \mathbf{V}^+(\lambda)\,,&\lambda\in\Sigma_+\,,\\
    \mathbf{V}^-(\lambda)\,,&\lambda\in\Sigma_-\,,\\
    \mathbf{V}^\mathbb{R}(\lambda)\,,&\lambda\in\Sigma_\mathrm{L}\cup\Sigma_\mathrm{R}\,,
    \end{cases}
    \label{eq:V-mat}
  \end{equation}
  where
  \begin{align}
        \mathbf{V}^+(\lambda)&\defeq\begin{bmatrix} a(\lambda)^{-1}\mathbf{j}^{-,1}(\lambda;L,0); &\mathbf{j}^{+,2}(\lambda;L,0)\end{bmatrix}
        \label{eq:Vplus}\,,\\
        \mathbf{V}^-(\lambda)&\defeq\begin{bmatrix} \mathbf{j}^{+,1}(\lambda;L,0); &\bar{a}(\lambda)^{-1}\mathbf{j}^{-,2}(\lambda;L,0)\end{bmatrix}^{-1}\label{eq:Vminus}\,,
  \end{align}
  \label{p:jump-NZBC}
  and $\mathbf{V}^\mathbb{R}(\lambda)$ is defined by \eqref{eq:VR} in terms of the reflection coefficient associated with the initial data $\psi_0$.
\end{prop}
The main part of the proof is to establish the jump conditions across $\Sigma_\pm$, and for this purpose it is useful to think of $\mathbf{U}^\mathrm{in}(\lambda;x,t)$ as a ``transfer matrix'' as illustrated in Figure~\ref{fig:jump-condition}.  The sort of reasoning illustrated in this figure is similar to that used to study mixed initial-boundary value problems using the unified transform method; see \cite{FokasUnifiedArticle,FokasUnified}.
\begin{figure}[h]
\begin{tikzpicture}
\coordinate (o) at (0,0);
\def\xtlen{1.4}
\def\xtlenSR{1.48}
\def\xtlenSL{1.32}
\coordinate (xt) at (\xtlen,\xtlen);
\coordinate (xtright) at (\xtlenSR,\xtlen);
\coordinate (xtleft) at (\xtlenSL,\xtlen);
\def\yL{2}
\def\xL{5.8}
\def\R{\yL-2}

\draw[help lines,opacity=0.65,<->] (-\xL,0) -- (\xL,0) coordinate (xaxis);
\draw[help lines,opacity=0.65,->] (0,0) -- (0,\yL) coordinate (yaxis);
\node[right] at (xaxis) {\small $x$};
\node[above] at (yaxis) {\small $t$};
\node[below] at (o) {$(L,0)$};
\node at (xt) {\color{PineGreen}\textbullet};
\node at (o) {\color{PineGreen}\textbullet};
\node[above] at (xt) {$(x,t)$};
\begin{scope}[very thick,decoration={
  	 markings,
	 mark=at position 0.5 with {\arrow[line width =1.5pt]{>}},
	 mark=at position 1 with {\arrow[line width =1.5pt]{>}}
	 }
]
\path[draw,DarkOrchid, line cap=round, line width =2pt, postaction=decorate] (\xL,\xtlen)--(xtright);
\path[draw,DarkOrchid, line cap=round, line width =2pt, postaction=decorate] (\xL,0)--(0.08,0);
\path[draw,BurntOrange, line cap=round, line width =2pt, postaction=decorate] (-\xL,\xtlen)--(xtleft);
\path[draw,BurntOrange, line cap=round, line width =2pt, postaction=decorate] (-\xL,0)--(-0.08,0);
\end{scope}
\begin{scope}[very thick,decoration={
  	 markings,
	 mark=at position 0.5 with {\arrow[line width =1.5pt]{>}},
	 }
]
\path[draw,PineGreen, line cap=round, line width =2pt, postaction=decorate] (o)--(xt);
\end{scope}

\node[above] at (4,\xtlen) {$\color{DarkOrchid}\mathbf{j}^{+,2}(\lambda;x,t)\E^{\I\rho(\lambda)\lambda t\sigma_3}$};
\node[below] at (2.9,0) {$\color{DarkOrchid}\mathbf{j}^{+,2}(\lambda;L,0)$};
\node[above] at (-2,\xtlen) {$\color{BurntOrange}\mathbf{j}^{-,1}(\lambda;x,t)\E^{-\I\rho(\lambda)\lambda t\sigma_3}$};
\node[below] at (-2.9,0) {$\color{BurntOrange}\mathbf{j}^{-,1}(\lambda;L,0)$};
  \node[right,xshift=2] at (0.7,0.7) {${\color{PineGreen}\mathbf{U}^{\mathrm{in}}(\lambda;x,t)}$};
\end{tikzpicture}
\caption{The computation of the jump condition for $\lambda\in\Sigma_+$.  Thus $\mathbf{j}^{-,1}(\lambda;x,t)\E^{-\I\rho(\lambda)\lambda t\sigma_3}=\mathbf{U}^\mathrm{in}(\lambda;x,t)\mathbf{j}^{-,1}(\lambda;L,0)$ and $\mathbf{j}^{+,2}(\lambda;x,t)\E^{\I\rho(\lambda)\lambda t\sigma_3}=\mathbf{U}^\mathrm{in}(\lambda;x,t)\mathbf{j}^{+,2}(\lambda;L,0)$.}
\label{fig:jump-condition}
\end{figure}
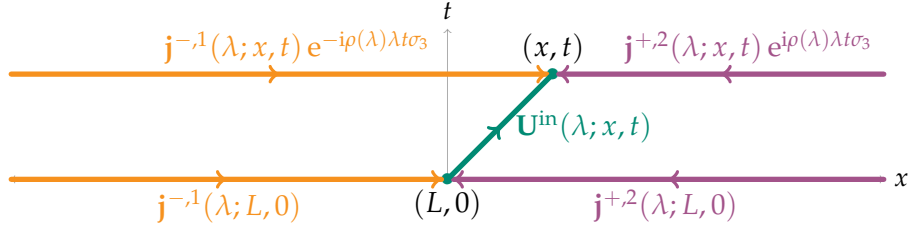

\begin{proof}
Since $\mathbf{M}(\lambda;x,t)=\mathbf{M}^{\mathrm{BC}}(\lambda;x,t)$ for $ \lambda \in D_+\cup D_-$, where the Beals-Coifman matrix function is defined in \eqref{eq:M-mat-NZBC}, the jump condition on $\Sigma_\mathrm{L}\cup\Sigma_\mathrm{R}\subset \mathbb{R}$ is just a special case of the formula \eqref{eq:traditional-jump-real}--\eqref{eq:VR}.

Next, fix $\lambda\in\Sigma_+$.  Observe that since $\mathbf{j}^{-,1}(\lambda;x,t)\E^{-\I\rho(\lambda)\lambda t}$ is a simultaneous solution of the Lax pair \eqref{eq:Lax-x}--\eqref{eq:Lax-t} and since $\mathbf{U}^{\mathrm{in}}(\lambda;x,t)$ is a fundamental solution matrix for the same system, there exists a vector $\mathbf{d}^{+,1}(\lambda)$ (independent of $(x,t)\in\mathbb{R}^2$) such that the first column of $\mathbf{U}^\mathrm{BC}(\lambda;x,t)$ satisfies
  \begin{equation}
    a(\lambda)^{-1}\mathbf{j}^{-,1}(\lambda;x,t)\E^{-\I\rho(\lambda)\lambda t} = \mathbf{U}^{\mathrm{in}}(\lambda;x,t)\mathbf{d}^{+,1}(\lambda)\,,\quad \lambda\in \Sigma_+\,.
    \label{eq:dplus-1}
  \end{equation}
Taking $(x,t)=(L,0)$, the identity $\mathbf{U}^{\mathrm{in}}(\lambda;L,0)=\mathbb{I}$ then determines $\mathbf{d}^{+,1}(\lambda)=a(\lambda)^{-1}\mathbf{j}^{-,1}(\lambda;L,0)$. Similarly, since $\mathbf{j}^{+,2}(\lambda;x,t)\E^{\I\rho(\lambda)\lambda t}$ is a simultaneous solution of the system \eqref{eq:Lax-x}--\eqref{eq:Lax-t}, there exists a vector $\mathbf{d}^{+,2}(\lambda)$ such that the second column of $\mathbf{U}^\mathrm{BC}(\lambda;x,t)$ satisfies
  \begin{equation}
    \mathbf{j}^{+,2}(\lambda;x,t)\E^{\I\rho(\lambda)\lambda t} = \mathbf{U}^{\mathrm{in}}(\lambda;x,t)\mathbf{d}^{+,2}(\lambda)\,,\quad \lambda\in \Sigma_+\,,
    \label{eq:dplus-2}
  \end{equation}
  and again taking $(x,t)=(L,0)$ we get $\mathbf{d}^{+,2}(\lambda)=\mathbf{j}^{+,2}(\lambda;L,0)$. Then \eqref{eq:dplus-1} and \eqref{eq:dplus-2} imply that for $\lambda\in\Sigma_+$, by the clockwise orientation of $\Sigma_+$ we have
  \begin{equation}
    \mathbf{U}_+(\lambda;x,t)=\mathbf{U}_-(\lambda;x,t)\begin{bmatrix}\mathbf{d}^{+,1}(\lambda)\,,& \mathbf{d}^{+,2}(\lambda)\end{bmatrix}=\mathbf{U}_-(\lambda;x,t)\begin{bmatrix}a(\lambda)^{-1}\mathbf{j}^{-,1}(\lambda;L,0)\,, & \mathbf{j}^{+,2}(\lambda;L,0) \end{bmatrix},
  \end{equation}
  because $\mathbf{U}_+(\lambda;x,t)=\mathbf{U}^\mathrm{in}(\lambda;x,t)$ and $\mathbf{U}_-(\lambda;x,t)=\mathbf{U}^\mathrm{BC}(\lambda;x,t)$ for $\lambda\in\Sigma_+$.
  Because $\rho(\lambda)$ is analytic on $\Sigma_0$, this proves that the jump condition satisfied by $\mathbf{M}(\lambda;x,t)$ on $\Sigma_+$ is
  \begin{equation}
  \mathbf{M}_+(\lambda;x,t)=\mathbf{M}_-(\lambda;x,t) \E^{-\I\rho(\lambda)(x+\lambda t)\sigma_3}\mathbf{V}^+(\lambda)\E^{\I\rho(\lambda)(x+\lambda t)\sigma_3}\,,\quad \lambda\in\Sigma_+\,,
  \end{equation}
  where $\mathbf{V}^+(\lambda)$ is defined by \eqref{eq:Vplus}.
A completely analogous calculation shows that
\begin{equation}
\mathbf{M}_+(\lambda;x,t)=\mathbf{M}_-(\lambda;x,t) \E^{-\I\rho(\lambda)(x+\lambda t)\sigma_3}\mathbf{V}^-(\lambda)\E^{\I\rho(\lambda)(x+\lambda t)\sigma_3}\,,\quad \lambda\in\Sigma_-\,,
\end{equation}
where $\mathbf{V}^-(\lambda)$ is defined by \eqref{eq:Vminus}.

Finally, to prove the jump condition of $\mathbf{M}(\lambda;x,t)$ across $\Sigma_\mathrm{c}=\Sigma_\mathrm{c}^+ \cup\Sigma_\mathrm{c}^-$, observe that $\Sigma_\mathrm{c}$ is contained in $D_0$, where $\mathbf{U}^{\mathrm{in}}(\lambda;x,t)$ is single-valued and analytic, and that $\rho_+(\lambda)=-\rho_-(\lambda)$ for $\lambda\in\Sigma_\mathrm{c}$. Thus,
\begin{equation}
\begin{aligned}
\mathbf{M}_+(\lambda;x,t)=\mathbf{U^{\mathrm{in}}}(\lambda;x,t)\E^{\I\rho_+(\lambda)(x+\lambda t)\sigma_3} &=\mathbf{U^{\mathrm{in}}}(\lambda;x,t)\E^{\I\rho_-(\lambda)(x+\lambda t)\sigma_3}\left[\E^{\I\left[\rho_+(\lambda) - \rho_-(\lambda)\right](x+\lambda t)\sigma_3}\right]\\
&=\mathbf{M}_-(\lambda;x,t)\E^{2\I\rho_+(\lambda;x,t)\sigma_3}\,, \quad\lambda\in\Sigma_\mathrm{c}^\pm,
\end{aligned}
\end{equation}
and this completes the proof.
\end{proof}
Note that $a(\lambda)$ and $\bar{a}(\lambda)$ appearing in the jump matrices \eqref{eq:Vplus}--\eqref{eq:Vminus} respectively are defined by \eqref{eq:a-b-NZBC}, and as the relevant Wronskian determinants are independent of the variables $(x,t)$ we may choose $(x,t)=(L,0)$:
\begin{equation}
a(\lambda) = \det \left(\begin{bmatrix} \mathbf{j}^{-,1}(\lambda;L,0); & \mathbf{j}^{+,2}(\lambda;L,0)\end{bmatrix}\right)\,,\quad
        \bar{a}(\lambda) = \det\left(\begin{bmatrix} \mathbf{j}^{+,1}(\lambda;L,0); & \mathbf{j}^{+,2}(\lambda;L,0)\end{bmatrix}\right)\,,
        \label{eq:a-at-L}
\end{equation}
and thus it is obvious that $\det(\mathbf{V}^\pm(\lambda))=1$.

One can now see that the matrix $\mathbf{M}(\lambda;x,t)$ defined in \eqref{eq:M-mat-NZBC-new} satisfies the conditions of the following Riemann-Hilbert problem, the data for which only involves the initial condition $\psi_0(x)$.  Even more significantly, \emph{we require no condition on $\psi_0$ to guarantee any particular rate of growth of $\mathbf{M}^\mathrm{BC}(\lambda;x,t)$ as $\lambda\to\pm\I$}.  Indeed, there is no need for any growth condition like \eqref{eq:growth-condition} to make the following problem well-posed.
\begin{rhp}
  Seek a $2\times 2$ matrix function $\mathbf{M}(\lambda;x,t)$ that has the following properties:
  \begin{itemize}
    \item \textbf{Analyticity:} $\mathbf{M}(\lambda;x,t)$ is analytic for $\lambda\in D_+ \cup D_-$ and for $\lambda\in D_0 \setminus\Sigma_\mathrm{c}$\,.
    \item \textbf{Jump Condition:} $\mathbf{M}(\lambda;x,t)$ takes continuous boundary values $\mathbf{M}_\pm(\lambda;x,t)$ on $\Sigma\cup\Sigma_\mathrm{c}$, and they are related by the jump conditions described in Proposition~\ref{p:jump-NZBC}.
    \item \textbf{Normalization:} $\lim_{\lambda\to\infty} \mathbf{M}(\lambda;x,t) = \mathbb{I}$.
  \end{itemize}
  \label{rhp:M-NZBC}
\end{rhp}
This Riemann-Hilbert problem has many of the features that are expected and desirable in the setting of inverse scattering.  Indeed, 
the ``core'' jump matrix $\mathbf{V}(\lambda)$ in \eqref{eq:V-mat} depends solely on scattering data (well-defined Jost solutions) associated with the initial data $\psi_0$ for the Cauchy initial-value problem \eqref{eq:NLS}--\eqref{eq:IC}, and all $(x,t)$-dependence in the problem enters via conjugation of the core jump matrix by elementary diagonal exponential factors.  Furthermore, we can show that this problem has a unique solution for all $(x,t)\in\mathbb{R}^2$.         

\begin{theorem}
\rhref{rhp:M-NZBC} has a unique solution for all $(x,t)\in\mathbb{R}^2$.
  \label{t:rhp-N-NZBC-unique}
\end{theorem}
\begin{proof}
  We will show that the jump conditions and the jump matrices in \rhref{rhp:M-NZBC} satisfy the hypotheses of Zhou's \emph{Vanishing Lemma} \cite[Theorem 9.3]{Zhou1989e}. Note that the jump contour $\Sigma\cup\Sigma_\mathrm{c}$ has the necessary invariance under Schwarz reflection with orientation. 
We first show that for all $\lambda\in(\Sigma\cup\Sigma_\mathrm{c})\setminus\mathbb{R}$, the jump matrix $\mathbf{V}(\lambda;x,t):=\mathbf{M}_-(\lambda;x,t)^{-1}\mathbf{M}_+(\lambda;x,t)$ satisfies $\mathbf{V}(\lambda^*;x,t)=\mathbf{V}(\lambda;x,t)^\dagger$.
Note that $\mathbf{V}^+(\lambda)=\mathbf{U}^\mathrm{BC}(\lambda;L,0)$ for $\lambda\in \Sigma_+$ while $\mathbf{V}^-(\lambda)=\mathbf{U}^\mathrm{BC}(\lambda;L,0)^{-1}$ for $\lambda\in\Sigma_-$.  Suppose that $\lambda\in\Sigma_+$. Since the Schwarz symmetry \eqref{eq:MBC-Schwarz} implies that $\mathbf{U}^\mathrm{BC}(\lambda^*;x,t)=\sigma_2\mathbf{U}^\mathrm{BC}(\lambda;x,t)^*\sigma_2$ which can be re-written as $\mathbf{U}^\mathrm{BC}(\lambda^*;x,t)^\dagger = \mathbf{U}^\mathrm{BC}(\lambda;x,t)^{-1}$ because $\det(\mathbf{U}^\mathrm{BC}(\lambda;x,t))=1$, 
  \begin{equation}
    \begin{aligned}
    \mathbf{V}(\lambda^*;x,t)&=\E^{-\I\rho(\lambda^*)(x+\lambda^*t)\sigma_3}\mathbf{V}^-(\lambda^*)\E^{\I\rho(\lambda^*)(x+\lambda^*t)\sigma_3}\\
    &=\E^{-\I\rho(\lambda)^*(x+\lambda t)^*\sigma_3}
    \mathbf{U}^\mathrm{BC}(\lambda^*;x,t)^{-1}
    \E^{\I\rho(\lambda)^*(x+\lambda t)^*\sigma_3}\\
    &=\left[\E^{\I\rho(\lambda)(x+\lambda t)\sigma_3}\right]^\dagger
    \mathbf{U}^\mathrm{BC}(\lambda;x,t)^\dagger
    \left[\E^{-\I\rho(\lambda)(x+\lambda t)\sigma_3}\right]^\dagger\\
   & =\left[\E^{-\I\rho(\lambda)(x+\lambda t)\sigma_3} \mathbf{V}^+(\lambda)\E^{\I\rho(\lambda)(x+\lambda t)\sigma_3}\right]^\dagger=\mathbf{V}(\lambda;x,t)^\dagger.
  \end{aligned}
  \end{equation}
Supposing next that $\lambda\in\Sigma^+_\mathrm{c}$, we have
  \begin{equation}
      \begin{aligned}
      \mathbf{V}(\lambda^*;x,t) &= \E^{2\I\rho_+(\lambda^*)(x+\lambda^* t)\sigma_3} \\
      &=\E^{2\I\rho_-(\lambda)^*(x+\lambda^* t)\sigma_3}=\left[\E^{2\I\rho_+(\lambda)(x+\lambda t)\sigma_3}\right]^\dagger=\mathbf{V}_{\mathrm{c}}(\lambda;x,t)^\dagger\,.
  \end{aligned}
  \end{equation}

It remains to show that if $\lambda\in\Sigma_\mathrm{L}\cup\Sigma_\mathrm{R}$ (i.e., the part of the jump contour on the real axis), then $\mathbf{V}(\lambda;x,t)+\mathbf{V}(\lambda;x,t)^\dagger$ is positive definite.  
For such $\lambda$, $\rho(\lambda)$ is real-valued and $\mathbf{V}^\mathbb{R}(\lambda)$ defined by \eqref{eq:VR} is hermitian, so we have
  \begin{equation}
   \mathbf{V}(\lambda;x,t)+\mathbf{V}(\lambda;x,t)^\dagger= 2\E^{-\I\rho(\lambda)(x+\lambda t)\sigma_3}\mathbf{V}^{\mathbb{R}}(\lambda)\E^{\I\rho(\lambda)(x+\lambda t)\sigma_3}=\mathbf{H}(\lambda;x,t)^\dagger \mathbf{H}(\lambda;x,t),
  \end{equation}
  where
  \begin{equation}
    \mathbf{H}(\lambda;x,t)\defeq \sqrt{2}\begin{bmatrix}1& 0 \\ R(\lambda)\E^{\I\rho(\lambda)(x+\lambda t)} & 1 \end{bmatrix}
  \end{equation}
  is an invertible matrix ($\det(\mathbf{H}(\lambda))=2$). Since every matrix of the form $\mathbf{H}^\dagger\mathbf{H}$ with $\mathbf{H}$ invertible is positive definite, 
 the desired conclusion follows. 
With the normalization condition $\mathbf{M}(\lambda;x,t)\to\mathbb{I}$ as $\lambda\to\infty$, we have confirmed all the hypotheses of Zhou's Vanishing Lemma. Consequently, \rhref{rhp:M-NZBC} is uniquely solvable.
\end{proof}
From the solution $\mathbf{M}(\lambda;x,t)$ of \rhref{rhp:M-NZBC} we recover the solution $\psi(x,t)$ of the Cauchy problem \eqref{eq:NLS}--\eqref{eq:IC} via the limit
\begin{equation}
  \psi(x,t)=2\I \lim_{\lambda\to\infty}\lambda M_{12}(\lambda;x,t)=2\I\lim_{\lambda\to\infty} \lambda M^{\mathrm{BC}}_{12}(\lambda;x,t).
  \label{eq:NLS-recovery}
\end{equation}
A reasonable question to ask is:  what happened to all of the poles and spectral singularities that can be present for reasonable initial data in the traditional IST outlined in Section~\ref{s:IST-NZBC}?  The answer is that this information is now encoded instead in the in the jump matrices $\mathbf{V}^\pm(\lambda)$ supported on the circle $\Sigma_0$.  These jump matrices are locally analytic on $\Sigma_0$ and upon carrying out analytic continuation into $D_0$ all of these singularities will re-emerge.  We wish to emphasize that in order to compute the jump matrices $\mathbf{V}^\pm(\lambda)$ no analytic continuation away from the continuous spectrum $\Gamma$ is necessary, since the Volterra integral equations characterizing the Jost solutions appearing in \eqref{eq:Vplus}--\eqref{eq:Vminus} can be solved by iteration for fixed $\lambda$ on the corresponding jump contours.  Finally, for the purposes of the robust IST it is sufficient to assume that $\Delta\psi:=\psi-1$ and $\Delta\psi'$ lie in $L^1(\mathbb{R})$; there is no need for any stronger decay condition such as $(1+x^2)\Delta\psi(x)\in L^1(\mathbb{R})$.  Note that the Peregrine solution $\psi_\mathrm{P}$ given by \eqref{eq:Peregrine} and its higher-order generalizations satisfy $\Delta\psi,\Delta\psi'\in L^1(\mathbb{R})$, but \emph{not} $(1+x^2)\Delta\psi(x)\in L^1(\mathbb{R})$.  

\begin{remark}
\label{r:rhp-generalize}
The proof of Theorem~\ref{t:rhp-N-NZBC-unique} relies on only basic properties of the jump matrix $\mathbf{V}(\lambda;x,t)=\E^{-\I\rho(\lambda)(x+\lambda t)\sigma_3}\mathbf{V}(\lambda)\E^{\I\rho(\lambda)(x+\lambda t)\sigma_3}$ defined for $\lambda\in\Sigma$, namely that $\mathbf{V}(\lambda^*)=\mathbf{V}(\lambda)^\dagger$ holds for $\lambda\in\Sigma_0\subset\Sigma$ and that $\mathbf{V}(\lambda)+\mathbf{V}(\lambda)^\dagger$ is positive definite for $\lambda\in\Sigma_\mathrm{L}\cup\Sigma_\mathrm{R}\subset\Sigma$.  Hence, whenever $\mathbf{V}(\lambda)$ has such properties (in additional to standard technical properties characterizing a function space for $\mathbf{V}$ and the sense in which $\mathbf{V}\to\mathbb{I}$ as $\lambda\to\infty$ from within $\Sigma_\mathrm{L}\cup\Sigma_\mathrm{R}\subset\Sigma$), the conclusion of Theorem~\ref{t:rhp-N-NZBC-unique} holds.  Moreover, a standard dressing argument based on the the fact that the matrix $\mathbf{U}(\lambda;x,t)=\mathbf{M}(\lambda;x,t)\E^{-\I\rho(\lambda)(x+\lambda t)\sigma_3}$ satisfies Riemann-Hilbert conditions independent of $(x,t)\in\mathbb{R}^2$ shows that the function $\psi(x,t)$ extracted from the unique solution of \rhref{rhp:M-NZBC} via \eqref{eq:NLS-recovery} is a global classical solution of the focusing NLS equation \eqref{eq:NLS}.  Significantly, this conclusion holds even in the generalized setting in which $\mathbf{V}$ need not originate via the (robust) direct scattering transform for any Cauchy data \eqref{eq:IC}. 
\end{remark}

\begin{remark}
  Note that $L\in\mathbb{R}$ appears as an auxiliary parameter in the construction of $\mathbf{U}^{\mathrm{in}}(\lambda;x,t)$ by Proposition~\ref{p:Uin-def} (and hence in $\mathbf{M}(\lambda;x,t)$). The  
  parameter $L$ then appears in the ``core'' jump matrices $\mathbf{V}^\pm(\lambda)$ on $\Sigma^{\pm}$ in \rhref{rhp:M-NZBC}. Now suppose that $\mathbf{U}^{\mathrm{in},j}(\lambda;x,t)$ is matrix of Proposition~\ref{p:Uin-def} normalized at $(x,t)=(L_j,0)$ for $j=1,2$.
  Then
  \begin{equation}
    \mathbf{U}^{\mathrm{in},2}(\lambda;x,t)=\mathbf{U}^{\mathrm{in},1}(\lambda;x,t)\mathbf{U}^{\mathrm{in},1}(\lambda;L_2,0)^{-1}
  \end{equation}
since $\mathbf{U}^{\mathrm{in},1}(\lambda;x,t)$ is a matrix of fundamental simultaneous solutions of \eqref{eq:Lax-x}--\eqref{eq:Lax-t}. Thus, changing the value of $L$ manifests itself through the multiplication of $\mathbf{U}^{\mathrm{in}}(\lambda;x,t)$ on the right by a matrix that is constant in $(x,t)$ and analytic in $\lambda$. This constant matrix can also be appropriately applied to the jump condition on the circle $\Sigma_0$, leaving $\mathbf{U}^{\mathrm{BC}}(\lambda;x,t)$ unchanged outside $\Sigma_0$. Doing so, for different values of $L$, we leave the residue of $\mathbf{M}(\lambda;x,t)$ at $\lambda=\infty$ (and hence the solution of the NLS equation recovered via \eqref{eq:NLS-recovery}) unchanged. As in \cite{Deift1991b,Zhou1989d}, the IST method we have here provides a correspondence between a suitable solution of the NLS equation \eqref{eq:NLS} and an equivalence class of Riemann-Hilbert problems, or an equivalence class of jump matrices supported on the circle $\Sigma_0$ augmented by the jump matrix on the real line.
  \label{r:equivalence-class-rhp}
\end{remark}

\begin{remark}
  The basic method described here for the case of nonzero boundary conditions \eqref{eq:NZBC} can also be implemented in the case of zero boundary conditions \eqref{eq:ZBC} simply by replacing $\rho(\lambda)$ with $\lambda$, and hence $\mathbf{E}(\lambda)$ by the identity matrix $\mathbb{I}$.  Furthermore, the branch cut $\Sigma_\mathrm{c}$ is not present in the latter case.  Thus the robust IST provides an alternative approach in that case to the ``cut-off'' method advanced in \cite{Deift1991b,Zhou1989d}, and the robust IST avoids the jump on the real axis in $D_0$ that is generally present in the ``cut-off'' method.  The basic approach used to develop the robust IST described in this section can also be applied to other integrable problems like the derivative nonlinear Schr\"odinger equation with minimal modifications.  
\end{remark}


\section{Darboux Transformations in the Context of the Robust IST}
\label{s:Darboux}
In this section, we develop a Darboux transformation scheme for \rhref{rhp:M-NZBC} and apply it to generate arbitrary-order rogue waves in such a way that they directly admit a Riemann-Hilbert representation with no limit process required.  A good general reference for Darboux transformations in integrable systems is the book of Matveev and Salle \cite{MatveevS1991-book}.

\subsection{Basic pole insertion}
\label{s:pole-insertion}
Let $\mathbf{M}(\lambda;x,t)$ be the matrix function characterized by the conditions of \rhref{rhp:M-NZBC}, and let $\mathbf{U}(\lambda;x,t):=\mathbf{M}(\lambda;x,t)\E^{-\I\rho(\lambda)(x+\lambda t)\sigma_3}$.  Choose any point $\xi\in D_0$.  For a two-nilpotent matrix $\mathbf{R}(x,t)$ to be determined, consider the gauge transformation
\begin{equation}
  \dot{\mathbf{U}}(\lambda;x,t)\defeq \left(\mathbb{I}+ \frac{\mathbf{R}(x,t)}{\lambda-\xi}\right)\mathbf{U}(\lambda;x,t),\quad\mathbf{R}(x,t)^2 = \mathbf{0}.
  \label{eq:N-tilde}
\end{equation}
This gauge transformation has the following effects.
\begin{itemize}
\item
Because $\mathbf{R}(x,t)$ is 2-nilpotent, $\det(\mathbf{U}(\lambda;x,t))=1$ implies that also $\det(\dot{\mathbf{U}}(\lambda;x,t))=1$.
\item
The normalization condition $\mathbf{M}(\lambda;x,t)\to\mathbb{I}$ as $\lambda\to\infty$ implies that also $\dot{\mathbf{M}}(\lambda;x,t):=\dot{\mathbf{U}}(\lambda;x,t)\E^{\I\rho(\lambda)(x+\lambda t)\sigma_3}\to\mathbb{I}$ as $\lambda\to\infty$.
\item
At each non-self-intersection point of $\Sigma$, the jump condition $\mathbf{U}_+(\lambda;x,t)=\mathbf{U}_-(\lambda;x,t)\mathbf{V}(\lambda)$ implies that also $\dot{\mathbf{U}}_+(\lambda;x,t)=\dot{\mathbf{U}}_-(\lambda;x,t)\mathbf{V}(\lambda)$.
\end{itemize}
For the gauge transformation to be nontrivial, we assume that although $\mathbf{R}(x,t)^2=\mathbf{0}$, $\mathbf{R}(x,t)\neq \mathbf{0}$, in which case $\mathbf{R}(x,t)$ necessarily has the form
\begin{equation}
\mathbf{R}(x,t)=\mathbf{v}(x,t)\mathbf{v}(x,t)^\top\sigma_2
\label{eq:Rmat-def}
\end{equation}
for a vector-valued function $\mathbf{v}(x,t)$ not identically zero.  The vector $\mathbf{v}(x,t)$ is then to be determined so that
\begin{equation}
\underset{\lambda=\xi}{\res}\,\dot{\mathbf{U}}(\lambda;x,t)=\lim_{\lambda\to\xi}\dot{\mathbf{U}}(\lambda;x,t)\mathbf{C},
\label{eq:res1}
\end{equation}
where $\mathbf{C}\neq\mathbf{0}$ is a $2\times 2$ complex-valued constant two-nilpotent matrix: $\mathbf{C}^2 =\mathbf{0}$.  By analogy with \eqref{eq:Rmat-def}, we may write $\mathbf{C}$ in the form
\begin{equation}
  \mathbf{C}=\mathbf{c}\mathbf{c}^\top\sigma_2,\quad\text{where}~\mathbf{c}= \big[\begin{matrix}c_1 & c_2 \end{matrix}\big]^\top.
\label{eq:C-nilpotent}
\end{equation}
Here $c_1$ and $c_2$ are complex parameters, not both zero.

To see how the condition \eqref{eq:res1} determines $\mathbf{v}(x,t)$ (and hence $\dot{\mathbf{U}}(\lambda;x,t)$) in terms of the parameters $\xi\in D_0$ and $\mathbf{c}\in\mathbb{C}^2\setminus\{\mathbf{0}\}$, one starts from the Laurent expansion about the pole $\lambda=\xi$ of $\dot{\mathbf{U}}(\lambda;x,t)$ defined by \eqref{eq:N-tilde}:
\begin{equation}
\dot{\mathbf{U}}(\lambda;x,t)=\mathbf{R}(x,t)\mathbf{U}(\xi;x,t)(\lambda-\xi)^{-1} + \mathbf{U}(\xi;x,t)+\mathbf{R}(x,t)\mathbf{U}'(\xi;x,t) + O(\lambda-\xi),\quad\lambda\to\xi,
\label{eq:laurent1}
\end{equation}
where the prime denotes differentiation with respect to $\lambda$. On this expansion we impose the condition \eqref{eq:res1}; the existence of the limit on the right-hand side requires that 
\begin{equation}
  \mathbf{R}(x,t)\mathbf{U}(\xi;x,t)\mathbf{C}=\mathbf{0},
  \label{eq:R}
\end{equation}
and then matching the residue on the left-hand side to the limit on the right,
\begin{equation}
  \mathbf{R}(x,t)\mathbf{U}(\xi;x,t)=\left(\mathbf{U}(\xi;x,t) + \mathbf{R}(x,t)\mathbf{U}'(\xi;x,t)\right)\mathbf{C}.
  \label{eq:RC}
\end{equation}
Using the representation \eqref{eq:C-nilpotent} in \eqref{eq:R}, we see that (since $\mathbf{c}^\top\sigma_2\neq\mathbf{0}$), the vector $\mathbf{U}(\xi;x,t)\mathbf{c}$ must lie in the kernel of $\mathbf{R}(x,t)$.  Given the form \eqref{eq:Rmat-def} of $\mathbf{R}(x,t)$, the vector $\mathbf{v}(x,t)$ must be proportional to $\mathbf{U}(\xi;x,t)\mathbf{c}$, so \eqref{eq:Rmat-def} can be rewritten in the form
\begin{equation}
  \mathbf{R}(x,t)=\varphi(x,t)\mathbf{U}(\xi;x,t)\mathbf{c}\mathbf{c}^\top\mathbf{U}(\xi;x,t)^\top\sigma_2,
  \label{eq:R3}
\end{equation}
where the only ambiguity remaining is $\varphi(x,t)$, a nonzero complex scalar.  To determine $\varphi(x,t)$, we substitute \eqref{eq:R3} into \eqref{eq:RC}; using \eqref{eq:C-nilpotent} and the fact that for any $2\times 2$ matrix $\mathbf{A}$ with $\det(\mathbf{A})=1$, $\mathbf{A}^\top\sigma_2\mathbf{A}=\det(\mathbf{A})\sigma_2=\sigma_2$ we find
\begin{equation}
  \mathbf{U}(\xi;x,t)\mathbf{c}\cdot\varphi(x,t)\cdot\mathbf{c}^\top\sigma_2 =
  \mathbf{U}(\xi;x,t)\mathbf{c}\cdot\left(1 + \varphi(x,t)\mathbf{c}^\top\mathbf{U}(\xi;x,t)^\top\sigma_2\mathbf{U}'(\xi;x,t)\mathbf{c}\right)\cdot\mathbf{c}^\top\sigma_2.
  \label{eq:RC3}
\end{equation}
Since $\mathbf{c}\neq\mathbf{0}$ and $\mathbf{U}(\xi;x,t)$ is invertible, neither the column vector $\mathbf{U}(\xi;x,t)\mathbf{c}$ nor the row vector $\mathbf{c}^\top\sigma_2$ can vanish, so \eqref{eq:RC3} is just a scalar linear equation which can be solved explicitly for $\varphi(x,t)$:
\begin{equation}
  \varphi(x,t)=\varphi(\xi,\mathbf{c};x,t)\defeq\frac{1}{1-\mathbf{c}^\top\mathbf{U}(\xi;x,t)^\top\sigma_2\mathbf{U}'(\xi;x,t)\mathbf{c}}.
  \label{eq:alpha1-def}
\end{equation}
 
\subsection{Darboux transformation for the robust IST}
\label{s:DarbouxRobustIST}
The insertion of a single pole breaks the Schwarz symmetry \eqref{eq:U-sym} initially present in the matrix $\mathbf{U}(\lambda;x,t)$.  This is obviously true if $\xi\not\in\mathbb{R}$, but even if one takes $\xi\in D_0\cap\mathbb{R}$ one finds that the condition $\mathbf{R}(x,t)\neq 0$ is inconsistent with Schwarz symmetry of the Laurent expansion of $\dot{\mathbf{U}}(\lambda;x,t)$ about $\lambda=\xi$.  To maintain the Schwarz symmetry it is necessary to insert a complex-conjugate pair of poles at points $\xi$ and $\xi^*$ in $D_0$ with $\xi\neq\xi^*$.  
The insertion of the conjugate pair of poles is done in two consecutive steps, after which the poles are removed in favor of a modified jump:
\begin{itemize}
\item A gauge transformation of the form \eqref{eq:N-tilde} is applied with data $(\xi,\mathbf{c})$ as described in Section~\ref{s:pole-insertion}.  This step generates $\dot{\mathbf{U}}(\lambda;x,t)$ from $\mathbf{U}(\lambda;x,t)$.
\item A second gauge transformation of the form \eqref{eq:N-tilde} is applied with data $(\xi^*,\sigma_2\mathbf{c}^*)$.  This step generates a matrix $\ddot{\mathbf{U}}(\lambda;x,t)$ from $\dot{\mathbf{U}}(\lambda;x,t)$.
\item The two new simple-pole singularities of $\ddot{\mathbf{U}}(\lambda;x,t)$ are removed and instead transferred to the jump condition across the circle $\Sigma_0$ by making an explicit renormalization for $\lambda\in D_0$:
\begin{equation}
\widetilde{\mathbf{U}}(\lambda;x,t):=\begin{cases}\ddot{\mathbf{U}}(\lambda;x,t),&\quad\lambda\in D^+\cup D^-\\
\ddot{\mathbf{U}}(\lambda;x,t)\ddot{\mathbf{U}}(\lambda;L,0)^{-1},&\quad\lambda\in D_0.
\end{cases}
\label{eq:poles-to-circle}
\end{equation}
\end{itemize}
We first construct $\ddot{\mathbf{U}}(\lambda;x,t)$ by explicitly composing the two pole insertions:
\begin{equation}
  \ddot{\mathbf{U}}(\lambda;x,t)=\left(\mathbb{I}+\frac{\dot{\varphi}(\xi^*,\sigma_2\mathbf{c}^*;x,t)\dot{\mathbf{U}}(\xi^*;x,t)\sigma_2\mathbf{c}^*\mathbf{c}^{\dagger}\sigma_2^\top\dot{\mathbf{U}}(\xi^*;x,t)^\top\sigma_2}{\lambda-\xi^*}\right)\dot{\mathbf{U}}(\lambda;x,t),
  \label{eq:second-step}
\end{equation}
where
\begin{equation}
  \dot{\mathbf{U}}(\lambda;x,t)=\left(\mathbb{I}+\frac{\varphi(\xi,\mathbf{c};x,t)\mathbf{U}(\xi;x,t)\mathbf{c}\mathbf{c}^\top\mathbf{U}(\xi;x,t)^\top\sigma_2}{\lambda-\xi}\right){\mathbf{U}}(\lambda;x,t),
  \label{eq:first-step}
\end{equation}
and in \eqref{eq:second-step}, $\dot{\varphi}$ is defined from \eqref{eq:alpha1-def} using the matrix function $\dot{\mathbf{U}}(\lambda;x,t)$ in place of $\mathbf{U}(\lambda;x,t)$.
Inserting \eqref{eq:first-step} into \eqref{eq:second-step} and expanding the product of the parentheses in partial fractions gives $\ddot{\mathbf{U}}(\lambda;x,t)=\mathbf{G}(\lambda;x,t)\mathbf{U}(\lambda;x,t)$, where the composite gauge transformation matrix $\mathbf{G}(\lambda;x,t)$ is
\begin{equation}
\mathbf{G}(\lambda;x,t):=\mathbb{I} + \frac{\mathbf{Y}(x,t)}{\lambda-\xi} + \frac{\mathbf{Z}(x,t)}{\lambda-\xi^*}
  \label{eq:G-partial-fraction}
\end{equation}
with coefficients $\mathbf{Y}(x,t)$ and $\mathbf{Z}(x,t)$ described as follows. Set
\begin{equation}
  \begin{aligned}
    \beta&\defeq\mathrm{Im}(\xi),\\
    \mathbf{s}(x,t) &\defeq \mathbf{U}(\xi;x,t)\mathbf{c},\\
    N(x,t)&\defeq\|\mathbf{s}(x,t)\|^2=\mathbf{s}(x,t)^\dagger \mathbf{s}(x,t).
\end{aligned}
\label{eq:b-s-N}
\end{equation} 
Then, by using the given Schwarz symmetry condition $\mathbf{U}(\lambda^*;x,t)=\sigma_2\mathbf{U}(\lambda;x,t)^*\sigma_2$ (cf., \eqref{eq:U-sym}), along with the identity $\mathbf{v}^\top\sigma_2\mathbf{v}=0$ for all vectors $\mathbf{v}$, we find
\begin{multline}
  \mathbf{Y}(x,t)=\varphi(\xi,\mathbf{c};x,t)\left(1-\frac{\varphi(\xi,\mathbf{c};x,t)\dot{\varphi}(\xi^*,\sigma_2\mathbf{c}^*;x,t)N(x,t)^2}{4\beta^2}\right)\mathbf{s}(x,t)\mathbf{s}(x,t)^\top\sigma_2 \\-\frac{\varphi(\xi,\mathbf{c};x,t)\dot{\varphi}(\xi^*,\sigma_2\mathbf{c}^*;x,t)N(x,t)}{2\I \beta}\sigma_2\mathbf{s}(x,t)^*\mathbf{s}(x,t)^\top\sigma_2
  \label{eq:Y-def}
\end{multline}
and
\begin{multline}
  \mathbf{Z}(x,t)=-\dot{\varphi}(\xi^*,\sigma_2\mathbf{c}^*;x,t)\sigma_2\mathbf{s}(x,t)^*\mathbf{s}(x,t)^\dagger \\
  + \frac{\dot{\varphi}(\xi^*,\sigma_2\mathbf{c}^*;x,t)\varphi(\xi,\mathbf{c};x,t)N(x,t)}{2\I \beta}\mathbf{s}(x,t)\mathbf{s}(x,t)^\dagger.
  \label{eq:Z-def}
\end{multline}
We simplify the scalars $\varphi(\xi,\mathbf{c};x,t)$, $\dot{\varphi}(\xi^*,\sigma_2\mathbf{c};x,t)$, and  $\dot{\varphi}(\xi^*,\sigma_2\mathbf{c}^*;x,t)\varphi(\xi,\mathbf{c};x,t)$ as follows. First note that $\varphi(\xi,\mathbf{c};x,t)$ can be expressed as:
\begin{equation}
  \varphi(\xi,\mathbf{c};x,t)=\frac{1}{1-w(x,t)}=\frac{1-w(x,t)^*}{|1-w(x,t)|^2},
  \label{eq:varphi}
\end{equation}
where
\begin{equation}
  w(x,t)\defeq \mathbf{c}^\top\mathbf{U}(\xi;x,t)^\top \sigma_2 \mathbf{U}'(\xi;x,t)\mathbf{c}=\mathbf{s}(x,t)^\top\sigma_2\mathbf{s}'(x,t)\quad\text{and}\quad\mathbf{s}'(x,t):=\mathbf{U}'(\xi;x,t)\mathbf{c}.
  \label{eq:w-sprime}
\end{equation}
Next, note that
\begin{equation}
  \begin{aligned}
  \mathbf{c}^\dagger\sigma_2^\top\dot{\mathbf{U}}(\xi^*;x,t)^\top\sigma_2\dot{\mathbf{U}}'(\xi^*;x,t)\sigma_2\mathbf{c}^* &=-\frac{\varphi(\xi,\mathbf{c};x,t)N(x,t)^2}{4\beta^2}-\mathbf{s}(x,t)^\dagger\sigma_2\mathbf{s}'(x,t)^*\\
  &=-\frac{\varphi(\xi,\mathbf{c};x,t)N(x,t)^2}{4\beta^2}+w(x,t)^*
  \end{aligned}
\end{equation}
and hence $\dot{\varphi}(\xi^*,\sigma_2\mathbf{c}^*;x,t)$ can be simplified to: 
\begin{equation}
  \begin{aligned}
\dot{\varphi}(\xi^*,\sigma_2\mathbf{c}^*;x,t) &= \frac{1}{1-\mathbf{c}^\dagger \sigma_2^\top \dot{\mathbf{U}}(\xi^*;x,t)^\top\sigma_2 \dot{\mathbf{U}}'(\xi^*;x,t)\sigma_2\mathbf{c}^*}\\
&=4\beta^2\frac{|1-w(x,t)|^2}{\left(4\beta^2|1-w(x,t)|^2 + N(x,t)^2\right)(1-w(x,t)^*)}.
  \end{aligned}
  \label{eq:dot-varphi}
\end{equation}
Combining \eqref{eq:varphi} and \eqref{eq:dot-varphi} then gives
\begin{equation}
\varphi(\xi,\mathbf{c};x,t)\dot{\varphi}(\xi^*,\sigma_2\mathbf{c}^*;x,t)=\frac{4\beta^2}{4\beta^2|1-w(x,t)|^2 + N(x,t)^2}.
\label{eq:alpha1alpha2}
\end{equation}
Using \eqref{eq:varphi}, \eqref{eq:dot-varphi}, and \eqref{eq:alpha1alpha2} in \eqref{eq:Y-def} and \eqref{eq:Z-def} gives:
\begin{multline}
  \mathbf{Y}(x,t)=\frac{4\beta^2(1-w(x,t)^*)}{4\beta^2|1-w(x,t)|^2 + N(x,t)^2}\mathbf{s}(x,t)\mathbf{s}(x,t)^\top\sigma_2 \\+\frac{2\I \beta N(x,t)}{4\beta^2|1-w(x,t)|^2 + N(x,t)^2}\sigma_2\mathbf{s}(x,t)^*\mathbf{s}(x,t)^\top\sigma_2
 \label{eq:Ymat-formula}
\end{multline}
and
\begin{multline}
  \mathbf{Z}(x,t)=\frac{-4\beta^2(1-w(x,t))}{4\beta^2|1-w(x,t)|^2 + N(x,t)^2}\sigma_2\mathbf{s}(x,t)^*\mathbf{s}(x,t)^\dagger  \\+\frac{-2\I \beta N(x,t)}{4\beta^2|1-w(x,t)|^2 + N(x,t)^2}\mathbf{s}(x,t)\mathbf{s}(x,t)^\dagger.
 \label{eq:Zmat-formula}
\end{multline}
Note that
\begin{equation}
  4\beta^2|1-w(x,t)|^2 + N(x,t)^2 \geq N(x,t)^2,
  \label{eq:Ymat-denom-bound}
\end{equation}
which can never vanish for any $(x,t)\in\mathbb{R}^2$ because $\mathbf{c}\neq\mathbf{0}$ and $N(x,t)$ is the squared length of the vector $\mathbf{s}(x,t)$, a non-trivial linear combination of the columns of a matrix with unit determinant. It is now completely obvious that $\mathbf{Z}(x,t)=\sigma_2 \mathbf{Y}(x,t)^*\sigma_2$, which implies that $\mathbf{G}(\lambda^*;x,t)=\sigma_2\mathbf{G}(\lambda;x,t)^*\sigma_2$ and hence $\ddot{\mathbf{U}}(\lambda;x,t)=\mathbf{G}(\lambda;x,t)\mathbf{U}(\lambda;x,t)$ maintains the Schwarz symmetry of the original matrix function $\mathbf{U}(\lambda;x,t)$, namely $\ddot{\mathbf{U}}(\lambda^*;x,t)=\sigma_2\ddot{\mathbf{U}}(\lambda;x,t)^*\sigma_2$.  

A ``dressing'' construction shows that for each $\lambda\in\mathbb{C}\setminus(\Sigma\cup\{\xi,\xi^*\})$, $\ddot{\mathbf{U}}(\lambda;x,t)$ simultaneously satisfies Lax pair equations of the form \eqref{eq:Lax-x}--\eqref{eq:Lax-t} in which $\psi(x,t)$ is replaced by a modified potential $\widetilde{\psi}(x,t)$ given by (here $\ddot{\mathbf{M}}(\lambda;x,t):=\ddot{\mathbf{U}}(\lambda;x,t)\E^{\I\rho(\lambda)(x+\lambda t)\sigma_3}$ is a matrix satisfying $\ddot{\mathbf{M}}(\lambda;x,t)\to\mathbb{I}$ as $\lambda\to\infty$)
\begin{equation}
\begin{split}
\widetilde{\psi}(x,t)&=2i\lim_{\lambda\to\infty}\lambda\ddot{M}_{12}(\lambda;x,t)\\
&=2iY_{12}(x,t)+2iZ_{12}(x,t)+2i\lim_{\lambda\to\infty}\lambda M_{12}(\lambda;x,t)\\
&=\psi(x,t)+2iY_{12}(x,t)+2iZ_{12}(x,t)\\
&=\psi(x,t)+2i(Y_{12}(x,t)-Y_{21}(x,t)^*).
\end{split}
\label{eq:new-psi}
\end{equation}
To see this, one notes that since each gauge transformation step $\mathbf{U}(\lambda;x,t)\to\dot{\mathbf{U}}(\lambda;x,t)\to\ddot{\mathbf{U}}(\lambda;x,t)$ involves multiplication on the left by a matrix that is analytic everywhere except at the pole to be inserted, all pre-existant jump conditions and residue conditions of the form \eqref{eq:res1} are preserved.  Thus $\ddot{\mathbf{U}}(\lambda;x,t)$ is analytic for $\lambda\in\mathbb{C}\setminus(\Sigma\cup\{\xi,\xi^*\})$ and it satisfies the jump condition $\ddot{\mathbf{U}}_+(\lambda;x,t)=\ddot{\mathbf{U}}_-(\lambda;x,t)\mathbf{V}(\lambda)$ for $\lambda\in\Sigma$ exactly as does $\mathbf{U}(\lambda;x,t)$ according to \rhref{rhp:M-NZBC}.  Moreover the simple poles at $\xi$ and $\xi^*$ are characterized by the related residue conditions
\begin{equation}
\underset{\lambda=\xi}{\res}\,\ddot{\mathbf{U}}(\lambda;x,t)=\lim_{\lambda\to\xi}\ddot{\mathbf{U}}(\lambda;x,t)\mathbf{C}\quad\text{and}\quad
\underset{\lambda=\xi^*}{\res}\,\ddot{\mathbf{U}}(\lambda;x,t)=\lim_{\lambda\to\xi^*}\ddot{\mathbf{U}}(\lambda;x,t)\sigma_2\mathbf{C}^*\sigma_2.
\label{eq:res-both}
\end{equation}
Because the ``core'' jump matrix $\mathbf{V}(\lambda)$ and the residue matrix $\mathbf{C}=\mathbf{c}\mathbf{c}^\top\sigma_2$ are independent of $(x,t)\in\mathbb{R}^2$, $\ddot{\mathbf{U}}(\lambda;x,t)$ and its partial derivatives $\ddot{\mathbf{U}}_x(\lambda;x,t)$ and $\ddot{\mathbf{U}}_t(\lambda;x,t)$ are all analytic in the same domain and satisfy the same jump and residue conditions, and it then follows 
that the matrices
\begin{equation}
\ddot{\mathbf{X}}(\lambda;x,t):=\ddot{\mathbf{U}}_x(\lambda;x,t)\ddot{\mathbf{U}}(\lambda;x,t)^{-1}\quad\text{and}\quad
\ddot{\mathbf{T}}(\lambda;x,t):=\ddot{\mathbf{U}}_t(\lambda;x,t)\ddot{\mathbf{U}}(\lambda;x,t)^{-1}
\label{eq:ddotXT}
\end{equation}
have only removable singularities, and hence are essentially entire functions of $\lambda$.
Moreover, assuming that the asymptotic expansion 
\begin{equation}
\ddot{\mathbf{M}}(\lambda;x,t)= \mathbb{I}+\frac{\ddot{\mathbf{M}}^{(1)}(x,t)}{\lambda}+
\frac{\ddot{\mathbf{M}}^{(2)}(x,t)}{\lambda^2}+O(\lambda^{-3}),\quad\lambda\to\infty
\end{equation}
is differentiable term-by-term with respect to $(x,t)$, it follows that
\begin{equation}
\begin{split}
\ddot{\mathbf{X}}(\lambda;x,t)&=\left[\ddot{\mathbf{M}}(\lambda;x,t)\E^{-\I\rho(\lambda)(x+\lambda t)\sigma_3}\right]_x \E^{\I\rho(\lambda)( x+\lambda t)\sigma_3}\ddot{\mathbf{M}}(\lambda;x,t)^{-1}\\
&=-\I\rho(\lambda)\ddot{\mathbf{M}}(\lambda;x,t)\sigma_3\ddot{\mathbf{M}}(\lambda;x,t)^{-1} + \ddot{\mathbf{M}}_x(\lambda;x,t)\ddot{\mathbf{M}}(\lambda;x,t)^{-1}\\
&=-\I\lambda\sigma_3 +\I[\sigma_3,\ddot{\mathbf{M}}^{(1)}(x,t)] + O(\lambda^{-1})\\
&=-\I\lambda\sigma_3 +\I[\sigma_3,\ddot{\mathbf{M}}^{(1)}(x,t)],
\end{split}
\label{eq:A-linear}
\end{equation}
and
\begin{equation}
\begin{split}
\ddot{\mathbf{T}}(\lambda;x,t)&=\left[\ddot{\mathbf{M}}(\lambda;x,t)\E^{-\I\rho(\lambda)(x+\lambda t)\sigma_3}\right]_t \E^{\I\rho(\lambda)(x+\lambda t)\sigma_3}\ddot{\mathbf{M}}(\lambda;x,t)^{-1}\\
&=-\I\rho(\lambda)\lambda\ddot{\mathbf{M}}(\lambda;x,t)\sigma_3\ddot{\mathbf{M}}(\lambda;x,t)^{-1} + \ddot{\mathbf{M}}_t(\lambda;x,t)\ddot{\mathbf{M}}(\lambda;x,t)^{-1}\\
&=-\I\lambda^2\sigma_3+\I\lambda[\sigma_3,\ddot{\mathbf{M}}^{(1)}(x,t)] \\
&\quad\quad{}+ \I[\sigma_3,\ddot{\mathbf{M}}^{(2)}(x,t)]+\I[\ddot{\mathbf{M}}^{(1)}(x,t),\sigma_3\ddot{\mathbf{M}}^{(1)}(x,t)] - \frac{\I}{2}\sigma_3 + O(\lambda^{-1})\\
&=-\I\lambda^2\sigma_3+\I\lambda[\sigma_3,\ddot{\mathbf{M}}^{(1)}(x,t)] + \I[\sigma_3,\ddot{\mathbf{M}}^{(2)}(x,t)]+\I[\ddot{\mathbf{M}}^{(1)}(x,t),\sigma_3\ddot{\mathbf{M}}^{(1)}(x,t)]- \frac{\I}{2}\sigma_3,
\end{split}
\label{eq:B-linear}
\end{equation}
where the last equality in each case is a consequence of Liouville's Theorem. 
The dependence on $\ddot{\mathbf{M}}^{(2)}(x,t)$ in \eqref{eq:B-linear} can be eliminated because the coefficient of $\lambda^{-1}$ in the $O(\lambda^{-1})$ error term on the third line of \eqref{eq:A-linear} is
\begin{equation}
\I[\sigma_3,\ddot{\mathbf{M}}^{(2)}(x,t)] + \I[\ddot{\mathbf{M}}^{(1)}(x,t),\sigma_3\ddot{\mathbf{M}}^{(1)}(x,t)] + \ddot{\mathbf{M}}_x^{(1)}(x,t)-\frac{\I}{2}\sigma_3
\label{eq:M-tilde-error}
\end{equation}
which must vanish (again by Liouville's Theorem). Therefore $\ddot{\mathbf{T}}(\lambda;x,t)$ is the quadratic polynomial
\begin{equation}
\ddot{\mathbf{T}}(\lambda;x,t)=-\I\lambda^2\sigma_3 + \I\lambda[\sigma_3,\ddot{\mathbf{M}}^{(1)}(x,t)] - \ddot{\mathbf{M}}_x^{(1)}(x,t)
\end{equation}
and from the diagonal part of \eqref{eq:M-tilde-error} we see that
\begin{equation}
\ddot{M}^{(1)}_{11,x}(x,t)=2\I \ddot{M}^{(1)}_{12}(x,t)\ddot{M}^{(1)}_{21}(x,t) + \frac{\I}{2}\quad\text{and}\quad\ddot{M}^{(1)}_{22,x}(x,t)=-2\I \ddot{M}^{(1)}_{12}(x,t)\ddot{M}^{(1)}_{21}(x,t) - \frac{\I}{2}.
\label{eq:M2-diag}
\end{equation}
It then follows from the fact that $\ddot{\mathbf{M}}(\lambda^*;x,t)=\sigma_2\ddot{\mathbf{M}}(\lambda;x,t)^*\sigma_2$ and the definition \eqref{eq:new-psi} that $\ddot{\mathbf{X}}(\lambda;x,t)$ and $\ddot{\mathbf{T}}(\lambda;x,t)$ have exactly the form \eqref{eq:Lax-x} and \eqref{eq:Lax-t} respectively, in which $\psi(x,t)$ is merely replaced with $\widetilde{\psi}(x,t)$.  Rewriting \eqref{eq:ddotXT} in the form $\ddot{\mathbf{U}}_x=\ddot{\mathbf{X}}\ddot{\mathbf{U}}$ and $\ddot{\mathbf{U}}_t=\ddot{\mathbf{T}}\ddot{\mathbf{U}}$ shows that $\ddot{\mathbf{U}}(\lambda;x,t)$ is indeed a simultaneous solution matrix for the Lax pair \eqref{eq:Lax-x}--\eqref{eq:Lax-t} for the modified potential $\widetilde{\psi}(x,t)$.  This Lax pair is therefore also compatible, and it follows that $\widetilde{\psi}(x,t)$ is a solution of the focusing NLS equation \eqref{eq:NLS}.
The bound \eqref{eq:Ymat-formula} and the formula \eqref{eq:new-psi} together imply that since by hypothesis (cf., Theorem~\ref{t:rhp-N-NZBC-unique}) $\psi(x,t)$ is a global solution of the focusing NLS equation \eqref{eq:NLS}, so is the transformed potential $\widetilde{\psi}(x,t)$.  Therefore, \eqref{eq:new-psi} constitutes the \emph{B\"acklund transformation} that corresponds to the twice-iterated gauge transformation $\mathbf{U}\mapsto\ddot{\mathbf{U}}$.

Having constructed $\ddot{\mathbf{U}}(\lambda;x,t)$, it remains only to implement the final step by defining $\widetilde{\mathbf{U}}(\lambda;x,t)$ from $\ddot{\mathbf{U}}(\lambda;x,t)$ using \eqref{eq:poles-to-circle}.  For this, we observe that for $\lambda\in D_0$, the original matrix $\mathbf{U}(\lambda;x,t)$ is the solution $\mathbf{U}^\mathrm{in}(\lambda;x,t)$ of an initial-value problem for the simultaneous equations of the Lax pair \eqref{eq:Lax-x}--\eqref{eq:Lax-t} (cf., Proposition~\ref{p:Uin-def}) and therefore $\mathbf{U}(\lambda;L,0)=\mathbb{I}$ holds for all $\lambda\in D_0$.  Consequently, \eqref{eq:poles-to-circle} can be explicitly written in terms of the composite gauge transformation matrix $\mathbf{G}(\lambda;x,t)$ as
\begin{equation}
\widetilde{\mathbf{U}}(\lambda;x,t):=\begin{cases}\mathbf{G}(\lambda;x,t)\mathbf{U}(\lambda;x,t),&\quad\lambda\in D_+\cup D_-\\
\mathbf{G}(\lambda;x,t)\mathbf{U}(\lambda;x,t)\mathbf{G}(\lambda;L,0)^{-1},&\quad\lambda\in D_0.
\end{cases}
\label{eq:poles-to-circle-again}
\end{equation}
Since the multiplication of $\ddot{\mathbf{U}}(\lambda;x,t)$ on the right by $\mathbf{G}(\lambda;L,0)^{-1}$ for $\lambda\in D_0$ preserves the Lax pair equations \eqref{eq:Lax-x}--\eqref{eq:Lax-t} for the modified potential $\widetilde{\psi}(x,t)$, it follows that for all $\lambda\in D_0$ with the possible exception of $\xi$, $\xi^*$, $\widetilde{\mathbf{U}}(\lambda;x,t)$ is a simultaneous solution matrix for the modified Lax pair that satisfies also $\widetilde{\mathbf{U}}(\lambda;L,0)=\mathbb{I}$.  Therefore by the uniqueness asserted in Proposition~\ref{p:Uin-def}, for $\lambda\in D_0$, $\widetilde{\mathbf{U}}(\lambda;L,0)$ has removable singularities at $\xi$ and $\xi^*$ and therefore is analytic within $D_0$ (actually it can be continued to $\lambda\in\mathbb{C}$ as the entire function $\ddot{\mathbf{U}}^\mathrm{in}(\lambda;x,t)$ for the modified potential $\widetilde{\psi}(x,t)$).  Since $\xi,\xi^*\in D_0$, it is also obvious that $\widetilde{\mathbf{U}}(\lambda;x,t)$ is analytic for $\lambda\in D_+\cup D_-$ and since $\mathbf{G}(\lambda;x,t)\to\mathbb{I}$ as $\lambda\to\infty$, the related matrix $\widetilde{\mathbf{M}}(\lambda;x,t):=\widetilde{\mathbf{U}}(\lambda;x,t)\E^{\I\rho(\lambda)(x+\lambda t)\sigma_3}$ satisfies the conditions of \rhref{rhp:M-NZBC} in which the only change is that for $\lambda\in \Sigma_0\subset\Sigma$, the ``core'' jump matrix $\mathbf{V}(\lambda)$ is replaced by
\begin{equation}
\widetilde{\mathbf{V}}(\lambda):=\begin{cases}\mathbf{G}(\lambda;L,0)\mathbf{V}(\lambda),&\quad\lambda\in\Sigma_+\\
\mathbf{V}(\lambda)\mathbf{G}(\lambda;L,0)^{-1},&\quad\lambda\in\Sigma_-.
\end{cases}
\label{eq:modified-jump}
\end{equation}
A crucial fact is that since $\mathbf{G}(\lambda^*;x,t)=\sigma_2\mathbf{G}(\lambda;x,t)^*\sigma_2$, the modified jump matrix satisfies exactly the same Schwarz symmetry condition as does the original jump matrix, and hence the proof of Theorem~\ref{t:rhp-N-NZBC-unique} applies once again to guarantee unique solvability of the transformed Riemann-Hilbert problem for all $(x,t)\in\mathbb{R}^2$.

\begin{remark}
The use of nilpotent residue matrices $\mathbf{C}$ in \eqref{eq:res-both} that are not necessarily triangular matrices may seem unusual to some readers used to applying Darboux transformations to Beals-Coifman solutions in order to introduce poles specifically into one column or the other.  Here the generalization is necessary because $\xi\in D_0$ and in this domain $\mathbf{U}(\lambda;x,t)$ differs from $\mathbf{U}^\mathrm{BC}(\lambda;x,t)$ by a right-multiplication by a matrix depending on $\lambda$ only.  This matrix factor conjugates triangular nilpotent residue matrices into general two-nilpotent form.  
\end{remark}

\begin{remark}
Since after multiplying on the right by $\E^{\I\rho(\lambda)(x+\lambda t)\sigma_3}$ the result of the Darboux transformation $\mathbf{U}(\lambda;x,t)\mapsto\widetilde{\mathbf{U}}(\lambda;x,t)$ described above also satisfies \rhref{rhp:M-NZBC} with only a modified jump matrix on the circle $\Sigma_0$, it is obvious that the transformation may be iterated any number of times.  In each iteration it makes no difference whether the points $\xi,\xi^*$ at which the poles are first introduced and then transferred to a jump on $\Sigma_0$ vary from iteration to iteration or whether they are fixed once and for all.  
\end{remark}

\begin{remark}
In each iteration, the poles may be placed at any non-real conjugate pair of points within $D_0$.  Although the matrix $\widetilde{\mathbf{M}}(\lambda;x,t)=\widetilde{\mathbf{U}}(\lambda;x,t)\E^{\I\rho(\lambda)(x+\lambda t)\sigma_3}$ satisfies the conditions of \rhref{rhp:M-NZBC} for a modified ``core'' jump matrix $\widetilde{\mathbf{V}}(\lambda)$ on $\Sigma_0$, it need not be the case that the resulting potential $\widetilde{\psi}(x,t)$ satisfies the boundary condition \eqref{eq:NZBC} initially satisfied by $\psi(x,t)$.  For instance, if one starts with the Riemann-Hilbert problem for the background solution and applies a Darboux transformation in which $\xi$ is taken to lie on the imaginary axis between the origin and $\pm\I$, then for some choices of the auxiliary parameters $\mathbf{c}$ the resulting solution $\widetilde{\psi}(x,t)$ is an \emph{Akhmediev breather} \cite{Akhmediev}, a solution that is periodic in $x$ for fixed $t$, and hence does not decay to the background as $x\to\pm\infty$.  The deciding factor in whether the boundary conditions are preserved is whether the gauge transformation matrix $\mathbf{G}(\lambda;x,t)$ decays to the identity as $x\to+\infty$ and remains bounded as $x\to-\infty$.  In many cases, this decay can be verified directly.
\end{remark}

\begin{remark}
\label{r:limiting-case}
The composite gauge transformation matrix $\mathbf{G}(\lambda;x,t)$ taking $\mathbf{U}(\lambda;x,t)$ into $\ddot{\mathbf{U}}(\lambda;x,t)$ also makes sense in a limiting case where $\xi\in D_0\setminus\mathbb{R}$ is fixed and $\mathbf{c}$ is taken in the form $\epsilon^{-1}\mathbf{c}_\infty$ for a fixed vector $\mathbf{c}_\infty\in\mathbb{C}^2\setminus\{\mathbf{0}\}$ and $\epsilon\in\mathbb{C}$, after which the limit $\epsilon\to 0$ is taken in the matrices $\mathbf{Y}(x,t)$ and $\mathbf{Z}(x,t)$ given in \eqref{eq:Ymat-formula} and \eqref{eq:Zmat-formula} respectively.  If we denote by $\mathbf{s}_\infty(x,t)$, $\mathbf{s}'_\infty(x,t)$, $N_\infty(x,t)$, and $w_\infty(x,t)$ the quantities defined in \eqref{eq:b-s-N} and \eqref{eq:w-sprime} where $\mathbf{c}$ is simply replaced with $\mathbf{c}_\infty$, then $\mathbf{s}(x,t)=\epsilon^{-1}\mathbf{s}_\infty(x,t)$ and $\mathbf{s}'(x,t)=\epsilon^{-1}\mathbf{s}'_\infty(x,t)$ while $N(x,t)=|\epsilon|^{-2}N_\infty(x,t)$ and $w(x,t)=\epsilon^{-2}w_\infty(x,t)$, so it follows easily that in this situation
\begin{multline}
\label{eq:Y-bar}
\mathbf{Y}_\infty(x,t):=\lim_{\epsilon\to 0}\mathbf{Y}(x,t)=-\frac{4\beta^2w_\infty(x,t)^*}{4\beta^2|w_\infty(x,t)|^2+N_\infty(x,t)^2}\mathbf{s}_\infty(x,t)\mathbf{s}_\infty(x,t)^\top\sigma_2\\{}+\frac{2\I\beta N_\infty(x,t)}{4\beta^2|w_\infty(x,t)|^2+N_\infty(x,t)^2}\sigma_2\mathbf{s}_\infty(x,t)^*\mathbf{s}_\infty(x,t)^\top\sigma_2
\end{multline}
and
\begin{multline}
\label{eq:Z-bar}
\mathbf{Z}_\infty(x,t):=\lim_{\epsilon\to 0}\mathbf{Z}(x,t)=\frac{4\beta^2w_\infty(x,t)}{4\beta^2|w_\infty(x,t)|^2+N_\infty(x,t)^2}\sigma_2\mathbf{s}_\infty(x,t)^*\mathbf{s}_\infty(x,t)^\dagger\\{}-\frac{2\I\beta N_\infty(x,t)}{4\beta^2|w_\infty(x,t)|^2+N_\infty(x,t)^2}\mathbf{s}_\infty(x,t)\mathbf{s}_\infty(x,t)^\dagger.
\end{multline}
The corresponding composite gauge transformation $\mathbf{G}_\infty(\lambda;x,t):=\mathbb{I}+(\lambda-\xi)^{-1}\mathbf{Y}_\infty(x,t)+(\lambda-\xi^*)^{-1}\mathbf{Z}_\infty(x,t)$ can then be used in place of $\mathbf{G}(\lambda;x,t)$ whenever desired.  This limiting case is especially useful for locating rogue wave solutions at the normalization point $(x,t)=(L,0)$ as will be seen shortly.  It is also obvious from the formulae \eqref{eq:Y-bar}--\eqref{eq:Z-bar} that $\mathbf{G}_\infty(\lambda;x,t)$ only depends on $\mathbf{c}_\infty\in\mathbb{C}^2\setminus\{0\}$ up to a nonzero complex multiple; in other words, the parameter space of the limiting Darboux transformation for fixed $\xi$ is complex projective space:  $\mathbf{c}_\infty\in\mathbb{CP}^1$.
\end{remark}

We summarize this description of the Darboux transformation for the robust IST in the following theorem.
\begin{theorem}
\label{t:Darboux}
Let $\mathbf{U}(\lambda;x,t):=\mathbf{M}(\lambda;x,t)\E^{-\I\rho(\lambda)(x+\lambda t)\sigma_3}$ correspond to the unique solution $\mathbf{M}(\lambda;x,t)$ of \rhref{rhp:M-NZBC} formulated with ``core'' jump matrices $\mathbf{V}^\mathbb{R}(\lambda)$ defined on $\Sigma_\mathrm{L}\cup\Sigma_\mathrm{R}$ and $\mathbf{V}^\pm(\lambda)$ defined on $\Sigma_\pm$, and generating a solution $\psi(x,t)$ of the focusing nonlinear Schr\"odinger equation \eqref{eq:NLS} via the formula \eqref{eq:NLS-recovery}.  Given Darboux transformation data $\xi\in D_0\setminus\mathbb{R}$ and $\mathbf{c}\in\mathbb{C}^2\setminus\{\mathbf{0}\}$, define the gauge transformation matrix $\mathbf{G}(\lambda;x,t)$  by \eqref{eq:G-partial-fraction} in terms of the coefficient matrices $\mathbf{Y}(x,t)$ and $\mathbf{Z}(x,t)$
given by \eqref{eq:Ymat-formula} and \eqref{eq:Zmat-formula} respectively, using \eqref{eq:b-s-N} and \eqref{eq:w-sprime}.  Then
\begin{itemize}
\item $\widetilde{\mathbf{M}}(\lambda;x,t):=\widetilde{\mathbf{U}}(\lambda;x,t)\E^{\I\rho(\lambda)(x+\lambda t)\sigma_3}$, where $\widetilde{\mathbf{U}}(\lambda;x,t)$ is defined explicitly in terms of $\mathbf{G}(\lambda;x,t)$ and $\mathbf{U}(\lambda;x,t)$ by \eqref{eq:poles-to-circle-again}, is the unique solution of another Riemann-Hilbert problem of the form of \rhref{rhp:M-NZBC} with the same jump contour and ``core'' jump matrices $\widetilde{\mathbf{V}}^\mathbb{R}(\lambda):=\mathbf{V}^\mathbb{R}(\lambda)$ for $\lambda\in\Sigma_\mathrm{L}\cup\Sigma_\mathrm{R}$ and $\widetilde{\mathbf{V}}(\lambda)$ is defined in terms of $\mathbf{V}(\lambda)$ for $\lambda\in\Sigma_0=\Sigma_+\cup\Sigma_-$ by \eqref{eq:modified-jump}.
\item The transformed Riemann-Hilbert problem generates a new solution $\widetilde{\psi}(x,t)$ of \eqref{eq:NLS} by the formula \eqref{eq:NLS-recovery} in which $\mathbf{M}(\lambda;x,t)$ is replaced by $\widetilde{\mathbf{M}}(\lambda;x,t)$ on the right-hand side.  Equivalently, $\widetilde{\psi}(x,t)$ is given in terms of $\psi(x,t)$ and the Darboux transformation data by the B\"acklund transformation \eqref{eq:new-psi}.
\end{itemize}
All of these statements also apply for transformation data $(\xi,\mathbf{c})$ with $\mathbf{c}=\infty$ in the sense of Remark~\ref{r:limiting-case}, i.e., $\mathbf{c}$ is replaced by $\mathbf{c}_\infty\in\mathbb{CP}^1$ and $\mathbf{Y}(x,t)$ and $\mathbf{Z}(x,t)$ are replaced by $\mathbf{Y}_\infty(x,t)$ and $\mathbf{Z}_\infty(x,t)$ given by \eqref{eq:Y-bar} and \eqref{eq:Z-bar} respectively.
\end{theorem}

\subsection{The simplest case:  Darboux/B\"acklund transformation of the background potential in the setting of the robust IST}
\label{s:OneDarboux}

As mentioned in Section~\ref{s:IST-NZBC}, the background potential $\psi=\psi_\mathrm{bg}(x,t)\equiv 1$ has the simultaneous fundamental solution matrix $\mathbf{U}_\mathrm{bg}(\lambda;x,t)=\mathbf{E}(\lambda)\E^{-\I\rho(\lambda)(x+\lambda t)\sigma_3}$
for the corresponding linear equations of the Lax pair \eqref{eq:Lax-x}--\eqref{eq:Lax-t}.  Here $\mathbf{E}(\lambda)$ is the matrix function defined in \eqref{eq:NZBC-asymptotic-eigenfunction}.  Taking $t=0$ and $\lambda\in\Gamma$, we may normalize this solution in the limits $x\to\pm\infty$ to obtain the Jost solution matrices $\mathbf{J}_\mathrm{bg}^\pm(\lambda;x)$.  It is easy to confirm that in fact $\mathbf{J}^+_\mathrm{bg}(\lambda;x)=\mathbf{J}^-_\mathrm{bg}(\lambda;x)=\mathbf{U}_\mathrm{bg}(\lambda;x,0)$, so:
\begin{itemize}
\item The scattering matrix satisfies $\mathbf{S}_\mathrm{bg}(\lambda;t)=\mathbb{I}$ for all $\lambda\in\Gamma$ and $t\in\mathbb{R}$.
\item The Beals-Coifman matrix is analytic for $\lambda\in\mathbb{C}\setminus\Sigma_\mathrm{c}$ and is given explicitly by 
$\mathbf{U}^\mathrm{BC}_\mathrm{bg}(\lambda;x,t)=\mathbf{U}_\mathrm{bg}(\lambda;x,t)$.
\end{itemize}
Therefore, the scattering data for the background solution $\psi=\psi_\mathrm{bg}(x,t)\equiv 1$ in the setting of the robust transform described in Section~\ref{s:New-IST} consists of:
\begin{itemize}
\item the ``core'' jump matrix $\mathbf{V}^\mathbb{R}_\mathrm{bg}(\lambda)\equiv\mathbb{I}$ for $\lambda\in\Sigma_\mathrm{L}\cup\Sigma_\mathrm{R}$,
\item the ``core'' jump matrix $\mathbf{V}^+_\mathrm{bg}(\lambda)=\mathbf{E}(\lambda)$ for $\lambda\in\Sigma_+$, and
\item the ``core'' jump matrix $\mathbf{V}^-_\mathrm{bg}(\lambda)=\mathbf{E}(\lambda)^{-1}$ for $\lambda\in\Sigma_-$.
\end{itemize}
Here the radius $r=r_\mathrm{bg}$ of $\Sigma_0$ can be taken to be any number larger than $1$.  It will also be useful to have an explicit expression for the matrix $\mathbf{U}^\mathrm{in}_\mathrm{bg}(\lambda;x,t)$ described by Proposition~\ref{p:Uin-def} for the case of $\psi=\psi_\mathrm{bg}(x,t)\equiv 1$.  It can be found by normalizing\footnote{For convenience, for the rest of Section~\ref{s:Darboux} we choose the normalization point of the robust IST to be $L=0$.} the matrix $\mathbf{U}_\mathrm{bg}(\lambda;x,t)$ at $(x,t)=(0,0)$ to the identity:
\begin{equation}
\mathbf{U}^\mathrm{in}_\mathrm{bg}(\lambda;x,t)=\mathbf{U}_\mathrm{bg}(\lambda;x,t)\mathbf{U}_\mathrm{bg}(\lambda;0,0)^{-1}=
\mathbf{E}(\lambda)\E^{-\I\rho(\lambda)(x+\lambda t)\sigma_3}\mathbf{E}(\lambda)^{-1}.
\label{eq:Uin-bg-1}
\end{equation}
As pointed out in Remark~\ref{r:Uin-renormalize}, this formula can have only removable singularities for $\lambda\in\mathbb{C}$, and we can see it explicitly by using the definition of $\mathbf{E}(\lambda)$ given in \eqref{eq:NZBC-asymptotic-eigenfunction}--\eqref{eq:f-def}.  Thus, \eqref{eq:Uin-bg-1}
can be rewritten as
\begin{equation}
\mathbf{U}^\mathrm{in}_\mathrm{bg}(\lambda;x,t)=
(x+\lambda t)\frac{\sin(\theta)}{\theta}\begin{bmatrix}-\I\lambda & 1\\-1 & \I\lambda\end{bmatrix} +
\cos(\theta)\mathbb{I},\quad
\theta:=\rho(\lambda)(x+\lambda t),
\label{eq:Uin-bg-2}
\end{equation}
and since $\sin(\theta)/\theta$ and $\cos(\theta)$ are both even entire functions of $\theta$, the desired analyticity for $\lambda\in\mathbb{C}$ is now obvious because $\rho(\lambda)^2=\lambda^2+1$.

Now we choose an arbitrary point $\xi\in\mathbb{C}\setminus\mathbb{R}$ and assume that the radius $r$ of the circle $\Sigma_0$ exceeds $|\xi|$ so that $\xi\in D_0$.  We also select a vector $\mathbf{c}\in\mathbb{C}^2\setminus\{\mathbf{0}\}$.  To apply the Darboux transformation we need to calculate the vectors $\mathbf{s}(x,t):=\mathbf{U}_\mathrm{bg}(\xi;x,t)\mathbf{c}=\mathbf{U}_\mathrm{bg}^\mathrm{in}(\xi;x,t)\mathbf{c}$ and $\mathbf{s}'(x,t):=\mathbf{U}'_\mathrm{bg}(\xi;x,t)\mathbf{c}=\mathbf{U}_\mathrm{bg}^{\mathrm{in}\prime}(\xi;x,t)\mathbf{c}$.  In general, 
\begin{equation}
\mathbf{s}(x,t)=-\I\xi\mathscr{S}(\xi;x,t)\sigma_3\mathbf{c}+\I\mathscr{S}(\xi;x,t)\sigma_2\mathbf{c}+
\mathscr{C}(\xi;x,t)\mathbf{c},
\label{eq:sxt}
\end{equation}
where
\begin{equation}
\mathscr{S}(\xi;x,t):=\frac{\sin(\rho(\xi)(x+\xi t))}{\rho(\xi)}\quad\text{and}\quad
\mathscr{C}(\xi;x,t):=\cos(\rho(\xi)(x+\xi t))
\end{equation}
are well defined and analytic for $\xi\in D_0$.  Differentiating with respect to $\xi$ gives
\begin{equation}
\mathbf{s}'(x,t)=-\I\xi\mathscr{S}'(\xi;x,t)\sigma_3\mathbf{c}-\I\mathscr{S}(\xi;x,t)\sigma_3\mathbf{c}+\I\mathscr{S}'(\xi;x,t)\sigma_2\mathbf{c}+\mathscr{C}'(\xi;x,t)\mathbf{c}.
\label{eq:sprimext}
\end{equation}
From \eqref{eq:sxt} and \eqref{eq:sprimext} we calculate $N(x,t):=\mathbf{s}(x,t)^\dagger\mathbf{s}(x,t)$ as
\begin{multline}
N(x,t)=(\mathbf{c}^\dagger\mathbf{c})\left[(1+|\xi|^2)|\mathscr{S}(\xi;x,t)|^2+|\mathscr{C}(\xi;x,t)|^2\right]+2\beta(\mathbf{c}^\dagger\sigma_1\mathbf{c})|\mathscr{S}(\xi;x,t)|^2\\
{}+2(\mathbf{c}^\dagger\sigma_3\mathbf{c})\Im\{\xi\mathscr{S}(\xi;x,t)\mathscr{C}(\xi;x,t)^*\}-2(\mathbf{c}^\dagger\sigma_2\mathbf{c})\Im\{\mathscr{S}(\xi;x,t)\mathscr{C}(\xi;x,t)^*\}
\end{multline}
and
$w(x,t):=\mathbf{s}(x,t)^\top\sigma_2\mathbf{s}'(x,t)$ as
\begin{multline}
w(x,t)=
-(\mathbf{c}^\top\sigma_3\mathbf{c})\mathscr{S}(\xi;x,t)^2 + (\mathbf{c}^\top\sigma_1\mathbf{c})\mathscr{S}(\xi;x,t)\mathscr{C}(\xi;x,t) \\
{}+ (\mathbf{c}^\top[\I\mathbb{I}+\xi\sigma_1]\mathbf{c})(\mathscr{C}(\xi;x,t)\mathscr{S}'(\xi;x,t)-\mathscr{S}(\xi;x,t)\mathscr{C}'(\xi;x,t)).
\label{eq:w-one-time}
\end{multline}
In terms of these, the B\"acklund transformation \eqref{eq:new-psi} takes the form ($\psi(x,t)=\psi_\mathrm{bg}(x,t)\equiv 1$ here)
\begin{multline}
\widetilde{\psi}(x,t)=\psi(x,t)\\
{}+8\beta\frac{\beta s_1(x,t)^2(1-w(x,t)^*)-\beta s_2(x,t)^{*2}(1-w(x,t))+ s_1(x,t)s_2(x,t)^*N(x,t)}{4\beta^2|1-w(x,t)|^2+N(x,t)^2}.
\label{eq:Backlund}
\end{multline}

\subsubsection{Properties of the solutions obtained for $\xi\neq\pm\I$}
Suppose that $\xi\neq\pm\I$.
Given such a value of $\xi$, if $\mathbf{c}$ is chosen such that $\mathbf{c}^\top[\I\mathbb{I}+\xi\sigma_1]\mathbf{c}=0$, the solution obtained from the B\"acklund transformation exhibits quite a different character than for $\mathbf{c}$ in general position.  Indeed, this condition removes the terms on the second line of \eqref{eq:w-one-time}, with the result being that all dependence on $(x,t)$ in $\widetilde{\psi}(x,t)-\psi(x,t)$ enters through the functions $\mathscr{S}(\xi;x,t)$ and $\mathscr{C}(\xi;x,t)$.  In the general case, however, derivatives of these functions with respect to $\xi$ appear, and these produce additional dependence on $(x,t)$ via polynomial factors.  For convenience, we illustrate the difference between the special and general case for the Kuznetsov-Ma family of solutions \cite{Kuznetsov1977,Ma1979b}, corresponding to choosing $\xi=(1+\delta)\I$ for some $\delta>0$.  For such $\xi$, the condition $\mathbf{c}^\top[\I\mathbb{I}+\xi\sigma_1]\mathbf{c}=0$ reads $c_1^2+c_2^2+2(1+\delta)c_1c_2=0$, which implies $c_2=c_1(-(1+\delta)\pm\sqrt{(1+\delta)^2-1})$.  Figure~\ref{fig:KM} compares the solutions $\widetilde{\psi}(x,t)$ obtained for $\xi=2\I$ with $c_1=1+\I$ and $c_2=c_1(-2+\sqrt{3})$ (the special case, on the left) and with $c_1=1$ and $c_2=1+\I$ (the general case, on the right).
\begin{figure}[h]
\begin{center}
\includegraphics{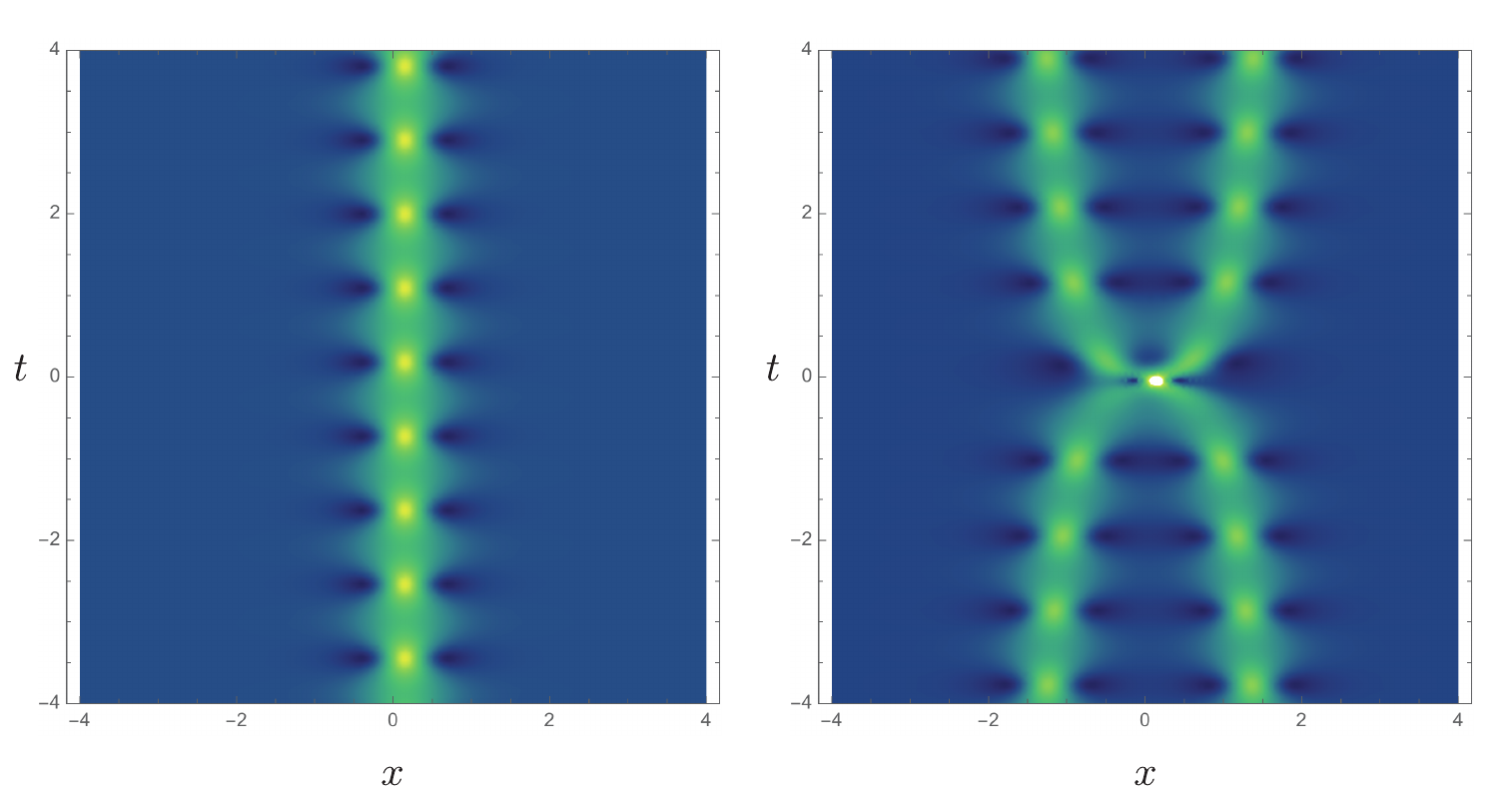}
\end{center}
\caption{Left:  $|\widetilde{\psi}(x,t)|$ for $\xi=2\I$, $c_1=1+\I$, and $c_2=c_1(-2+\sqrt{3})$ for which the solution is periodic in $t$ and exponential in $x$.  Right:  $|\widetilde{\psi}(x,t)|$ for $\xi=2\I$, $c_1=1$, and $c_2=1+\I$ for which the solution formula includes additional terms polynomial in $(x,t)$.}
\label{fig:KM}
\end{figure}
The reader may be familiar with the fact that in the context of the traditional IST for the Cauchy problem \eqref{eq:NLS}--\eqref{eq:IC}, the application of a single Darboux transformation inserting poles at a conjugate pair $\lambda=\pm\I (1+\delta)$ for $\delta>0$ always produces a solution periodic in $t$ and exponentially localized in $x$ as in the special case shown in the plot in the left-hand panel of Figure~\ref{fig:KM}.  The reason that the special case is selected in that setting is that the Darboux transformation is applied to the Beals-Coifman solution $\mathbf{U}^\mathrm{BC}(\lambda;x,t)$ rather than to $\mathbf{U}^\mathrm{in}(\lambda;x,t)$, and then the parameters $(c_1,c_2)$ are chosen to make sure that the pole is inserted into one column only of the simultaneous solution matrix (that is, the Darboux transformation is designed to introduce an additional Blaschke factor $(\lambda-\xi)/(\lambda-\xi^*)$ into the function $a(\lambda)$).  Recalling that working with $\mathbf{U}^\mathrm{in}(\lambda;x,t)$ rather than $\mathbf{U}^\mathrm{BC}(\lambda;x,t)$ introduces a rotation in the parameter space, the condition $\mathbf{c}^\top[\I\mathbb{I}+\xi\sigma_1]\mathbf{c}=0$ is seen as the analogue for the Darboux transformation of the robust IST of the latter choice.  Moreover, we see that for $\mathbf{c}$ in general position, the solution obtained (see the right-hand panel of Figure~\ref{fig:KM}) is what one would expect in the context of the traditional IST from a solution in which $\mathbf{U}^\mathrm{BC}(\lambda;x,t)$ has a conjugate pair of \emph{double poles}; the solution resembles a non-ballistic collision of identical Kuznetsov-Ma solutions.  

Under the special condition that $\mathbf{c}^\top[\I\mathbb{I}+\xi\sigma_1]\mathbf{c}=0$, all $(x,t)$-dependence in the solution enters via the functions $\mathscr{S}(\xi;x,t)$ and $\mathscr{C}(\xi;x,t)$.  These functions are periodic and bounded in the real part of $\rho(\xi)(x+\xi t)$, but they grow exponentially in all directions for which $\Im\{\rho(\xi)(x+\xi t)\}$ is unbounded.  The only values of $\xi\in D_0\setminus\mathbb{R}$ for which $\rho(\xi)$ is real and nonzero are the values $\xi=\I\delta$ with $\delta\in (-1,0)\cup (0,1)$ and only for these values is the solution periodic rather than exhibiting exponential decay to the background in the $x$-direction.  This case corresponds to the so-called \emph{Akhmediev breather} solutions \cite{Akhmediev}, and such solutions clearly do not satisfy the boundary condition \eqref{eq:NZBC}.  Nonetheless, they have a Riemann-Hilbert representation obtained by modifying the jump matrices $\mathbf{V}^\pm_\mathrm{bg}(\lambda)=\mathbf{E}(\lambda)^{\pm 1}$ for $\lambda\in\Sigma_0$ by the gauge transformation matrix $\mathbf{G}(\lambda;0,0)$ as indicated in \eqref{eq:modified-jump}.  Although the solution is no longer periodic in the $x$-direction for $\xi=\I\delta$ with $\delta\in (-1,0)\cup (0,1)$ if $\mathbf{c}$ is in general position ($\mathbf{c}^\top[\I\mathbb{I}+\xi\sigma_1]\mathbf{c}=\I(c_1^2+c_2^2+2\delta c_1c_2)\neq 0$), still the boundary condition \eqref{eq:NZBC} is not satisfied.  See Figure~\ref{fig:Akhmediev}.
\begin{figure}[h]
\begin{center}
\includegraphics{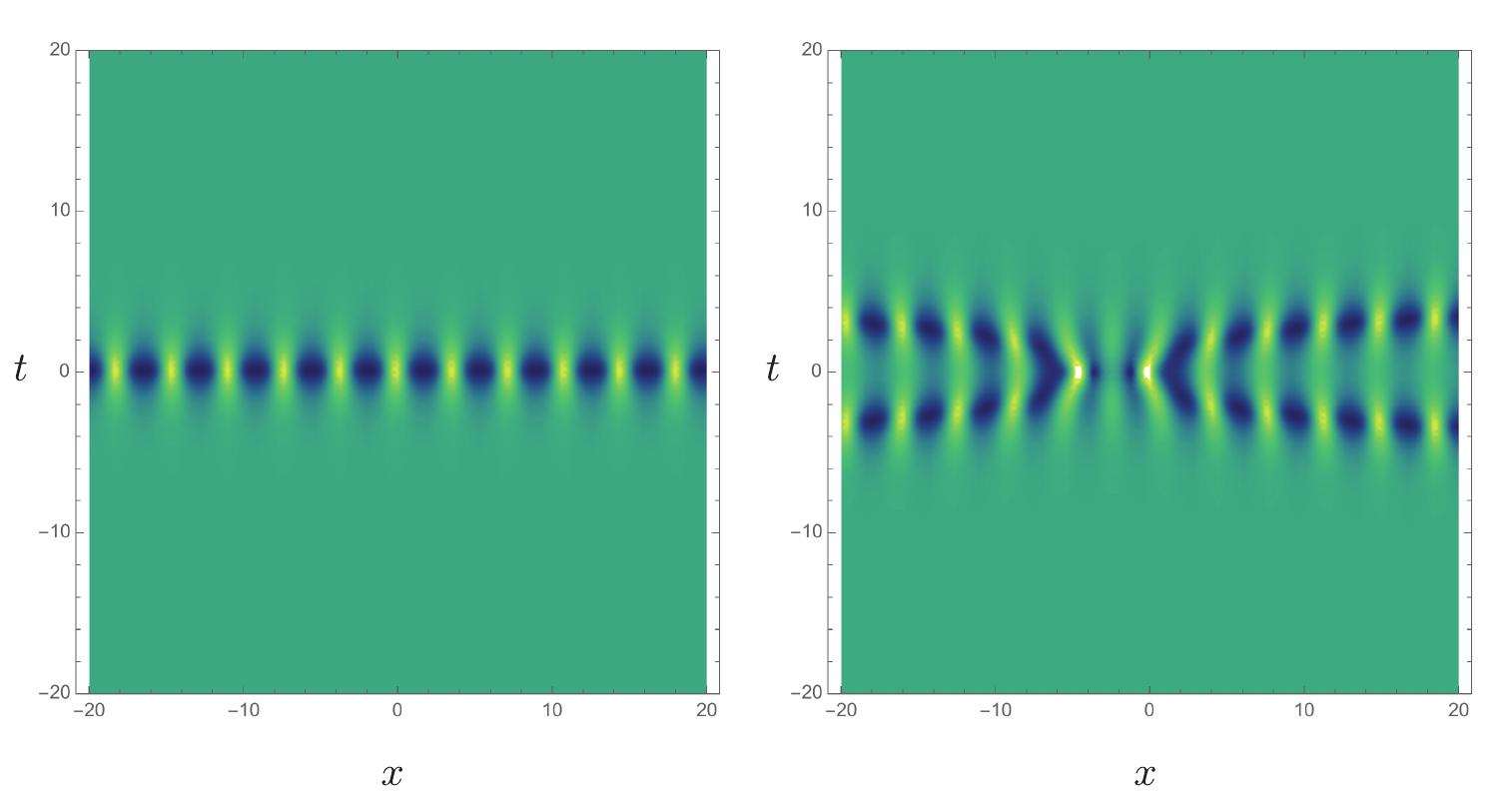}
\end{center}
\caption{Left:  $|\widetilde{\psi}(x,t)|$ for $\xi=\tfrac{1}{2}\I$, $c_1=1$, and $c_2=c_1(-\tfrac{1}{2}+\tfrac{1}{2}\I\sqrt{3})$ for which the solution is periodic in $x$ and exponential in $t$.  Right:  $|\widetilde{\psi}(x,t)|$ for $\xi=\tfrac{1}{2}\I$, $c_1=1$, and $c_2=0$ for which the solution formula includes additional terms polynomial in $(x,t)$.}
\label{fig:Akhmediev}
\end{figure}

\subsubsection{Properties of the solutions obtained for $\xi=\I$}
\label{s:OneDarbouxAtI}
Now we return to the previously-excluded case that $\xi=\pm\I$ (we take $\xi=\I$ to be precise). The reason for considering this case separately is \emph{not} that the Darboux transformation method described in Section~\ref{s:DarbouxRobustIST} requires any modification\footnote{This is a key property of the Darboux transformation in the setting of the robust IST that distinguishes it from other methods in the literature.}, but rather that the nature of the resulting solution is quite different from the general case.  
The reason for the difference emerges upon evaluation of the analytic functions $\mathscr{S}(\xi;x,t)$ and $\mathscr{C}(\xi;x,t)$ and their derivatives at $\xi=\I$ which requires some version of l'H\^opital's rule.  Equivalently, by Taylor expansion of $\sin(\theta)$ and $\cos(\theta)$ about $\theta=0$ one finds
\begin{equation}
\mathscr{S}(\I;x,t)=x+\I t,\quad\mathscr{S}'(\I;x,t)=t-\frac{1}{3}\I(x+\I t)^3,\quad\mathscr{C}(\I;x,t)=1,\quad\mathscr{C}'(\I;x,t)=-\I(x+\I t)^2.
\end{equation}
Therefore, in the special case that $\xi=\I$, the formulae \eqref{eq:sxt} and \eqref{eq:sprimext} become
\begin{equation}
\mathbf{s}(x,t)=\begin{bmatrix}c_1+(c_1+c_2)(x+\I t)\\c_2-(c_1+c_2)(x+\I t)\end{bmatrix},\quad\xi=\I
\label{eq:s-for-I}
\end{equation}
and
\begin{equation}
\mathbf{s}'(x,t)=\begin{bmatrix}(t-\tfrac{1}{3}\I(x+\I t)^3)(c_1+c_2)-\I(x+\I t)c_1-\I(x+\I t)^2c_1\\
-(t-\tfrac{1}{3}\I(x+\I t)^3)(c_1+c_2)+\I(x+\I t)c_2-\I(x+\I t)^2c_2\end{bmatrix},\quad\xi=\I.
\end{equation}
Therefore, $N(x,t)=\mathbf{s}(x,t)^\dagger\mathbf{s}(x,t)$ is
\begin{equation}
      N(x,t)
      =|c_1|^2 + |c_2|^2 + 2\Re\{(c_1-c_2)^*(c_1+c_2)(x+\I t)\}+2|c_1+c_2|^2|x+\I t|^2,\quad \xi=\I,
      \label{eq:norm-bg}
\end{equation}
and $w(x,t)=\mathbf{s}(x,t)^\top\sigma_2\mathbf{s}'(x,t)$ is
\begin{equation}
    w(x,t)=-\frac{2}{3}(c_1+c_2)^2 (x+\I t)^3 - (c_1^2 - c_2^2) (x+\I t)^2 + 2 c_1 c_2 (x+\I t) + \I t(c_1+c_2)^2,\quad \xi=\I.
    \label{eq:w-bg}
\end{equation}

Noting that for large $(x,t)$, the dominant terms in $\mathbf{s}(x,t)$, $\mathbf{s}'(x,t)$, $N(x,t)$, and $w(x,t)$ are all proportional to $c_1+c_2$, it is clear that the solution $\widetilde{\psi}(x,t)$ obtained from the B\"acklund transformation \eqref{eq:Backlund} has a different character if $c_1+c_2=0$ than otherwise.  Indeed, if $c_1=-c_2=c$, then
\begin{equation}
\mathbf{s}(x,t)=\begin{bmatrix}c\\-c\end{bmatrix},\quad N(x,t)=2|c|^2, \quad
\text{and}\quad
w(x,t)=-2c^2(x+\I t),\quad \xi=\I,\quad c_1=-c_2=c.
\label{eq:sNw-Peregrine}
\end{equation}
Using these formulae in \eqref{eq:Backlund} (also with $\beta=\Im\{\xi\}=1$) then yields the Peregrine breather solution \cite{Peregrine1983} $\widetilde{\psi}(x,t)=\psi_\mathrm{P}(x,t)$ \eqref{eq:Peregrine} (see Figure~\ref{f:Peregrine})
with parameters
\begin{equation}
x_0=-\frac{\Re\{c^2\}}{2|c|^4}\quad\text{and}\quad t_0=\frac{\Im\{c^2\}}{2|c|^4}.
\end{equation}
For finite $c$, the peak of the breather can be placed anywhere except the origin (this coincides with the normalization point in Proposition~\ref{p:Uin-def}).  If it is desired to place the breather exactly at the origin, we simply use the limiting case of the Darboux transformation described in Remark~\ref{r:limiting-case} with homogeneous coordinates $\mathbf{c}_\infty=[c: -c]^\top\in\mathbb{CP}^1$.  

If $c_1+c_2\neq 0$, the solution for $\xi=\I$ is more complicated.  We may introduce complex parameters $c$ and $\delta$ such that $c_1=c+\tfrac{1}{2}\delta$ and $c_2=-c+\tfrac{1}{2}\delta$ so that $\delta=c_1+c_2$ measures the deviation from the Peregrine case.  Taking $c\in\mathbb{C}\setminus\{0\}$ fixed and $\delta$ small, examination of the expressions \eqref{eq:norm-bg}--\eqref{eq:w-bg} suggests that a natural scaling of $(x,t)\in\mathbb{R}^2$ is to set $x=\bar{x}/|\delta|$ and $t=\bar{t}/|\delta|$.  If $(\bar{x},\bar{t})\in\mathbb{R}^2$ is also fixed, then as $\delta\to 0$,
\begin{equation}
N(x,t)=2|c|^2+4\Re\{c^*\E^{\I\arg(\delta)}(\bar{x}+\I\bar{t})\}+2(\bar{x}^2+\bar{t}^2)+O(\delta),
\end{equation}
and
\begin{equation}
w(x,t)=|\delta|^{-1}\left[-\frac{2}{3}\E^{2\I\arg(\delta)}(\bar{x}+\I\bar{t})^3-2c\E^{\I\arg(\delta)}(\bar{x}+\I\bar{t})^2-2c^2(\bar{x}+\I\bar{t})\right]+O(1).
\end{equation}
Also, 
\begin{equation}
\mathbf{s}(x,t)=\begin{bmatrix}c+\E^{\I\arg(\delta)}(\bar{x}+\I\bar{t})\\-c-\E^{\I\arg(\delta)}(\bar{x}+\I\bar{t})\end{bmatrix}+O(\delta).
\end{equation}
From the B\"acklund transformation formula \eqref{eq:Backlund} we then see that $\widetilde{\psi}(x,t)\approx 1$ unless the leading term in $w(x,t)$ proportional to $|\delta|^{-1}$ is cancelled.  These terms constitute a cubic equation for $\bar{x}+\I\bar{t}$, one root of which is $\bar{x}+\I\bar{t}=0$ and the other two of which are
\begin{equation}
\bar{x}+\I\bar{t}=\frac{1}{2}c\E^{-\I\arg(\delta)}(-1\pm\I\sqrt{3}).
\end{equation}
Therefore, when $\delta$ is small, the solution $\widetilde{\psi}(x,t)$ is very close to the background solution unless $(x,t)$ lies in $O(1)$ neighborhoods of the three points $x_0+\I t_0=0,\tfrac{1}{2}c\delta^{-1}(-1\pm\I\sqrt{3})$, which are easily seen to be the vertices of a large equilateral triangle of side length $|c|\sqrt{3}/(3|\delta|)$.  Near each of these three points, the solution resembles the Peregrine solution \eqref{eq:Peregrine} located near $(x_0,t_0)$.  See Figure~\ref{f:3Peregrines}.
\begin{figure}
    \includegraphics[scale=0.35]{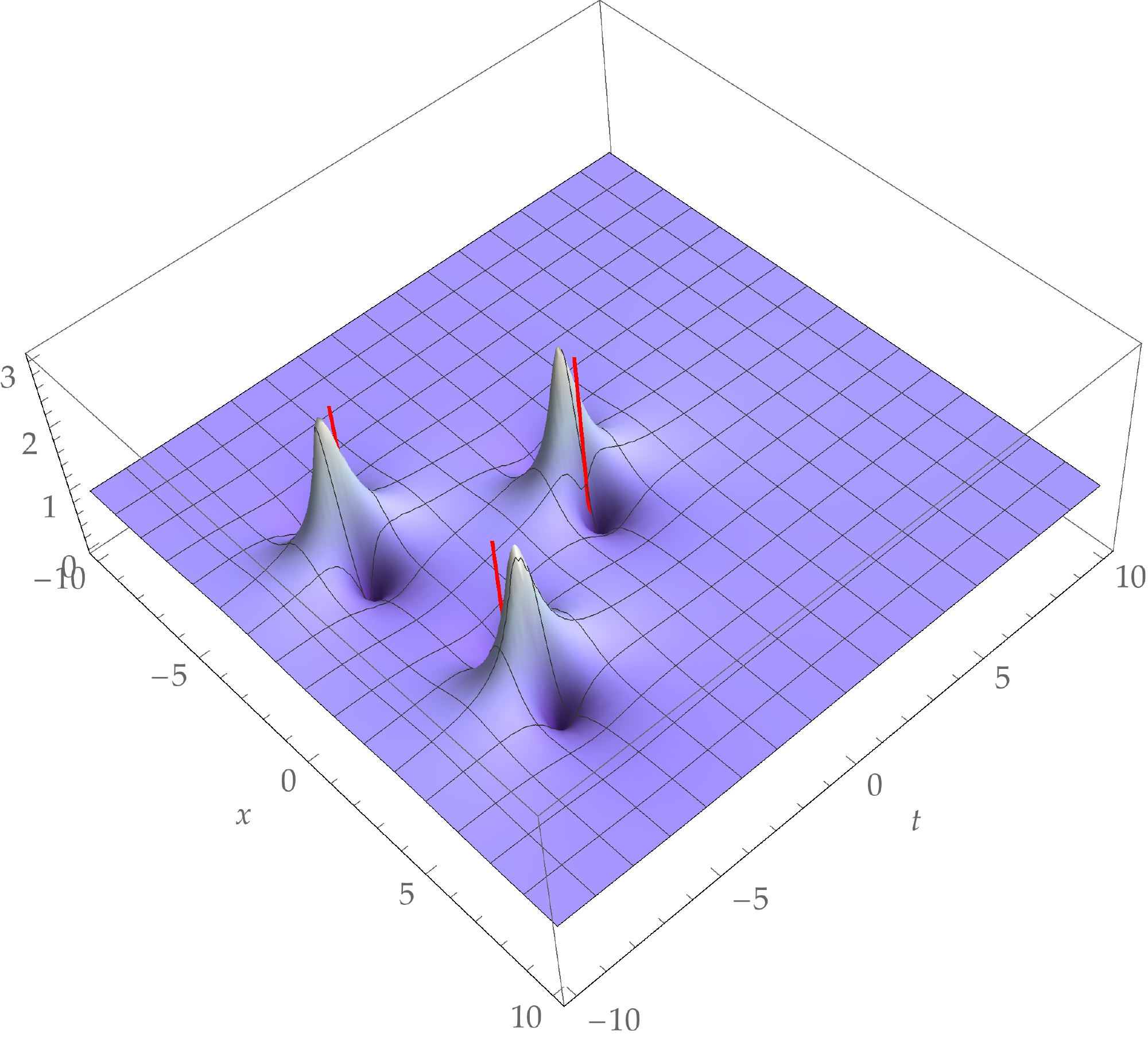}
    \caption{$|\widetilde{\psi}(x,t)|$ for a solution obtained by a single application of the Darboux transformation with the parameters $c_1=\E^{-\I\pi/12}$ and $c_2=-1$.  The red lines indicate  asymptotically valid (in the limit of large separation) predictions for the locations of the three peaks.}
    \label{f:3Peregrines}
\end{figure}
Another interesting limit corresponds to holding $c=(c_1-c_2)/2$ fixed and letting $\delta=c_1+c_2$ become large.  This brings the three vertices of the equilateral triangle, where the (approximate) Peregrine breathers are placed, in toward the origin, and suggests that the limit may result in the fusion of the three peaks into a single structure.  This limit is another application of the generalization described in Remark~\ref{r:limiting-case}, in which the parameter is given by the homogeneous coordinates $\mathbf{c}_\infty=[1:1]^\top\in\mathbb{CP}^1$.  The corresponding solution is given by 
\begin{multline}
\widetilde{\psi}(x,t)=1\\
{}+12\frac{-32 i t^5-80 t^4-16 i t^3 \left(4 x^2+1\right)-24 t^2 \left(4
    x^2+3\right)-2 i t \left(16 x^4-24 x^2-15\right)-16 x^4-24 x^2+3}{64 t^6+48 t^4
   \left(4 x^2+9\right)+12 t^2 \left(16 x^4-24 x^2+33\right)+64 x^6+48 x^4+108 x^2+9}
   \label{eq:Peregrine2}
\end{multline}
and it is plotted in Figure~\ref{f:Peregrine2}.  It is an example of a ``higher-order'' rogue wave solution of \eqref{eq:NLS}.  
\begin{figure}
    \includegraphics[scale=0.42]{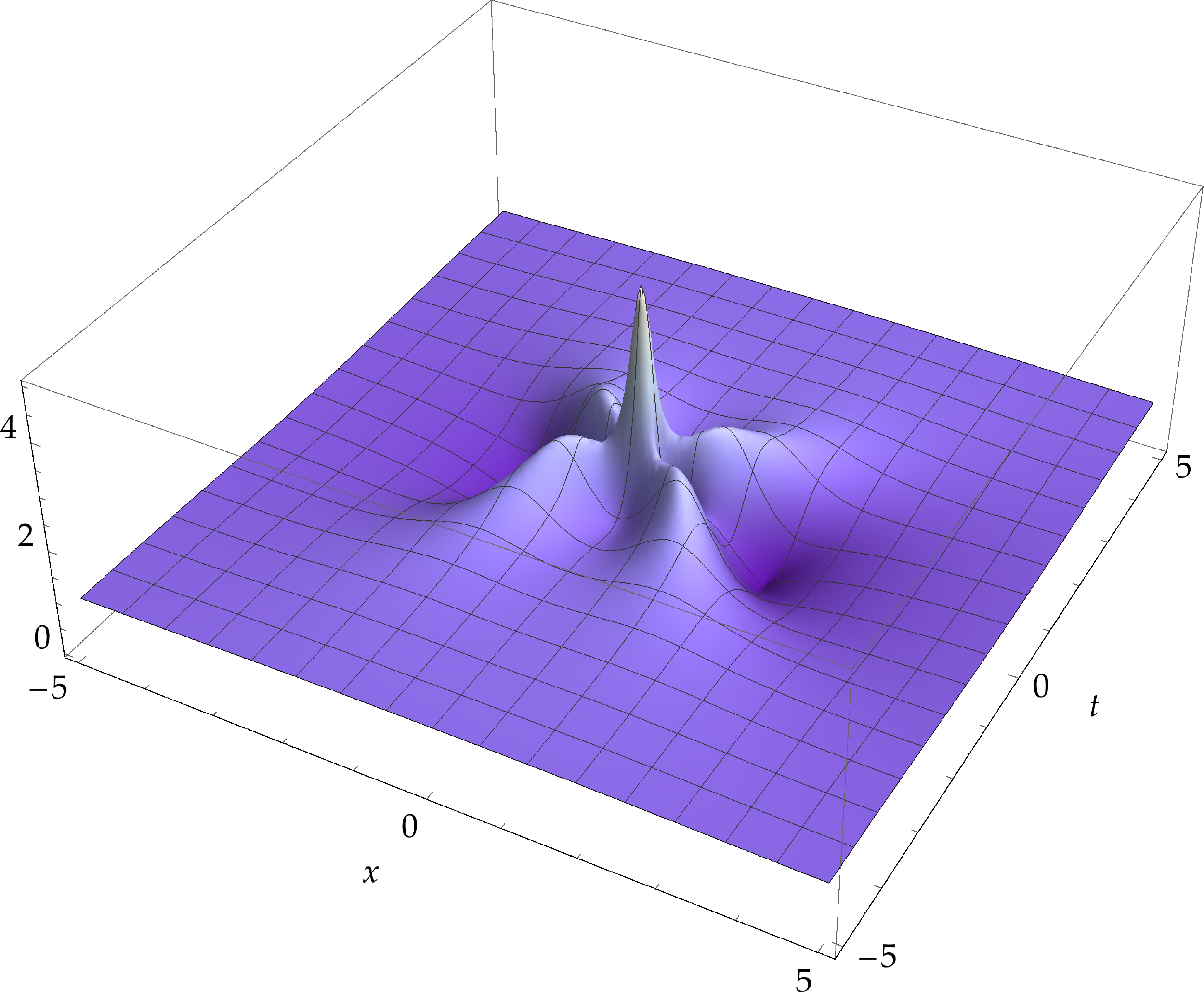}
    \caption{$|\widetilde{\psi}(x,t)|$ corresponding to a single application of the Darboux transformation to the background with parameters $\xi=\I$ and ``$\mathbf{c}=\infty$'' in the sense described in Remark~\ref{r:limiting-case} with parameter $\mathbf{c}_\infty=[1:1]^\top\in\mathbb{CP}^1$.}
    \label{f:Peregrine2}
\end{figure}

A useful property of the matrix coefficient $\mathbf{Y}(x,t)$ in the partial fraction expansion of the composite gauge transformation matrix $\mathbf{G}(\lambda;x,t)$ constructed from the background solution $\mathbf{U}_\mathrm{bg}(\lambda;x,t)$ for $\xi=\I$ is the following.  
\begin{lemma}
\label{l:IST-preserve}
Let $\mathbf{Y}(x,t)$ be defined by \eqref{eq:Ymat-formula} for $\beta=1$ with $\mathbf{s}(x,t)$, $N(x,t)$, and $w(x,t)$ being given by \eqref{eq:s-for-I}, \eqref{eq:norm-bg}, and \eqref{eq:w-bg} respectively.
Then for each $t\in\mathbb{R}$,
\begin{equation}
\mathbf{Y}(x,t)=\frac{\I\tau}{x}\begin{bmatrix}-1 & -1\\1 & 1\end{bmatrix} + O(x^{-2}),\quad x\to\pm\infty
\end{equation}
where $\tau=\tfrac{1}{4}$ if $c_1+c_2= 0$ and $\tau=\tfrac{3}{2}$ otherwise.  The same conclusion holds for $\mathbf{Y}_\infty(x,t)$ defined by \eqref{eq:Y-bar} where $\tau=\tfrac{1}{4}$ if $\mathbf{c}_\infty=[c:-c]^\top\in\mathbb{CP}^1$ and $\tau=\tfrac{3}{2}$ otherwise.
\end{lemma}
\begin{proof}
If $c_1+c_2\neq 0$, then $\mathbf{s}(x,t)=(c_1+c_2)x[1; -1]^\top + O(1)$, $N(x,t)=2|c_1+c_2|^2x^2+O(x)$, and $w(x,t)=-\tfrac{2}{3}(c_1+c_2)^2x^3+O(x^2)$ as $x\to\pm\infty$.  On the other hand, if $c_1+c_2=0$, then $\mathbf{s}(x,t)=[c;-c]^\top$, $N(x,t)=2|c|^2$, and $w(x,t)=-2c^2x+O(1)$ as $x\to\pm\infty$.  Using these in \eqref{eq:Ymat-formula} and \eqref{eq:Y-bar} completes the proof.
\end{proof}
This lemma has two related consequences. First, for any $\mathbf{c}\in\mathbb{C}^2$,
\begin{equation}
    Y_{12}(x,t)-Y_{21}(x,t)^* = O\left(x^{-2}\right),\quad x\to\pm \infty.
\end{equation}
Thus, according to \eqref{eq:new-psi}, the transformed potential $\widetilde{\psi}(x,t)$ is an $L^1$-perturbation of the background field $\psi(x,t)=\psi_\mathrm{bg}(x,t)\equiv 1$. Second, for any $\mathbf{c}\in\mathbb{C}^2$, the gauge transformation matrix $\mathbf{G}(\lambda;x,t)$ defined by \eqref{eq:G-partial-fraction} for $\xi=\I$ tends to the identity matrix as $x\to\pm\infty$ and hence it preserves the leading order behavior in the asymptotic expansion as $x\to\pm\infty$ when applied as a prefactor to a matrix function of $x$. This implies that for all $\lambda\in\Gamma$, $\mathbf{G}(\lambda;x,0)\mathbf{J}_\mathrm{bg}^\pm(\lambda;x,0)$ are precisely the Jost solution matrices associated with the transformed potential $\widetilde{\psi}(x,0)$. Since $\mathbf{J}^+_\mathrm{bg}(\lambda;x,0)=\mathbf{J}^-_\mathrm{bg}(\lambda;x,0)=\mathbf{U}_\mathrm{bg}(\lambda;x,0)$, we therefore arrive at the following result.
\begin{coro}
\label{c:RogueWaveScattering}
All solutions obtained from the background from a single application of the Darboux transformation described in Section~\ref{s:DarbouxRobustIST} for $\xi=\I$ have the same scattering matrix as the background potential:  $\mathbf{S}(\lambda;t)=\mathbb{I}$ for all $\lambda\in\Gamma\setminus\{\I,-\I\}$.
\end{coro}
In particular, this holds for the Peregrine solution $\psi_\mathrm{P}(x,t)$ defined in \eqref{eq:Peregrine}.  On the other hand, such a result is not true for $\xi\neq\pm\I$, and the Akhmediev breather case shows that it may not even be possible to define the scattering matrix after the application of a Darboux transformation.

\begin{remark}
\label{r:Peregrine-growth}
By definition (see \eqref{eq:Ymat-formula}), $\mathbf{Y}(x,t)$ always annihilates $\mathbf{s}(x,t)$.  According to \eqref{eq:sNw-Peregrine}, in the case that $c_1+c_2=0$ giving rise to the Peregrine solution, the kernel of $\mathbf{Y}(x,t)$ contains the span of $[1; -1]^\top$.  Now, the Beals-Coifman matrix for the background solution is defined by $\mathbf{U}^\mathrm{BC}_\mathrm{bg}(\lambda;x,t)=\mathbf{E}(\lambda)\E^{-\I\rho(x+\lambda t)\sigma_3}$ (cf., \eqref{eq:NZBC-asymptotic-eigenfunction}) and it satisfies
\begin{equation}
\mathbf{U}^\mathrm{BC}_\mathrm{bg}(\lambda;x,t)=n(\lambda)\left(\begin{bmatrix}1 & -1\\-1 & 1\end{bmatrix}-\I\rho(\lambda)\left(\sigma_1 +(x+\I t)\sigma_3\right)+O(\lambda-\I)\right),\quad\lambda\to\I.
\end{equation}
Therefore, applying the gauge transformation $\mathbf{G}(\lambda;x,t)$ to map $\mathbf{U}^\mathrm{BC}_{\mathrm{bg}}(\lambda;x,t)$ into its Peregrine analogue $\mathbf{U}^\mathrm{BC}_\mathrm{P}(\lambda;x,t)$, the leading term is explicitly cancelled:
\begin{equation}
\begin{split}
\mathbf{U}_\mathrm{P}^\mathrm{BC}(\lambda;x,t)&=\mathbf{G}(\lambda;x,t)\mathbf{U}^\mathrm{BC}_\mathrm{bg}(\lambda;x,t)\\
&=\left(\mathbf{Y}(x,t)(\lambda-\I)^{-1}+O(1)\right)\mathbf{U}_\mathrm{bg}^\mathrm{BC}(\lambda;x,t)\\
&=-\I n(\lambda)\rho(\lambda)(\lambda-\I)^{-1}\mathbf{Y}(x,t)\left(\sigma_1+(x+\I t)\sigma_3\right) + O(n(\lambda)),\quad\lambda\to\I.
\end{split}
\end{equation}
Since $n(\lambda)\rho(\lambda)=O((\lambda-\I)^{1/4})$ we see that for the Peregrine solution, the Beals-Coifman fundamental solution matrix exhibits a $(\lambda-\I)^{-3/4}$ singularity in all four matrix entries as $\lambda\to\I$.  Without the condition $c_1+c_2=0$, the cancellation does not occur, and therefore more generally the Beals-Coifman matrix has a $(\lambda-\I)^{-5/4}$ singularity after one application of the Darboux transformation to the background solution.  The same growth rate estimates apply to $\mathbf{M}^\mathrm{BC}(\lambda;x,t)$; in the above expansions one need only omit the term $(x+\I t)\sigma_3$ that comes from expanding the exponential factor $\E^{-\I\rho(\lambda)(x+\lambda t)\sigma_3}$.
\end{remark}

\subsection{Iteration.  Riemann-Hilbert representations for high-order rogue waves}
Since, according to Theorem~\ref{t:Darboux} the result of applying a Darboux transformation in the setting of the robust IST is to transform one Riemann-Hilbert problem into another one of the same form and at the same time to transform the \emph{solution} of the first problem into that of the second, the procedure can be repeated.  The data $(\xi,\mathbf{c})$ associated with each iteration can be related or unrelated to that of the previous step, and the basic procedure remains the same.  In this section, we show how to iterate the Darboux transformation an arbitrary number of times for certain data chosen to produce interesting solutions of the focusing NLS equation \eqref{eq:NLS} and also for which the resulting jump matrix after multiple iterations can be explicitly determined in closed form.  

The solutions we wish to obtain are those commonly referred to as the higher-order rogue wave solutions.  The simplest such solutions have already been obtained in Section~\ref{s:OneDarboux} (and in particular subsection~\ref{s:OneDarbouxAtI}), namely the Peregrine solution \eqref{eq:Peregrine} (a ``first-order'' rogue wave, see Figure~\ref{f:Peregrine}) and another solution obtained by choosing parameters so as to fuse together three copies of the elementary Peregrine solution at a single point (a ``second-order'' rogue wave, see Figure~\ref{f:Peregrine2}).  Let us denote these solutions respectively as $\psi_{\mathrm{P}_1}(x,t)$ and $\psi_{\mathrm{P}_2}(x,t)$.

Both of these solutions were obtained by a single iteration of the Darboux transformation applied to the background solution $\mathbf{U}_\mathrm{bg}(\lambda;x,t)$ with pole location $\xi=\I$ and different auxiliary data taken in the limiting sense of Remark~\ref{r:limiting-case}.  Namely, $\psi_{\mathrm{P}_1}(x,t)$ was generated from the data $\mathbf{c}_\infty=[1:-1]^\top\in\mathbb{CP}^1$ while $\psi_{\mathrm{P}_2}(x,t)$ was generated from $\mathbf{c}_\infty=[1:1]^\top\in\mathbb{CP}^1$.  To generate rogue wave solutions of arbitrary order, we simply apply these basic Darboux transformations iteratively.
\begin{definition}
Let $\mathcal{D}_\mathrm{o}$ denote the Darboux transformation associated with the data $(\xi=\I,\mathbf{c}_\infty=[1:-1]^\top)$ and let $\mathcal{D}_\mathrm{e}$ denote the Darboux transformation associated with the data $(\xi=\I,\mathbf{c}_\infty=[1:1]^\top)$.  Then the rogue wave solution of order $2n-1$ (respectively $2n$) is the solution $\psi=\psi_{\mathrm{P}_{2n-1}}(x,t)$ (respectively $\psi=\psi_{\mathrm{P}_{2n}}(x,t)$) obtained from the iterated Darboux/B\"acklund transformation $\mathcal{D}_\mathrm{o}^n$ (respectively $\mathcal{D}_\mathrm{e}^n$) applied to the background.
\label{d:HighOrderRogueWaves}
\end{definition}
Qualitatively speaking, $\psi_{\mathrm{P}_{k}}(x,t)$ represents a fusion or nonlinear superposition of
$2k-1$ copies of the elementary Peregrine solution $\psi_\mathrm{P}(x,t)$ at the same point, namely $(x,t)=(0,0)$.  It has recently been proven \cite{WangYWH17} that for any value of $k$, the maximum amplitude of this solution is $\max_{(x,t)\in\mathbb{R}^2}|\psi_{\mathrm{P}_k}(x,t)|=|\psi_{\mathrm{P}_k}(0,0)|=2k+1$.

According to Theorem~\ref{t:Darboux}, the effect of each application of $\mathcal{D}_\mathrm{o}$ or $\mathcal{D}_\mathrm{e}$ is to produce a new factor in the jump matrix on the upper/lower semicircles $\Sigma_+$/$\Sigma_-$ of the jump contour, and at each iteration the factor is calculated by evaluating the gauge transformation $\mathbf{G}(\lambda;x,t)=\mathbb{I}+(\lambda-\I)^{-1}\mathbf{Y}_\infty(x,t)+(\lambda+\I)^{-1}\mathbf{Z}_\infty(x,t)$ at $(x,t)=(0,0)$ \emph{after} computing $\mathbf{Y}_\infty(x,t)$ and $\mathbf{Z}_\infty(x,t)$ from the fundamental solution matrix obtained from the previous iteration.  Let
$\mathbf{M}_\mathrm{o}^{[n]}(\lambda;x,t)$ and $\mathbf{M}_\mathrm{e}^{[n]}(\lambda;x,t)$ denote the solution matrices of \rhref{rhp:M-NZBC} after $n$ applications of $\mathcal{D}_\mathrm{o}$ and $\mathcal{D}_\mathrm{e}$ respectively.  Thus $\mathbf{M}_\mathrm{o}^{[0]}(\lambda;x,t)=\mathbf{M}_\mathrm{e}^{[0]}(\lambda;x,t)$ is the solution of \rhref{rhp:M-NZBC} formulated for the background solution $\psi=\psi_\mathrm{bg}(x,t)\equiv 1$.  Let $\mathbf{G}_\mathrm{o}^{[n]}(\lambda;x,t)$ and $\mathbf{G}_\mathrm{e}^{[n]}(\lambda;x,t)$ denote the gauge matrices calculated for the application of $\mathcal{D}_\mathrm{o}$ to $\mathbf{U}_{\mathrm{o}}^{[n]}(\lambda;x,t):=\mathbf{M}_{\mathrm{o}}^{[n]}(\lambda;x,t)\E^{-\I\rho(\lambda)(x+\lambda)t\sigma_3}$ and of $\mathcal{D}_\mathrm{e}$ to $\mathbf{U}_{\mathrm{e}}^{[n]}(\lambda;x,t):=\mathbf{M}_\mathrm{e}^{[n]}(\lambda;x,t)\E^{-\I\rho(\lambda)(x+\lambda t)\sigma_3}$ respectively.  Thus, 
\begin{multline}
\mathbf{U}^{[n]}_{\mathrm{o}/\mathrm{e}}(\lambda;x,t)\\
=\begin{cases}
\mathbf{G}^{[n-1]}_{\mathrm{o}/\mathrm{e}}(\lambda;x,t)\cdots\mathbf{G}^{[0]}_{\mathrm{o}/\mathrm{e}}(\lambda;x,t)\mathbf{U}_\mathrm{bg}(\lambda;x,t),&\;\lambda\in D_+\cup D_-\\
\mathbf{G}^{[n-1]}_{\mathrm{o}/\mathrm{e}}(\lambda;x,t)\cdots\mathbf{G}^{[0]}_{\mathrm{o}/\mathrm{e}}(\lambda;x,t)\mathbf{U}_\mathrm{bg}(\lambda;x,t)\mathbf{G}^{[0]}_{\mathrm{o}/\mathrm{e}}(\lambda;0,0)^{-1}\cdots\mathbf{G}^{[n-1]}_{\mathrm{o}/\mathrm{e}}(\lambda;0,0)^{-1},&\;\lambda\in D_0
\end{cases}
\end{multline}
and $\mathbf{M}_{\mathrm{o}/\mathrm{e}}^{[n]}(\lambda;x,t)=\mathbf{U}_{\mathrm{o}/\mathrm{e}}^{[n]}(\lambda;x,t)\E^{\I\rho(\lambda)(x+\lambda t)\sigma_3}$.  The Riemann-Hilbert problem whose unique solution is $\mathbf{M}_{\mathrm{o}/\mathrm{e}}^{[n]}(\lambda;x,t)$ is then the following.
\begin{rhp}
  Let $n=0,1,2,3,\dots$, and seek a $2\times 2$ matrix function $\mathbf{M}(\lambda;x,t)=\mathbf{M}_{\mathrm{o}/\mathrm{e}}^{[n]}(\lambda;x,t)$ that has the following properties:
  \begin{itemize}
    \item \textbf{Analyticity:} $\mathbf{M}(\lambda;x,t)$ is analytic for $\lambda$ in the exterior of the disk $D_0$ and for $\lambda\in D_0 \setminus\Sigma_\mathrm{c}$\,.
    \item \textbf{Jump Condition:} $\mathbf{M}(\lambda;x,t)$ takes continuous boundary values $\mathbf{M}_\pm(\lambda;x,t)$ on $\Sigma_0\cup\Sigma_\mathrm{c}$, and they are related by a jump condition of the form $\mathbf{M}_+(\lambda;x,t)=\mathbf{M}_-(\lambda;x,t)\mathbf{V}_{\mathrm{o}/\mathrm{e}}^{[n]}(\lambda;x,t)$ for $\lambda\in\Sigma\cup\Sigma_\mathrm{c}$, where
\begin{equation}
\mathbf{V}_{\mathrm{o}/\mathrm{e}}^{[n]}(\lambda;x,t):=\E^{2\I\rho_+(\lambda)(x+\lambda t)\sigma_3},\quad
\lambda\in\Sigma_\mathrm{c},
\end{equation}
\begin{equation}
\mathbf{V}_{\mathrm{o}/\mathrm{e}}^{[n]}(\lambda;x,t):=\E^{-\I\rho(\lambda)(x+\lambda t)\sigma_3}
\mathbf{G}_{\mathrm{o}/\mathrm{e}}^{[n-1]}(\lambda;0,0)\cdots\mathbf{G}_{\mathrm{o}/\mathrm{e}}^{[0]}(\lambda;0,0)\mathbf{E}(\lambda)\E^{\I\rho(\lambda)(x+\lambda t)\sigma_3},\quad\lambda\in\Sigma_+,
\label{eq:iterated-jump-plus}
\end{equation}
and
\begin{multline}
\mathbf{V}_{\mathrm{o}/\mathrm{e}}^{[n]}(\lambda;x,t):=\E^{-\I\rho(\lambda)(x+\lambda t)\sigma_3}
\mathbf{E}(\lambda)^{-1}\mathbf{G}_{\mathrm{o}/\mathrm{e}}^{[0]}(\lambda;0,0)^{-1}\cdots\mathbf{G}_{\mathrm{o}/\mathrm{e}}^{[n-1]}(\lambda;0,0)^{-1}\E^{\I\rho(\lambda)(x+\lambda t)\sigma_3},\\\lambda\in\Sigma_-.
\label{eq:iterated-jump-minus}
\end{multline}
    \item \textbf{Normalization:} $\lim_{\lambda\to\infty} \mathbf{M}(\lambda;x,t) = \mathbb{I}$.
  \end{itemize}
  \label{rhp:M-Rogue-Waves}
\end{rhp}
Note that this Riemann-Hilbert problem, like that for the background solution $\psi_\mathrm{bg}(x,t)\equiv 1$, has no jump across the real axis.  The rogue waves themselves are obtained from the solution in the usual way:
\begin{equation}
\psi_{\mathrm{P}_{2n-1}}(x,t)=2i\lim_{\lambda\to\infty}\lambda M^{[n]}_{\mathrm{o},12}(\lambda;x,t)\quad\text{and}\quad
\psi_{\mathrm{P}_{2n}}(x,t)=2i\lim_{\lambda\to\infty}\lambda M^{[n]}_{\mathrm{e},12}(\lambda;x,t), \quad n=1,2,3,\dots.
\end{equation}
Although the gauge transformation matrix $\mathbf{G}^{[k]}_{\mathrm{o}/\mathrm{e}}(\lambda;x,t)$ is not generally equal to its predecessor $\mathbf{G}^{[k-1]}_{\mathrm{o}/\mathrm{e}}(\lambda;x,t)$ because the former has to be calculated using the values of the latter near $\xi=\I$, a remarkable simplification occurs for $(x,t)=(0,0)$, and this is enough to make the evaluation of the jump matrix in \rhref{rhp:M-Rogue-Waves} completely explicit.
\begin{prop}
  For any integer $n \geq 0$, $\mathbf{G}_{\mathrm{o}/\mathrm{e}}^{[n]}(\lambda;0,0)= \mathbf{G}_{\mathrm{o}/\mathrm{e}}^{[0]}(\lambda;0,0)$.
  \label{p:iterated-jumps}
\end{prop}
\begin{proof}
Let $\mathbf{Y}^{[n]}_{\infty,\mathrm{o}/\mathrm{e}}(x,t)$ denote the matrix defined by \eqref{eq:Y-bar}
for the data $\xi=\I$ and $\mathbf{c}_\infty=[1:-1]\in\mathbb{CP}^1$ (for $\mathbf{Y}^{[n]}_{\infty,\mathrm{o}}(x,t)$) or $\mathbf{c}_\infty=[1:1]\in\mathbb{CP}^1$ (for $\mathbf{Y}^{[n]}_{\infty,\mathrm{e}}(x,t)$).
  For any integer $n \geq 0$, the $(x,t)$-dependence of $\mathbf{Y}_{\infty,\mathbf{o}/\mathbf{e}}^{[n]}(\lambda;x,t)$ is encoded only via the quantities
  \begin{equation}
    \begin{aligned}
      \mathbf{s}_\infty(\I;x,t) &= \mathbf{U}_{\mathrm{o}/\mathrm{e}}^{[n]}(\I;x,t)\mathbf{c}_\infty\\
      N_\infty(x,t) &=\|  \mathbf{s}(\I;x,t) \|^2 =\mathbf{c}_\infty^\dagger\mathbf{U}_{\mathrm{o}/\mathrm{e}}^{[n]}(\I;x,t)^\dagger\mathbf{U}_{\mathrm{o}/\mathrm{e}}^{[n]}(\I;x,t)\mathbf{c}_\infty,\\
      w_\infty(x,t) &= \mathbf{c}_\infty^\top\mathbf{U}_{\mathrm{o}/\mathrm{e}}^{[n]}(\I;x,t)^\top\sigma_2\mathbf{U}_{\mathrm{o}/\mathrm{e}}^{[n]\prime}(\I;x,t)\mathbf{c}_\infty.
    \end{aligned}
  \end{equation}
Recall that by definition $\mathbf{U}_{\mathrm{o}/\mathrm{e}}^{[n]}(\lambda;0,0)=\mathbb{I}$ for all $\lambda\in D_0$, and therefore also $\mathbf{U}_{\mathrm{o}/\mathrm{e}}^{[n]\prime}(\I;0,0)=\mathbf{0}$. Thus, evaluating the quantities above at $(x,t)=(0,0)$ gives
  \begin{equation}
    \begin{aligned}
      \mathbf{s}_\infty(\I;0,0) &= \mathbf{c}_\infty,\\
      N_\infty(0,0) &=\mathbf{c}_\infty^\dagger\mathbf{c}_\infty,\\
      w_\infty(0,0) &= 0,
    \end{aligned}
    \label{eq:zero-hour}
  \end{equation}
independent of $n$, and these coincide with the values of $\mathbf{s}_\infty(\I;0,0)$, $N_\infty(0,0)$, and $w_\infty(0,0)$ for the gauge transformation matrix $\mathbf{G}_{\mathrm{o}/\mathrm{e}}^{[0]}(\lambda;x,t)$ built from the Riemann-Hilbert matrix $\mathbf{U}_{\mathrm{o}/\mathrm{e}}^{[0]}(\lambda;x,t)=\mathbf{U}_\mathrm{bg}(\lambda;x,t)$ for the background field. This proves the claim.
\end{proof}

\begin{coro}
\label{c:simple-jumps}
The product of factors appearing in the jump matrices in \eqref{eq:iterated-jump-plus}--\eqref{eq:iterated-jump-minus} is given by 
\begin{equation}
\mathbf{G}^{[n-1]}_{\mathrm{o}/\mathrm{e}}(\lambda;0,0)\cdots\mathbf{G}^{[0]}_{\mathrm{o}/\mathrm{e}}(\lambda;0,0)=\mathbf{G}_{\mathrm{o}/\mathrm{e}}^{[0]}(\lambda;0,0)^n,
\label{eq:power}
\end{equation}
where
\begin{equation}
\mathbf{G}_{\mathrm{o}/\mathrm{e}}^{[0]}(\lambda;0,0)=\mathbb{I}+\frac{\mathbf{Y}^{[0]}_{\infty,\mathrm{o}/\mathrm{e}}(0,0)}{\lambda-\I}+\frac{\mathbf{Z}^{[0]}_{\infty,\mathrm{o}/\mathrm{e}}(0,0)}{\lambda+\I}
\label{eq:G0-representation}
\end{equation}
in which
\begin{equation}
\mathbf{Y}^{[0]}_{\infty,\mathrm{o}}(0,0)=-\mathbf{Z}^{[0]}_{\infty,\mathrm{e}}(0,0)=\mathbf{H}:=\I\begin{bmatrix}1 & 1\\1 & 1\end{bmatrix}
\label{eq:coefficients-H}
\end{equation}
and
\begin{equation}
\mathbf{Y}^{[0]}_{\infty,\mathrm{e}}(0,0)=-\mathbf{Z}^{[0]}_{\infty,\mathrm{o}}(0,0)=\mathbf{K}:=\I\begin{bmatrix}1 & -1\\-1 & 1\end{bmatrix}.
\label{eq:coefficients-K}
\end{equation}
\end{coro}
\begin{proof}
The identity \eqref{eq:power} follows from Proposition~\ref{p:iterated-jumps}.
It only remains to prove the formulae \eqref{eq:coefficients-H}--\eqref{eq:coefficients-K}; but these follow from 
the definitions \eqref{eq:Y-bar}--\eqref{eq:Z-bar} together with \eqref{eq:zero-hour} and the values of $\mathbf{c}_\infty=[1:-1]$ (for case o) and $\mathbf{c}_\infty=[1:1]$ (for case e).
\end{proof}

\begin{remark}
The rogue wave solutions of \eqref{eq:NLS} of arbitrary order have thus been encoded in the solution of a simple Riemann-Hilbert problem with a jump matrix depending explicitly on $(x,t)\in\mathbb{R}^2$ and the order $k$ proportional to $n$, which appears as an exponent.  This kind of problem is likely well-suited to asymptotic analysis by the Deift-Zhou steepest descent method \cite{Deift1993a} to determine asymptotic properties of rogue waves in the limit of large order.  
For instance, the plots in Figures~\ref{f:odd-rogue-waves}, \ref{f:even-rogue-waves}, and \ref{f:8th-rogue-wave} (see \cite{Dubard2010b,Gaillard2013,Gaillard2014a,HeZWPF13} for similar plots) suggest the following questions:
\begin{itemize}
\item Can one describe the asymptotic properties of the peaks and zeros of $\psi_k(x,t)$ for large $k$?  In particular, what are the asymptotics of the extreme zeros?
\item Can one describe the spatio-temporal pattern of rapid oscillations of $\psi_k(x,t)$ in the limit of large $k$ by a proper multiple-scale formula accounting for a slowly-varying envelope modulating a rapidly-varying carrier wave?  What is the nature of the carrier wave?  Is it trigonometric, elliptic, or characterized by some other special function in the limit $k\to\infty$?
\end{itemize}
This is work in progress.
\end{remark}

\subsection{From analytic to algebraic representations of high-order rogue waves}
Although \rhref{rhp:M-Rogue-Waves} is an analytical characterization of arbitrary-order rogue waves, it yields an algebraic representation as well, which is more in line with what is in the literature \cite{AkhmedievAS09,Ankiewicz2010,Gaillard2013,Gaillard2014a,Guo2012a}, and which can lead to compact formulae for solutions in terms of determinants.  While perhaps less useful for determining properties of high-order rogue waves due to combinatorial complexity, such formulae are effective for the calculation of rogue wave solutions of low order.  Fix a value of $n=1,2,3,\dots$ (everything below depends on $n$ but we will not systematically indicate this dependence when making new definitions going forward).  To convert \rhref{rhp:M-Rogue-Waves} into a finite-dimensional linear algebraic system, first observe that for $\lambda\in\mathbb{C}\setminus D_0$, $\mathbf{M}^{[n]}_{\mathrm{o}/\mathrm{e}}(\lambda;x,t)=\mathbf{\Pi}_{\mathrm{o}/\mathrm{e}}(\lambda;x,t)\mathbf{M}^{[0]}(\lambda;x,t)=\mathbf{\Pi}_{\mathrm{o}/\mathrm{e}}(\lambda;x,t)\mathbf{E}(\lambda)$, where $\mathbf{\Pi}_{\mathrm{o}/\mathrm{e}}(\lambda;x,t)$ denotes the ordered product of the gauge transformation matrices:
\begin{equation}
\mathbf{\Pi}_{\mathrm{o}/\mathrm{e}}(\lambda;x,t):=\mathbf{G}^{[n-1]}_{\mathrm{o}/\mathrm{e}}(\lambda;x,t)\cdots\mathbf{G}^{[1]}_{\mathrm{o}/\mathrm{e}}(\lambda;x,t)\mathbf{G}^{[0]}_{\mathrm{o}/\mathrm{e}}(\lambda;x,t).
\end{equation}
As a product of $n$ matrix factors, each having simple poles at $\pm\I$ as its only singularities and decaying to $\mathbb{I}$ as $\lambda\to\infty$, $\mathbf{\Pi}_{\mathrm{o}/\mathrm{e}}(\lambda;x,t)$ also decays to $\mathbb{I}$ for large $\lambda$ and has poles of order $n$ at $\pm\I$ and is otherwise analytic.  Hence, it necessarily has a finite partial fraction expansion of the form
\begin{equation}
\mathbf{\Pi}_{\mathrm{o}/\mathrm{e}}(\lambda;x,t) = \mathbb{I} + \sum\limits_{k=1}^n \frac{\mathbf{A}^+_{\mathrm{o}/\mathrm{e},k}(x,t)}{(\lambda-\I)^k}+\frac{\mathbf{A}^-_{\mathrm{o}/\mathrm{e},k}(x,t)}{(\lambda + \I)^k}.
\label{eq:Pi-partial-fractions}
\end{equation}
Next, taking into account Corollary~\ref{c:simple-jumps}, the jump conditions of \rhref{rhp:M-Rogue-Waves} together imply that
\begin{equation}
\mathbf{\Pi}_{\mathrm{o}/\mathrm{e}}(\lambda;x,t)\mathbf{E}(\lambda)=\mathbf{U}^{[n],\mathrm{in}}_{\mathrm{o}/\mathrm{e}}(\lambda;x,t)\mathbf{G}^{[0]}_{\mathrm{o}/\mathrm{e}}(\lambda;0,0)^n\mathbf{E}(\lambda)\E^{\I\rho(\lambda)(x+\lambda t)\sigma_3}
\label{eq:jump-reinterpret-1}
\end{equation}
holds for $|\lambda|=r$, where $\mathbf{U}^{[n],\mathrm{in}}_{\mathrm{o}/\mathrm{e}}(\lambda;x,t)=\mathbf{M}^{[n],\mathrm{in}}_{\mathrm{o}/\mathrm{e}}(\lambda;x,t)\E^{-\I\rho(\lambda)(x+\lambda t)\sigma_3}$ is a matrix function analytic for $\lambda\in D_0$.  Therefore \eqref{eq:jump-reinterpret-1} can be recast as the identity (the exponent $-n$ indicates the $n^\mathrm{th}$ power of the inverse matrix)
\begin{equation}
\mathbf{\Pi}_{\mathrm{o}/\mathrm{e}}(\lambda;x,t)\cdot\left(\mathbf{E}(\lambda)\E^{-\I\rho(\lambda)(x+\lambda t)\sigma_3}\mathbf{E}(\lambda)^{-1}\right)\cdot\mathbf{G}^{[0]}_{\mathrm{o}/\mathrm{e}}(\lambda;0,0)^{-n} = \mathbf{U}^{[n],\mathrm{in}}_{\mathrm{o}/\mathrm{e}}(\lambda;x,t),\quad |\lambda|=r.
\label{eq:jump-reinterpret-2}
\end{equation}
The expression in parentheses on the left-hand side is precisely the matrix $\mathbf{U}^\mathrm{in}_\mathrm{bg}(\lambda;x,t)$ defined in \eqref{eq:Uin-bg-1}; it is an entire function of $\lambda$ that is written in explicit form in \eqref{eq:Uin-bg-2}.  Since the right-hand side of \eqref{eq:jump-reinterpret-2} admits analytic continuation to the domain $D_0$, the same must be true of the left-hand side.  But $\mathbf{U}^\mathrm{in}_\mathrm{bg}(\lambda;x,t)$ is a known entire function, and $\mathbf{G}_{\mathrm{o}/\mathrm{e}}^{[0]}(\lambda;0,0)^{-n}$ can be regarded as known using Corollary~\ref{c:simple-jumps} and will be shortly shown to have poles of order $n$ at $\lambda=\pm\I$.  Therefore, demanding analyticity of the left-hand side at $\lambda=\pm\I$ imposes conditions on the unknown coefficients in the partial fraction expansion \eqref{eq:Pi-partial-fractions}.  These conditions constitute an algebraic representation of the rogue wave solution.

To make these observations into an effective procedure,
first note that according to \eqref{eq:coefficients-H}--\eqref{eq:coefficients-K}, the following identities are obvious:
\begin{equation}
\mathbf{H}\mathbf{K} =  \mathbf{K}\mathbf{H} = \mathbf{0}
\label{eq:mixed-terms-vanish}
\end{equation}
and 
\begin{equation}
\mathbf{H}^2 = 2\I \mathbf{H} \quad \text{and}\quad
\mathbf{K}^2 = 2\I \mathbf{K}.
\label{eq:Cayley-Hamilton}
\end{equation}
The $n^\mathrm{th}$ power on the right-hand side of \eqref{eq:power} can now be explicitly computed from the representation \eqref{eq:G0-representation}.  Indeed \eqref{eq:mixed-terms-vanish} eliminates all mixed products and \eqref{eq:Cayley-Hamilton} reduces all matrix powers to scalar multiples:
\begin{equation}
\label{eq:G-odd-power}
\begin{aligned}
  \mathbf{G}^{[0]}_\mathrm{o}(\lambda;0,0)^n &= \mathbb{I}+\sum_{k=1}^n \binom{n}{k}\left[\frac{\mathbf{H}^k}{(\lambda-\I)^k} + \frac{(-\mathbf{K})^k}{(\lambda+\I)^k}\right] \\
  &=\mathbb{I}+\sum_{k=1}^n \binom{n}{k}(2\I)^{k-1}\left[\frac{\mathbf{H}}{(\lambda-\I)^k} + \frac{(-1)^{k}\mathbf{K}}{(\lambda+\I)^k}\right],
  \end{aligned}
\end{equation}
\begin{equation}
  \label{eq:G-even-power}
  \begin{aligned}
   \mathbf{G}^{[0]}_\mathrm{e}(\lambda;0,0)^n &= \mathbb{I}+\sum_{k=1}^n \binom{n}{k}\left[\frac{\mathbf{K}^k}{(\lambda-\I)^k} + \frac{(-\mathbf{H})^k}{(\lambda+\I)^k}\right] \\
   &=\mathbb{I}+\sum_{k=1}^n \binom{n}{k}(2\I)^{k-1}\left[\frac{\mathbf{K}}{(\lambda-\I)^k} + \frac{(-1)^{k}\mathbf{H}}{(\lambda+\I)^k}\right],
   \end{aligned}
\end{equation}
Moreover, \eqref{eq:mixed-terms-vanish} and \eqref{eq:Cayley-Hamilton} together imply the identity
\begin{equation}
\mathbf{G}^{[0]}_\mathrm{o}(\lambda;0,0)^{-1}=\mathbf{G}^{[0]}_\mathrm{e}(\lambda;0,0),
\label{eq:G-inverse-relation}
\end{equation}
which further implies $\mathbf{G}_{\mathrm{o}/\mathrm{e}}^{[0]}(\lambda;0,0)^{-n}=\mathbf{G}_{\mathrm{e}/\mathrm{o}}^{[0]}(\lambda;0,0)^n$.  Thus the final factor on the left-hand side of the relation \eqref{eq:jump-reinterpret-2} is given by either \eqref{eq:G-odd-power} or \eqref{eq:G-even-power} and is therefore an explicit meromorphic function with poles of order $n$ at $\lambda=\pm\I$, written in partial fraction expansion form.
\begin{remark}
The relation \eqref{eq:G-inverse-relation} has an interesting consequence. If an application of $\mathcal{D}_\mathrm{o}$ is followed by an application of  $\mathcal{D}_\mathrm{e}$ or vice versa, then the factors introduced in the jump matrix on $\Sigma_0$ by these transformations are inverses of each other. Therefore, the Riemann-Hilbert problem remains unchanged. By uniqueness, we conclude that $\mathcal{D}_\mathrm{o}\circ \mathcal{D}_\mathrm{e}=\mathscr{I}$ and the corresponding B\"acklund transformations are inverses of each other.
\end{remark}

We use this information to obtain the coefficients in the partial fraction expansion \eqref{eq:Pi-partial-fractions} in two steps:
\begin{itemize}
\item
Denote by $\mathbf{P}_{\mathrm{o}/\mathrm{e}}(\lambda;x,t)$ the product of the first two factors on the left-hand side of \eqref{eq:jump-reinterpret-2}:  $\mathbf{P}_{\mathrm{o}/\mathrm{e}}(\lambda;x,t):=\mathbf{\Pi}_{\mathrm{o}/\mathrm{e}}(\lambda;x,t)\mathbf{U}_\mathrm{bg}^{\mathrm{in}}(\lambda;x,t)$.  This function has poles of order $n$ at $\lambda=\pm\I$ and hence has convergent Laurent expansions of the form:
\begin{equation}
\mathbf{P}_{\mathrm{o}/\mathrm{e}}(\lambda;x,t)=\sum_{j=-n}^\infty
\mathbf{P}_{\mathrm{o}/\mathrm{e},j}^{\pm}(x,t)(\lambda \mp\I)^j,\quad |\lambda\mp\I|<2.
\label{eq:P-Laurent}
\end{equation}
Then the condition that the left-hand side of \eqref{eq:jump-reinterpret-2} have removable singularities at $\lambda=\I$ and $\lambda=-\I$ implies the following $4n$ necessary vector-valued conditions on the coefficient functions of the series \eqref{eq:P-Laurent}:
\begin{equation}
\mathbf{P}_{\mathrm{o}/\mathrm{e},j}^{+}(x,t)\mathbf{c}_{\infty,\mathrm{o}/\mathrm{e}} = \mathbf{0}\quad\text{and}\quad \mathbf{P}^{-}_{\mathrm{o}/\mathrm{e},j}(x,t)\sigma_3\mathbf{c}_{\infty,\mathrm{o}/\mathrm{e}} = \mathbf{0},\quad -n\leq j \leq n-1,
\label{eq:P-hom}
\end{equation}
where $\mathbf{c}_{\infty,\mathrm{o}}=[1:-1]^\top\in\mathbb{C}\mathbb{P}^1$ and $\mathbf{c}_{\infty,\mathrm{e}}=[1:1]^\top\in\mathbb{C}\mathbb{P}^1$. In each case, the first $n$ equations are obvious, but to derive the second $n$ equations one makes repeated use of the identities \eqref{eq:mixed-terms-vanish} and \eqref{eq:Cayley-Hamilton}.
\item Now substitute into \eqref{eq:P-hom} the Laurent coefficients of $\mathbf{P}_{\mathrm{o}/\mathrm{e}}(\lambda;x,t)=\mathbf{\Pi}_{\mathrm{o}/\mathrm{e}}(\lambda;x,t)\mathbf{U}_\mathrm{bg}^{\mathrm{in}}(\lambda;x,t)$, which can be explicitly written in terms of the unknown coefficients in the partial fraction expansion \eqref{eq:Pi-partial-fractions} and the Taylor coefficients of the known analytic function $\mathbf{U}^\mathrm{in}_\mathrm{bg}(\lambda;x,t)$ about $\lambda=\pm\I$:
\begin{equation} 
\mathbf{U}_\mathrm{bg}^{\mathrm{in}}(\lambda;x,t) = \sum\limits_{j=0}^{\infty} \mathbf{D}^\pm_j(x,t)(\lambda\mp\I)^j.
\label{eq:U-power}
\end{equation}
Defining associated vectors by 
\begin{equation}
\mathbf{w}^+_{\mathrm{o}/\mathrm{e},k}(x,t):=\mathbf{D}_k^+(x,t)\mathbf{c}_{\infty,\mathrm{o}/\mathrm{e}}\quad\text{and}\quad \mathbf{w}^-_{\mathrm{o}/\mathrm{e},k}(x,t):=\mathbf{D}_k^-(x,t)\sigma_3\mathbf{c}_{\infty,\mathrm{e}/\mathrm{o}},\quad k=0,1,2,3,\dots,
\end{equation} 
the equations \eqref{eq:P-hom} for indices $-n\le j\le -1$ imply the $2n$ vector equations
\begin{align}
\label{eq:hom-A}
\sum\limits_{k=0}^j \mathbf{A}^+_{\mathrm{o}/\mathrm{e},n-j+k}(x,t)\mathbf{w}^{+}_{\mathrm{o}/\mathrm{e},k}(x,t) &=0,\quad 0\leq j \leq n-1,
\\
\label{eq:hom-B}
\sum\limits_{k=0}^j \mathbf{A}^-_{\mathrm{o}/\mathrm{e},n-j+k}(x,t) \mathbf{w}^-_{\mathrm{o}/\mathrm{e},k}(x,t)&=0,\quad 0\leq j \leq n-1.
\end{align}
To express the remaining equations in \eqref{eq:P-hom}, first set
\begin{equation}
\gamma_{km}:=\frac{(-1)^k}{(2\I)^{m+k}}\binom{m+k-1}{k},\quad m=1,2,3,\dots,\quad k=0,1,2,3,\dots,
\end{equation}
so that
\begin{equation}
\frac{1}{(\lambda+\I)^m}=\sum\limits_{k=0}^{\infty}\gamma_{km}(\lambda-\I)^k\,,\quad \frac{1}{(\lambda-\I)^m}=\sum\limits_{k=0}^{\infty}(-1)^{m+k}\gamma_{km}(\lambda+\I)^k\,.
\end{equation}
Introducing the auxiliary unknown matrices
\begin{equation}
{\boldsymbol{\Gamma}}_{\mathrm{o}/\mathrm{e},k}^+(x,t)\defeq\sum\limits_{m=1}^n {\gamma}_{km}\mathbf{A}^-_{\mathrm{o}/\mathrm{e},m}(x,t), \quad {\boldsymbol{\Gamma}}^-_{\mathrm{o}/\mathrm{e},k}(x,t)\defeq\sum\limits_{m=1}^n (-1)^{m+k}{\gamma}_{km}\mathbf{A}^+_{\mathrm{o}/\mathrm{e},m}(x,t),
\label{eq:gamma-matrices}
\end{equation}
the equations \eqref{eq:P-hom} for indices $0\le j\le n-1$ take the form
\begin{align}
\label{eq:inhom-A-Gamma}
\sum\limits_{k=0}^j \boldsymbol{\Gamma}^+_{\mathrm{o}/\mathrm{e},j-k}(x,t)\mathbf{w}^{+}_{\mathrm{o}/\mathrm{e},k} (x,t)+ \sum\limits_{k=1}^n \mathbf{A}^+_{\mathrm{o}/\mathrm{e},k}(x,t)\mathbf{w}^{+}_{\mathrm{o}/\mathrm{e},j+k}(x,t)&=-\mathbf{w}^{+}_{\mathrm{o}/\mathrm{e},j}(x,t),\quad 0\leq j \leq n-1\\
\label{eq:inhom-B-Gamma}
 \sum\limits_{k=0}^j \boldsymbol{\Gamma}^-_{\mathrm{o}/\mathrm{e},j-k}(x,t)\mathbf{w}^{-}_{\mathrm{o}/\mathrm{e},k} (x,t)+\sum\limits_{k=1}^n \mathbf{A}^-_{\mathrm{o}/\mathrm{e},k}(x,t)\mathbf{w}^{-}_{\mathrm{o}/\mathrm{e},j+k}(x,t)&=-\mathbf{w}^{-}_{\mathrm{o}/\mathrm{e},j}(x,t),\quad 0\leq j \leq n-1.
\end{align}
Eliminating $\mathbf{\Gamma}^+_{\mathrm{o}/\mathrm{e}}(x,t)$ and $\mathbf{\Gamma}^-_{\mathrm{o}/\mathrm{e}}(x,t)$ using \eqref{eq:gamma-matrices}, it is clear that 
equations \eqref{eq:hom-A}, \eqref{eq:hom-B}, \eqref{eq:inhom-A-Gamma}, and \eqref{eq:inhom-B-Gamma} 
constitute a square inhomogeneous linear system of dimension $8n\times 8n$ governing the $8n$ entries of the matrix coefficients in the partial fraction expansion \eqref{eq:Pi-partial-fractions}.  
\end{itemize}
 
According to \eqref{eq:NLS-recovery} the solution of the NLS equation \eqref{eq:NLS} stemming from \rhref{rhp:M-Rogue-Waves} is
\begin{equation}
\begin{split}
\psi(x,t)&=2\I\lim_{\lambda\to\infty}\lambda M^{[n]}_{\mathrm{o/e},12}(\lambda;x,t)\\
&=2\I\lim_{\lambda\to\infty}\lambda\left(\mathbf{\Pi}_{\mathrm{o}/\mathrm{e}}(\lambda;x,t)\mathbf{M}^{[0]}(\lambda;x,t)\right)_{12}\\
&=1+2\I \left(\mathbf{A}^+_{\mathrm{o}/\mathrm{e},1}(x,t)+\mathbf{A}^-_{\mathrm{o}/\mathrm{e},1}(x,t)\right)_{12}
\end{split}
\end{equation}
where we have used that the second column of $\mathbf{M}^{[0]}(\lambda;x,t)=\mathbf{E}(\lambda)$ has the expansion $[0; 1]^\top + \lambda^{-1}[2\I; 0]^\top + O(\lambda^{-2})$ as $\lambda\to\infty$.  Therefore it is sufficient to solve for the first row of the matrices in \eqref{eq:Pi-partial-fractions}, so we introduce row vectors
\begin{equation}
[ \begin{matrix} r_{\mathrm{o}/\mathrm{e},k}(x,t) & u_{\mathrm{o}/\mathrm{e},k}(x,t)\end{matrix}]\defeq [ \begin{matrix} 1 & 0\end{matrix}]\mathbf{A}^+_{\mathrm{o}/\mathrm{e},k}(x,t), \quad [ \begin{matrix} s_{\mathrm{o}/\mathrm{e},k}(x,t) & v_{\mathrm{o}/\mathrm{e},k}(x,t)\end{matrix}]\defeq [ \begin{matrix} 1 & 0\end{matrix}]\mathbf{A}^-_{\mathrm{o}/\mathrm{e},k}(x,t)
\end{equation}
for $1\leq k \leq n$, and the system of equations \eqref{eq:hom-A}, \eqref{eq:hom-B}, \eqref{eq:inhom-A-Gamma}, and \eqref{eq:inhom-B-Gamma} results in a square inhomogeneous linear system for the vector unknown
\begin{equation}
\mathbf{y}=[\begin{matrix}r_{\mathrm{o}/\mathrm{e},n}&s_{\mathrm{o}/\mathrm{e},n}&u_{\mathrm{o}/\mathrm{e},n}&v_{\mathrm{o}/\mathrm{e},n} \cdots r_{\mathrm{o}/\mathrm{e},1}&s_{\mathrm{o}/\mathrm{e},1}&u_{\mathrm{o}/\mathrm{e},1}&v_{\mathrm{o}/\mathrm{e},1} \end{matrix}]^\top\in\mathbb{C}^{4n},
\end{equation}
with a coefficient matrix that consists of $2n$ $2\times 2$ blocks:
\setcounter{MaxMatrixCols}{20}
\begin{equation}
\mathbf{R}_{\mathrm{o}/\mathrm{e}}\defeq 
\begin{bsmallmatrix}
\mathbf{F}_0^{[1]}& \mathbf{F}_0^{[2]} & \mathbf{0} & \mathbf{0} &    \cdots   & \cdots &  \cdots  & \mathbf{0}\\
\mathbf{F}_1^{[1]}& \mathbf{F}_1^{[2]}& \mathbf{F}_1^{[0]}& \mathbf{F}_0^{[2]} & \mathbf{0}    & \cdots &  \cdots &\mathbf{0}\\
\vdots &&&&&&&\vdots\\
\mathbf{F}_{n-1}^{[1]} &\mathbf{F}_{n-1}^{[2]} &\mathbf{F}_{n-2}^{[1]} &\mathbf{F}_{n-2}^{[2]} & \cdots & \cdots & \mathbf{F}_{0}^{[1]} & \mathbf{F}_{0}^{[2]}\\
\mathbf{F}_{n}^{[1]}+\mathbf{H}^{[1]}_{0,n}&\mathbf{F}_{n}^{[2]}+\mathbf{H}^{[2]}_{0,n} & \mathbf{F}_{n-1}^{[1]}+\mathbf{H}^{[1]}_{0,n-1}&\mathbf{F}_{n-1}^{[2]}+\mathbf{H}^{[2]}_{0,n-1}&\cdots&\cdots  &\mathbf{F}_{1}^{[1]}+\mathbf{H}^{[1]}_{0,1}&\mathbf{F}_{1}^{[2]}+\mathbf{H}^{[2]}_{0,1}\\
\mathbf{F}_{n+1}^{[1]}+\mathbf{H}^{[1]}_{1,n}&\mathbf{F}_{n+1}^{[2]}+\mathbf{H}^{[2]}_{1,n} & \mathbf{F}_{n}^{[1]}+\mathbf{H}^{[1]}_{1,n-1}&\mathbf{F}_{n}^{[2]}+\mathbf{H}^{[2]}_{1,n-1}&\cdots&\cdots  &\mathbf{F}_{2}^{[1]}+\mathbf{H}^{[1]}_{1,1}&\mathbf{F}_{2}^{[2]}+\mathbf{H}^{[2]}_{1,1}\\
\mathbf{F}_{n+2}^{[1]}+\mathbf{H}^{[1]}_{2,n}&\mathbf{F}_{n+2}^{[2]}+\mathbf{H}^{[2]}_{2,n} & \mathbf{F}_{n+1}^{[1]}+\mathbf{H}^{[1]}_{2,n-1}&\mathbf{F}_{n+1}^{[2]}+\mathbf{H}^{[2]}_{2,n-1}&\cdots  &\cdots &\mathbf{F}_{n-1}^{[1]}+\mathbf{H}^{[1]}_{2,1}&\mathbf{F}_{n-1}^{[2]}+\mathbf{H}^{[2]}_{2,1}\\
\vdots &&&&&&&\vdots\\
\mathbf{F}_{2n-1}^{[1]}+\mathbf{H}^{[1]}_{n-1,n}&\mathbf{F}_{2n-1}^{[2]}+\mathbf{H}^{[2]}_{n-1,n-1} & \mathbf{F}_{2n-2}^{[1]}+\mathbf{H}^{[1]}_{n-1,n-1}&\mathbf{F}_{2n-2}^{[2]}+\mathbf{H}^{[2]}_{n-1,n-1}&\cdots&\cdots   &\mathbf{F}_{n}^{[1]}+\mathbf{H}^{[1]}_{n-1,1}&\mathbf{F}_{n}^{[2]}+\mathbf{H}^{[2]}_{n-1,1}
\end{bsmallmatrix},
\label{eq:lin-sys-coeff-mat}
\end{equation}
where we momentarily suppressed the subscript $\mathrm{o}/\mathrm{e}$ for the block elements. The blocks of $\mathbf{R}=\mathbf{R}(x,t)$ are defined in terms of the following $2\times 2$ matrices:
\begin{align}
\mathbf{F}^{[j]}_k &=\mathbf{F}^{[j]}_k(x,t)\defeq \begin{bmatrix} \left(\mathbf{w}^{+}_{\mathrm{e}/\mathrm{o},k}(x,t)\right)_j & 0 \\ 0 &\left(\mathbf{w}^{-}_{\mathrm{e}/\mathrm{o},k}(x,t)\right)_j \end{bmatrix} \eqdef \begin{bmatrix}{f}^{[j]+}_{\mathrm{e}/\mathrm{o},k}(x,t) & 0 \\ 0 & f^{[j]-}_{\mathrm{e}/\mathrm{o},k}(x,t)\end{bmatrix},\quad j=1,2,\\
\mathbf{H}^{[j]}_{m,k} &=\mathbf{H}^{[j]}_{m,k}(x,t) \defeq\begin{bmatrix} 0 & \sum\limits_{\ell=0}^m \gamma_{\ell m} f^{[j]+}_{\mathrm{e}/\mathrm{o},{m-\ell}}(x,t) \\
 \sum\limits_{\ell=0}^m (-1)^{\ell+m}\gamma_{\ell m} f^{[j]-}_{\mathrm{e}/\mathrm{o},{m-\ell}}(x,t) & 0
\end{bmatrix},\quad j=1,2.
\end{align}
Thus, the first row of the solution of Riemann-Hilbert problem~\ref{rhp:M-Rogue-Waves} can be obtained by solving the linear system
\begin{equation}
\mathbf{R}_{\mathrm{e}/\mathrm{o}}(x,t)\mathbf{y}(x,t) = \begin{bmatrix} 0\\  \vdots \\ 0 \\  -f^{[1]+}_{\mathrm{e}/\mathrm{o},0}(x,t) \\ -f^{[1]-}_{\mathrm{o}/\mathrm{e},0}(x,t)\\-f^{[1]+}_{\mathrm{e}/\mathrm{o},1}(x,t) \\ -f^{[1]-}_{\mathrm{e}/\mathrm{o},1}(x,t) \\ \vdots \\-f^{[1]+}_{\mathrm{e}/\mathrm{o},{n-1}}(x,t) \\ -f^{[1]-}_{\mathrm{e}/\mathrm{o},{n-1}}(x,t) \end{bmatrix},
\label{eq:lin-sys}
\end{equation}
and consequently the rogue wave solutions of arbitrary order are recovered by Cramer's rule:
\begin{align}
\psi_{\mathrm{P}_{2n-1}}(x,t) &=1+2\I\left(u_{\mathrm{o},1}(x,t)+v_{\mathrm{o},1}(x,t)\right)=1+ 2\I\frac{\det(\mathbf{R}_{\mathrm{o},2}(x,t))+\det(\mathbf{R}_{\mathrm{o},1}(x,t))}{\det(\mathbf{R}_\mathrm{o}(x,t))},\\
\psi_{\mathrm{P}_{2n}}(x,t) &=1+2\I\left(u_{\mathrm{e},1}(x,t)+v_{\mathrm{e},1}(x,t)\right)=1+ 2\I\frac{\det(\mathbf{R}_{\mathrm{e},2}(x,t))+\det(\mathbf{R}_{\mathrm{e},1}(x,t))}{\det(\mathbf{R}_\mathrm{e}(x,t))}.
\end{align}
\begin{figure}[ht!]
  \includegraphics[width=\textwidth]{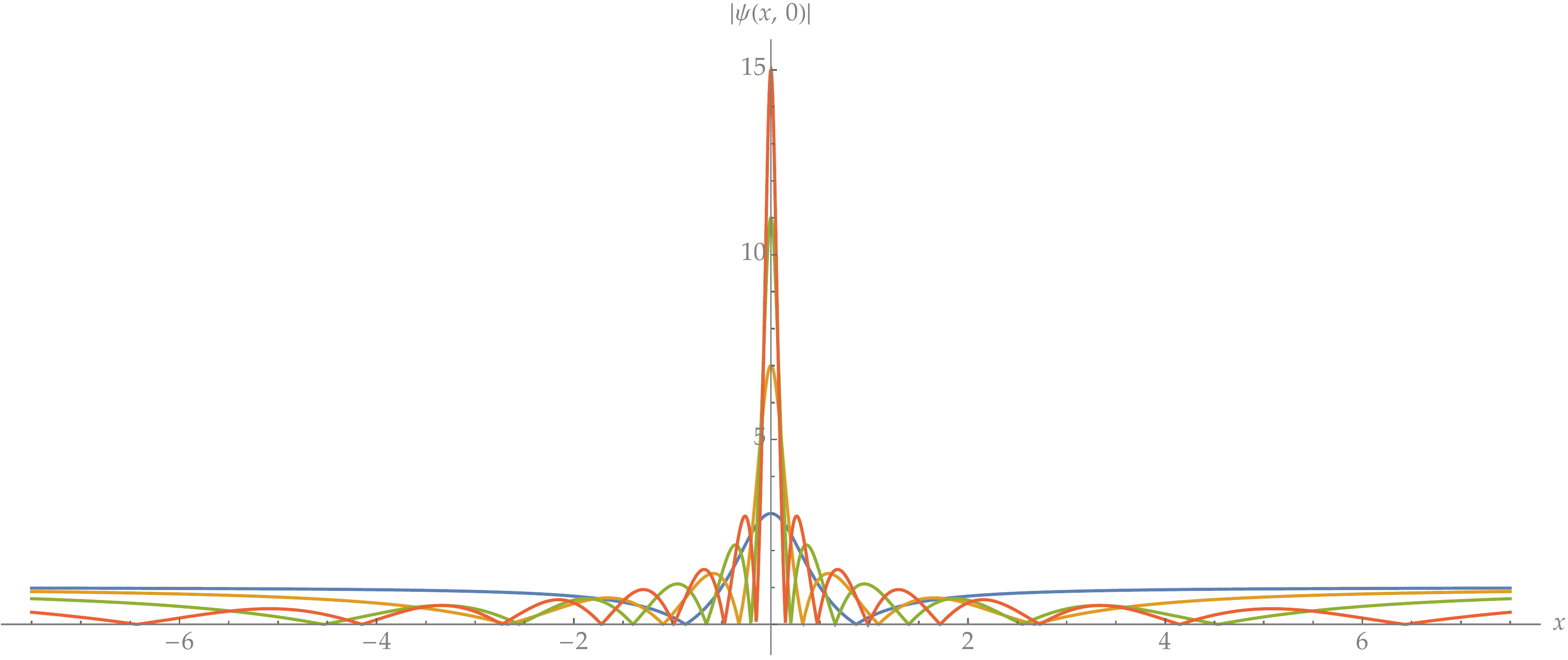}
  \caption{Moduli of the first four ``odd-order'' rogue waves: $|\psi_{\mathrm{P}_1}(x,0)|$ (blue), $|\psi_{\mathrm{P}_3}(x,0)|$ (orange), $|\psi_{\mathrm{P}_5}(x,0)|$ (green), and $|\psi_{\mathrm{P}_7}(x,0)|$ (red).}
  \label{f:odd-rogue-waves}
\end{figure}
Here $\mathbf{R}_{\mathrm{o}/\mathrm{e},k}(x,t)$ stands for the matrix $\mathbf{R}_{\mathrm{o}/\mathrm{e}}(x,t)$ with its $k^{\mathrm{th}}$ column replaced by the right-hand side vector in \eqref{eq:lin-sys}. For $n=1$, the system \eqref{eq:lin-sys} is $4\times 4$ and it can be solved by hand to obtain the Peregrine breather with its peak located at $(0,0)$. As $n$ grows, the complexity of the system grows rapidly and we solve both of the systems $\mathbf{R}_{\mathrm{o}/\mathrm{e},k}$ symbolically by computer for $n=1,2,3,4$. See Figure~\ref{f:odd-rogue-waves} for the ``odd-order'' rogue-wave solutions $\psi_{\mathrm{P}_{2n-1}}(x,t)$ of the NLS equation obtained by solving the linear system of \eqref{eq:lin-sys} with $\mathbf{R}(x,t)=\mathbf{R}_\mathrm{o}(x,t)$ and see Figure~\ref{f:even-rogue-waves} for the ``even-order'' rogue wave solutions $\psi_{\mathrm{P}_{2n}}(x,t)$ of the NLS equation in the case $\mathbf{R}(x,t)=\mathbf{R}_\mathrm{e}(x,t)$. See Figure~\ref{f:8th-rogue-wave} for a surface plot of $|\psi_8(x,t)|$ (corresponding to the linear system with $n=4$ and $\mathbf{R}(x,t)=\mathbf{R}_\mathrm{e}(x,t)$).
\begin{figure}[ht!]
  \includegraphics[width=\textwidth]{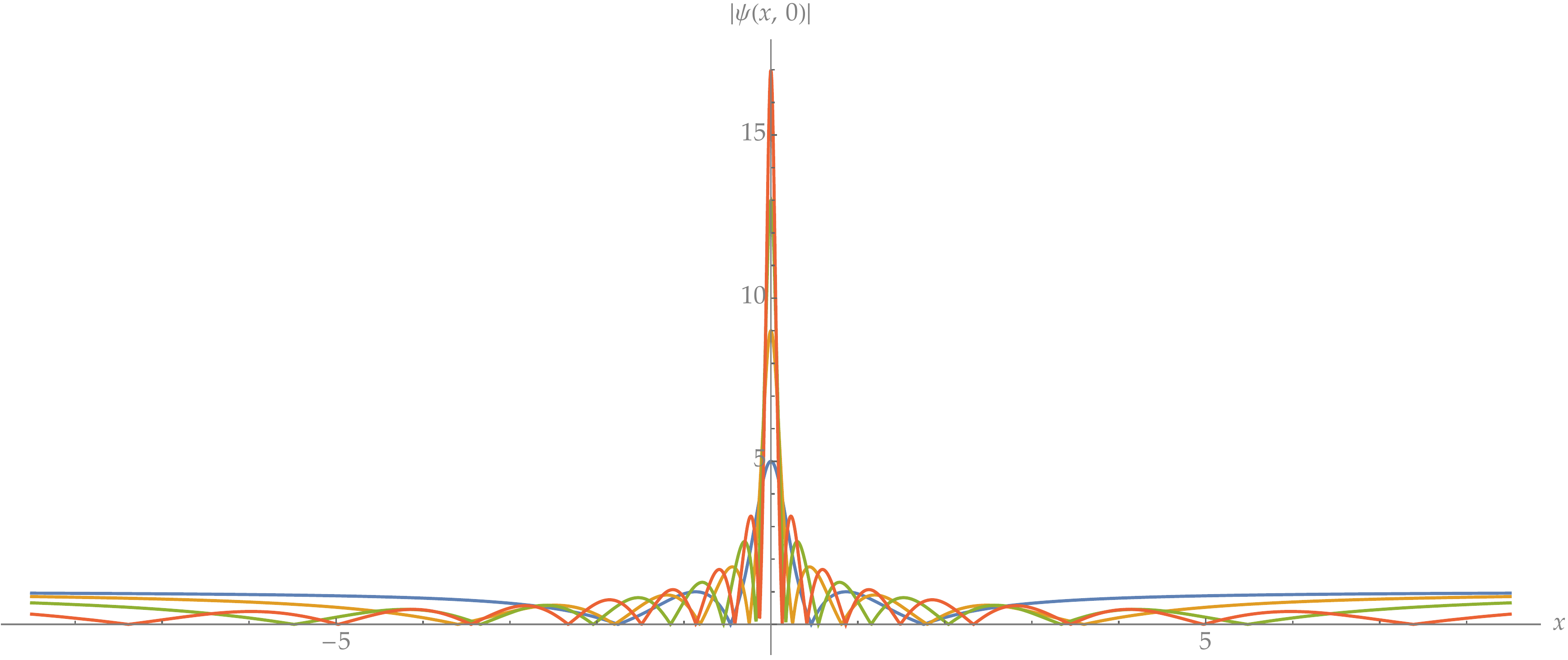}
  \caption{Moduli of the first four ``even-order'' rogue waves: $|\psi_{\mathrm{P}_2}(x,0)|$ (blue), $|\psi_{\mathrm{P}_4}(x,0)|$ (orange), $|\psi_{\mathrm{P}_6}(x,0)|$ (green), and $|\psi_{\mathrm{P}_8}(x,0)|$ (red).}
    \label{f:even-rogue-waves}
\end{figure}
\begin{figure}
  \includegraphics{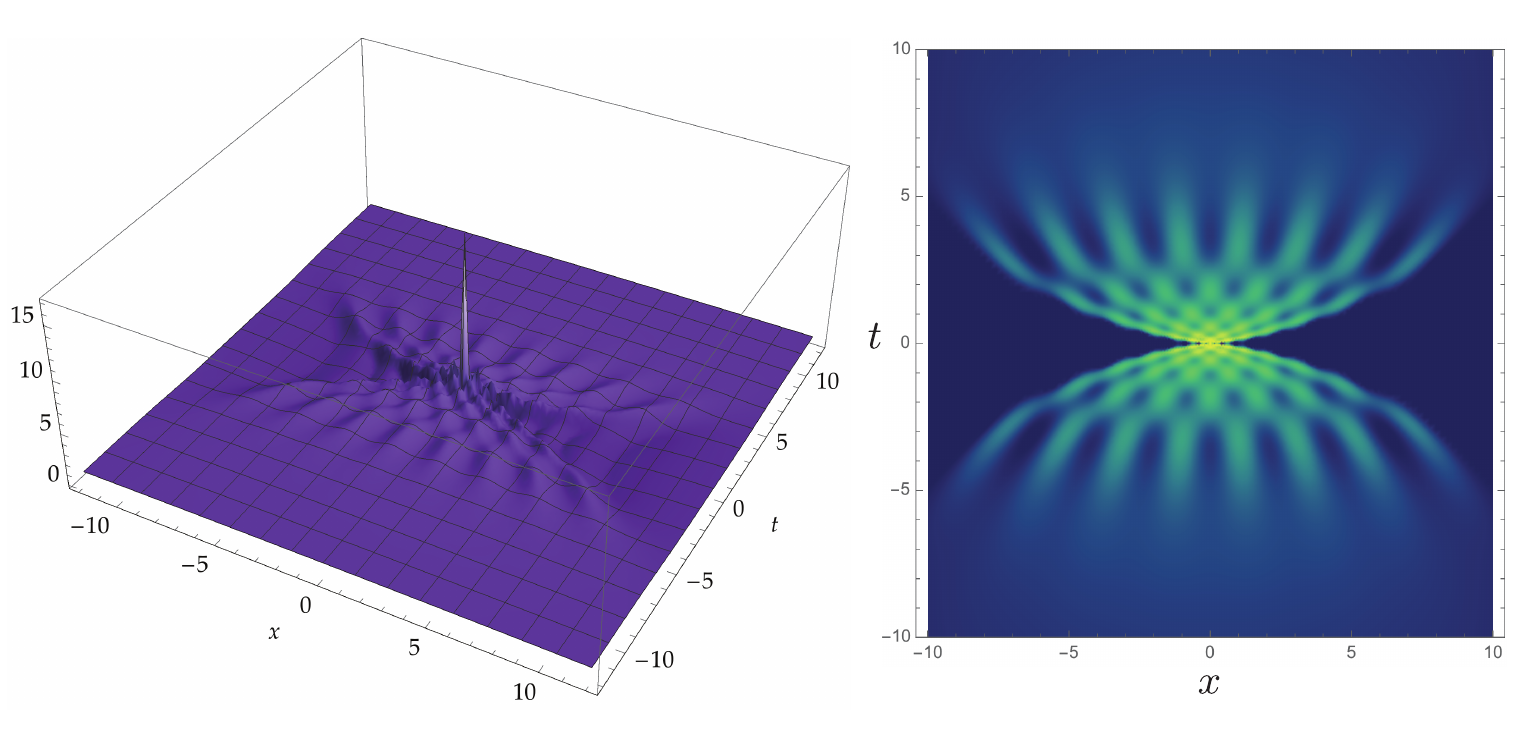}
  \caption{Modulus of the $8^{\mathrm{th}}$ order rogue wave: $|\psi_{\mathrm{P}_8}(x,t)|$, obtained from the solution of Riemann-Hilbert problem~\ref{rhp:M-Rogue-Waves} with $n=4$.  The pattern evident in the density plot (right-hand panel) is reminiscent of the wave patterns typical in solutions of the focusing NLS equation in the semiclassical limit \cite{MillerL2007}.}
  \label{f:8th-rogue-wave}
\end{figure}

\section{Linearization of the Direct and Inverse Scattering Transforms and Application to the Peregrine Solution}
\label{s:Linearization}

\subsection{Complex NLS and linearization preliminaries}
Consider the complex nonlinear Schr\"odinger (cNLS) for the scalar unknown functions $(\psi,\phi)$
\begin{equation}
  \begin{aligned}
    \I\psi_t +\frac{1}{2}\psi_{xx} +(-\psi\phi-1)\psi&=0\\
    \I\phi_t -\frac{1}{2}\phi_{xx} -(-\psi\phi-1)\phi&=0
    \label{eq:cNLS}
  \end{aligned}
\end{equation}
In case the focusing symmetry $\phi=-\psi^*$ holds, \eqref{eq:cNLS} reduces to the focusing nonlinear Schr\"odinger (NLS) equation \eqref{eq:NLS}.
Consider the formal linearization of \eqref{eq:cNLS} around a solution $(\psi_0,\phi_0)$ by considering solutions $(\psi,\phi)$ that are of the form:
\begin{equation}
  \begin{aligned}
  \psi&=\psi_0 +\varepsilon\psi_1+o(\varepsilon),\\
  \phi&=\phi_0+\varepsilon\phi_1+o(\varepsilon),\quad0<\varepsilon\ll 1.
  \end{aligned}
  \label{eq:complex-pert}
\end{equation}
Enforcing \eqref{eq:complex-pert} to solve \eqref{eq:cNLS} and retaining $O(\varepsilon)$ terms yields the following linear partial differential equations for the perturbation functions $(\psi_1,\phi_1)$:
\begin{equation}
  \begin{aligned}
    \I\psi_{1t} +\frac{1}{2}\psi_{1xx} +(-2\psi_0\phi_0-1)\psi_1 -\psi_0^2\phi_1&=0\\
    \I\phi_{1t} -\frac{1}{2}\phi_{1xx} -(-2\psi_0\phi_0-1)\phi_1 +\phi_0^2\psi_1&=0,
    \label{eq:lcNLS}
  \end{aligned}
\end{equation}
which we will refer to as the linearized complex NLS equation (lcNLS), and in case the focusing symmetry holds for the unperturbed fields ($\phi_0=-\psi_0^*$) and first-order perturbations ($\phi_1=-\psi_1^*$) \eqref{eq:lcNLS} reduces to the linearized NLS (lNLS) equation:
\begin{equation}
  \begin{aligned}
    \I\psi_{1t} +\frac{1}{2}\psi_{1xx} +(2|\psi_0|^2-1)\psi_1 +\psi_0^2\psi_1^*=0.
    \label{eq:lNLS}
  \end{aligned}
\end{equation}
\subsection{Squared eigenfunctions}
Let $\zeta$ be an arbitrary complex number, and
let $\mathbf{u}=\mathbf{u}^{[a]}(\zeta;x,t)$ and $\mathbf{u}=\mathbf{u}^{[b]}(\zeta;x,t)$ be any two (not necessarily independent or distinct) $2\times 1$ column vector simultaneous solutions  of the Lax pair associated with with the solutions $(\psi_0,\phi_0)$ of the cNLS equation for the same value of the spectral parameter $\lambda=\zeta$:
\begin{equation}
  \begin{aligned}
  \mathbf{u}_x&=\begin{bmatrix}-\I \lambda & \psi_0 \\ \phi_0 & \I\lambda\end{bmatrix}\mathbf{u}\\
  \mathbf{u}_t&=\begin{bmatrix}-\I \lambda^2 + \frac{\I}{2}(-\psi_0\phi_0-1)& \lambda\psi_0 + \frac{\I}{2}\psi_{0x} \\ \lambda\phi_0- \frac{\I}{2}\phi_{0x}  & \I \lambda^2 - \frac{\I}{2}(-\psi_0\phi_0-1)\end{bmatrix}\mathbf{u}.
\end{aligned}
\label{eq:lax-cNLS}
\end{equation}
Define the corresponding \emph{squared eigenfunctions}
\begin{equation}
  \begin{aligned}
    \mu_1^{[ab]}(\zeta;x,t)&\defeq u_1^{[a]}(\zeta;x,t)u^{[b]}_1(\zeta;x,t)\\
    \mu_2^{[ab]}(\zeta;x,t)&\defeq u_1^{[a]}(\zeta;x,t)u_2^{[b]}(\zeta;x,t)+u_2^{[a]}(\zeta;x,t)u_1^{[b]}(\zeta;x,t)\\
    \mu_3^{[ab]}(\zeta;x,t)&\defeq u_2^{[a]}(\zeta;x,t)u_2^{[b]}(\zeta;x,t).
  \end{aligned}
  \label{eq:sef}
\end{equation}
These squared eigenfunctions $\mu_j=\mu_j^{[ab]}$, $j=1,2,3$, satisfy the $3\times 3$ Lax pair:
\begin{equation}
  \begin{aligned}
  \frac{\partial}{\partial x}\begin{bmatrix}\mu_1\\\mu_2\\\mu_3\end{bmatrix}&=
  \begin{bmatrix}
    -2\I\lambda &\psi_0 & 0\\
    2\phi_0 & 0 & 2\psi_0\\
    0 & \phi_0 & 2\I\lambda
  \end{bmatrix}
  \begin{bmatrix}\mu_1\\\mu_2\\\mu_3\end{bmatrix}\\
  \frac{\partial}{\partial t}\begin{bmatrix}\mu_1\\\mu_2\\\mu_3\end{bmatrix}&=
  \begin{bmatrix}
    -2\I\lambda^2+\I(-\psi_0\phi_0-1) &\lambda\psi_0+\frac{\I}{2}\psi_{0x} & 0\\
    2\lambda\phi_0-\I\phi_{0x} & 0 & 2\lambda\psi_0 + \I\psi_{0x}\\
    0 & \lambda\phi_0-\frac{\I}{2}\phi_{0x} & 2\I\lambda^2-\I(-\psi_0\phi_0-1)
  \end{bmatrix}\begin{bmatrix}\mu_1\\\mu_2\\\mu_3\end{bmatrix}
\end{aligned}
\label{eq:3x3-lax-cNLS}
\end{equation}
for the spectral parameter evaluated at $\lambda=\zeta$.
Using the differential equations \eqref{eq:3x3-lax-cNLS} it follows that for any $\zeta\in\mathbb{C}$ the squared eigenfunctions $(\mu_1^{[ab]}(\zeta;x,t), \mu_3^{[ab]}(\zeta;x,t))$ satisfy the following partial differential equations in which the spectral parameter does not appear explicitly:
\begin{equation}
  \begin{aligned}
    \I\mu^{[ab]}_{1t} +\frac{1}{2}\mu^{[ab]}_{1xx} +(-2\psi_0\phi_0-1)\mu^{[ab]}_1 -\psi_0^2\mu^{[ab]}_3&=0\\
    \I\mu^{[ab]}_{3t} -\frac{1}{2}\mu^{[ab]}_{3xx} -(-2\psi_0\phi_0-1)\mu^{[ab]}_3 +\phi_0^2\mu^{[ab]}_1&=0,
  \end{aligned}
\end{equation}
which is precisely the lcNLS given in \eqref{eq:lcNLS} with $\psi_1\defeq \mu_1^{[ab]}(\zeta;x,t)$ and $\phi_1\defeq \mu_3^{[ab]}(\zeta;x,t)$. Explicitly, we have the lcNLS written as the linear system:
\begin{equation}
  \begin{bmatrix} \I \frac{\partial}{\partial t} + \frac{1}{2}\frac{\partial^2}{\partial x^2} -2\psi_0\phi_0 -1 & -\psi_0^2 \\ \phi_0^2 & \I \frac{\partial}{\partial t} - \frac{1}{2}\frac{\partial^2}{\partial x^2} +2\psi_0\phi_0 +1
  \end{bmatrix}
  \begin{bmatrix}
    \mu^{[ab]}_1(\zeta;x,t) \\ \mu^{[ab]}_3(\zeta;x,t)
  \end{bmatrix}=\begin{bmatrix}0\\0\end{bmatrix}.
\end{equation}

Now suppose that $\phi_0(x,t)=-\psi_0(x,t)^*$.
Unfortunately, the ordered pair of squared eigenfunctions $(\psi_1(x,t),\phi_1(x,t))=(\mu^{[ab]}_1(\zeta;x,t),\mu^{[ab]}_3(\zeta;x,t))$ does \underline{not} necessarily satisfy the focusing symmetry: $ \mu^{[ab]}_3(\zeta;x,t) = -\mu^{[ab]}_1(\zeta;x,t)^* $ needed to reduce \eqref{eq:lcNLS} to \eqref{eq:lNLS}, but this can be remedied by applying the superposition principle making use of different values of the spectral parameter. Indeed, the focusing symmetry $\phi_0(x,t)=-\psi_0(x,t)^*$ of the unperturbed problem implies that the vectors
\begin{equation}
  \mathbf{u}^{[j]}(\zeta^*;x,t):=\begin{bmatrix}-u_2^{[j]}(\zeta;x,t)^*\\u_1^{[j]}(\zeta;x,t)^*\end{bmatrix},\quad j=a,b,
\end{equation}
are simultaneous solutions of \eqref{eq:lax-cNLS} for the spectral parameter $\lambda=\zeta^*$.  Denoting the corresponding squared eigenfunctions defined via \eqref{eq:sef} from these two vectors as $\mu^{[ab]}_k(\zeta^*;x,t)$, $k=1,2,3$, it follows easily that
\begin{equation}
\mu^{[ab]}_1(\zeta^*;x,t)=\mu^{[ab]}_3(\zeta;x,t)^*\quad\text{and}\quad\mu^{[ab]}_3(\zeta^*;x,t)=\mu^{[ab]}_1(\zeta;x,t)^*.
\end{equation}
Now since the spectral parameter does not appear explicitly in the linear system \eqref{eq:lcNLS}, it is solved by both $(\psi_1,\phi_1)=(\mu^{[ab]}_1(\zeta;x,t),\mu^{[ab]}_3(\zeta;x,t))$ and $(\psi_1,\phi_1)=(\mu^{[ab]}_1(\zeta^*;x,t),\mu^{[ab]}_3(\zeta^*;x,t))$, and hence by  superposition
  \begin{equation}
  \begin{split}
      (\psi_1(x,t),\phi_1(x,t))&\defeq \left(\mu^{[ab]}_1(\zeta;x,t)-\mu^{[ab]}_1(\zeta^*;x,t), \mu^{[ab]}_3(\zeta;x,t) - \mu^{[ab]}_3(\zeta^*;x,t)\right)\\ &=
      \left(\mu^{[ab]}_1(\zeta;x,t)-\mu^{[ab]}_3(\zeta;x,t)^*,\mu^{[ab]}_3(\zeta;x,t)-\mu^{[ab]}_1(\zeta;x,t)^*\right)
      \end{split}
  \end{equation}
is a solution of the system \eqref{eq:lcNLS} for each $\zeta\in\mathbb{C}$ that satisfies the focusing symmetry $\phi_1(x,t)=-\psi_1(x,t)^*$. Therefore
\begin{equation}
\begin{split}
  \psi_1\defeq \mu^{[ab]}(\zeta;x,t)&\defeq \mu^{[ab]}_1(\zeta;x,t)-\mu^{[ab]}_3(\zeta;x,t)^*\\
  &=u_1^{[a]}(\zeta;x,t)u_1^{[b]}(\zeta;x,t)-u_2^{[a]}(\zeta;x,t)^*u_2^{[b]}(\zeta;x,t)^*
  \end{split}
  \label{eq:psi-1-first}
\end{equation}
solves the linearized focusing NLS equation \eqref{eq:lNLS} whenever $\mathbf{u}^{[j]}(\zeta;x,t)$, $j=a,b$, are any two simultaneous solutions of \eqref{eq:lax-cNLS} with $\phi_0(x,t)=-\psi_0(x,t)^*$ for the same arbitrary value $\zeta$ of the spectral parameter.  Note that if the vectors $\mathbf{u}^{[j]}(\zeta;x,t)$, $j=a,b$, are multiplied by scalar factors $c^{[j]}(\zeta)$, $j=a,b$, whose product is $C(\zeta)=c^{[a]}(\zeta)c^{[b]}(\zeta)$, then the formula \eqref{eq:psi-1-first} is transformed into
\begin{equation}
\psi_1(x,t)=C(\zeta)u_1^{[a]}(\zeta;x,t)u_1^{[b]}(\zeta;x,t)-C(\zeta)^*u_2^{[a]}(\zeta;x,t)^*u_2^{[b]}(\zeta;x,t)^*,
\label{eq:psi-1-general}
\end{equation}
which again is a solution of the linearized focusing NLS equation \eqref{eq:lNLS} regardless of the value of $C(\zeta)\in\mathbb{C}$.  Writing $C(\zeta)=\alpha(\zeta)+\I \beta(\zeta)$ for real numbers $\alpha(\zeta)$ and $\beta(\zeta)$, the solution \eqref{eq:psi-1-general} can be written as the real linear combination $\alpha(\zeta)\mu^{[ab]}(\zeta;x,t) + \beta(\zeta)\nu^{[ab]}(\zeta;x,t)$ where $\mu(\zeta;x,t)$ is defined by \eqref{eq:psi-1-first} and where
\begin{equation}
\nu^{[ab]}(\zeta;x,t):=\I u_1^{[a]}(\zeta;x,t)u_1^{[b]}(\zeta;x,t)+\I u_2^{[a]}(\zeta;x,t)^*u_2^{[b]}(\zeta;x,t)^*.
\label{eq:nu-def}
\end{equation}
Further solutions of the linearized NLS equation \eqref{eq:lNLS} can be obtained by taking real linear combinations of $\mu^{[ab]}(\zeta;x,t)$ and $\nu^{[ab]}(\zeta;x,t)$ for different complex values of $\zeta$.

A natural question that arises is that of completeness of the set of particular solutions of \eqref{eq:lNLS} of the type indicated above in which $\zeta$ is allowed to range over some subset of the complex numbers.  What class of solutions of \eqref{eq:lNLS} can be represented as suitable finite or infinite superpositions of them?

\subsection{Formal linearization of the inverse-scattering transform}
A systematic way to address the completeness issue is to analyze the solution of the \emph{nonlinear} problem with initial data $\psi(x,0)=\psi_0(x,0)+\varepsilon\psi_1(x,0)$ obtained via the inverse-scattering transform.  This approach was used by Kaup \cite{Kaup1976b} to study localized perturbations of solutions decaying to the zero background.  Here we use the robust inverse-scattering transform introduced in Section~\ref{s:New-IST} and consider localized perturbations of solutions decaying to the nonzero background value $\psi\equiv 1$.

\subsubsection{Linearizing the direct transform}
Let $\psi_0(x,t)$ be a solution of \eqref{eq:NLS} for which $\psi_0-1\in L^1(\mathbb{R})$ as a function of $x$ for each $t\in\mathbb{R}$.
Suppose that $\psi_1(x)$ is a sufficiently-localized perturbation. 
For $\varepsilon>0$ small, we consider the initial data $\psi(x,0;\varepsilon)=\psi_0(x,0)+\varepsilon\psi_1(x)$.  For sufficiently large radius $r>0$ (independent of $\varepsilon$) we calculate the ``core'' jump matrices
\begin{equation}
\mathbf{V}^+(\lambda;\varepsilon):=\left[a(\lambda;\varepsilon)^{-1}\mathbf{j}^{-,1}(\lambda;0,0;\varepsilon); \,\mathbf{j}^{+,2}(\lambda;0,0;\varepsilon)\right],\quad \Im\{\lambda\}>0,\quad |\lambda|=r,
\label{eq:Vplus-Linearization}
\end{equation}
\begin{equation}
\mathbf{V}^-(\lambda;\varepsilon):=\left[\mathbf{j}^{+,1}(\lambda;0,0;\varepsilon);\,a(\lambda^*;\varepsilon)^{-*}\mathbf{j}^{-,2}(\lambda;0,0;\varepsilon)\right],\quad\Im\{\lambda\}<0,\quad |\lambda|=r,
\end{equation}
and
\begin{equation}
\mathbf{V}^{\mathrm{out}}(\lambda;\varepsilon):=\begin{bmatrix}1+|R(\lambda;\varepsilon)|^2 & R(\lambda;\varepsilon)^*\\R(\lambda;\varepsilon) & 1\end{bmatrix},\quad R(\lambda;\varepsilon):=\frac{b(\lambda;\varepsilon)}{a(\lambda;\varepsilon)},\quad \lambda\in\mathbb{R},\quad |\lambda|>r,
\end{equation}
where $\mathbf{j}^{\pm,k}(\lambda;0,0;\varepsilon)$ denote the columns of the Jost matrices for the potential $\psi(x,0;\varepsilon)=\psi_0(x,0)+\varepsilon\psi_1(x)$ evaluated at $x=t=0$ and where $a(\lambda;\varepsilon)$ and $b(\lambda;\varepsilon)$ refer to the associated scattering matrix.  More to the point, the Jost column vectors are the unique solutions of the following Volterra integral equations:  $\mathbf{j}^{\pm,1}(\lambda;x,0;\varepsilon)=\E^{-\I\rho(\lambda)x}\mathbf{k}^{\pm,1}(\lambda;x,0;\varepsilon)$, where
\begin{multline}
\mathbf{k}^{\pm,1}(\lambda;x,0;\varepsilon)-\mathscr{V}^{\pm,1}[\psi_0(\cdot,0)-1]\mathbf{k}^{\pm,1}(\lambda;x,0;\varepsilon)
=\mathbf{e}^1(\lambda) +\varepsilon \mathscr{V}^{\pm,1}[\psi_1]\mathbf{k}^{\pm,1}(\lambda;x,0;\varepsilon), \\
  \mp\Im\{\rho(\lambda)\}\ge 0,
\label{eq:kpm1-epsilon}
\end{multline}
where $\mathscr{V}^{\pm,1}[\psi]$ denote Volterra integral operators
\begin{equation}
\mathscr{V}^{\pm,1}[\psi]\mathbf{f}(x):=\int_{\pm\infty}^x\mathbf{E}(\lambda)
\begin{bmatrix}1 & 0\\0 & \E^{2\I\rho(\lambda)(x-y)}\end{bmatrix}\mathbf{E}(\lambda)^{-1}\begin{bmatrix}0 & \psi(y)\\-\psi(y)^* & 0\end{bmatrix}\mathbf{f}(y)\D y,\quad \mp\Im\{\rho(\lambda)\}\ge 0,
\end{equation}
and $\mathbf{j}^{\pm,2}(\lambda;x,0;\varepsilon)=\E^{\I\rho(\lambda)x}\mathbf{k}^{\pm,2}(\lambda;x,0;\varepsilon)$, where
\begin{multline}
\mathbf{k}^{\pm,2}(\lambda;x,0;\varepsilon)-\mathscr{V}^{\pm,2}[\psi_0(\cdot,0)-1]\mathbf{k}^{\pm,2}(\lambda;x,0;\varepsilon)
=\mathbf{e}^2(\lambda)+ \varepsilon\mathscr{V}^{\pm,2}[\psi_1]\mathbf{k}^{\pm,2}(\lambda;x,0;\varepsilon), \\  \pm\Im\{\rho(\lambda)\}\ge 0,
\label{eq:kpm2-epsilon}
\end{multline}
where $\mathscr{V}^{\pm,2}[\psi]$ denote Volterra operators
\begin{equation}
\mathscr{V}^{\pm,2}[\psi]\mathbf{f}(x):=\int_{\pm\infty}^x\mathbf{E}(\lambda)
\begin{bmatrix}\E^{-2\I\rho(\lambda)(x-y)} & 0\\0 & 1\end{bmatrix}\mathbf{E}(\lambda)^{-1}\begin{bmatrix}0 & \psi(y)\\
-\psi(y)^* & 0\end{bmatrix}\mathbf{f}(y)\D y,\quad\pm\Im\{\rho(\lambda)\}\ge 0.
\end{equation}
Under the assumption that $\psi\in L^1(\mathbb{R})$ and the indicated conditions on $\Im\{\rho(\lambda)\}$, the operators $\mathscr{V}^{\pm,j}[\psi]$, $j=1,2$, are bounded on $L^\infty(\mathbb{R})$, and also $\mathscr{I}-\mathscr{V}^{\pm,j}[\psi]$, $j=1,2$,  have bounded inverses on $L^\infty(\mathbb{R})$.  The corresponding operator norms $\|\mathscr{V}^{\pm,j}[\psi]\|$ and $\|(\mathscr{I}-\mathscr{V}^{\pm,j}[\psi])^{-1}\|$, $j=1,2$, are uniformly bounded with respect to $\lambda$ under the
indicated conditions on $\Im\{\rho(\lambda)\}$.  Thus, \eqref{eq:kpm1-epsilon} can be rewritten in the form
\begin{multline}
\mathbf{k}^{\pm,1}(\lambda;x,0;\varepsilon)-\varepsilon (\mathscr{I}-\mathscr{V}^{\pm,1}[\psi_0(\cdot,0)-1])^{-1}\circ\mathscr{V}^{\pm,1}[\psi_1]\mathbf{k}^{\pm,1}(\lambda;x,0;\varepsilon)=\mathbf{k}_0^{\pm,1}(\lambda;x,0),\\ \mp\Im\{\rho(\lambda)\}\ge 0,
\label{eq:kpm1-epsilon-rewrite}
\end{multline}
where
\begin{equation}
\mathbf{k}_0^{\pm,1}(\lambda;x,0):=(\mathscr{I}-\mathscr{V}^{\pm,1}[\psi_0(\cdot,0)-1])^{-1}\mathbf{e}^1(\lambda),\quad \mp\Im\{\rho(\lambda)\}\ge 0,\quad \rho(\lambda)\neq 0.
\label{eq:k0-pm-1}
\end{equation}
Similarly, \eqref{eq:kpm2-epsilon} can be rewritten as
\begin{multline}
\mathbf{k}^{\pm,2}(\lambda;x,0;\varepsilon)-\varepsilon(\mathscr{I}-\mathscr{V}^{\pm,2}[\psi_0(\cdot,0)-1])^{-1}\circ\mathscr{V}^{\pm,2}[\psi_1]\mathbf{k}^{\pm,2}(\lambda;x,0;\varepsilon)=\mathbf{k}_0^{\pm,2}(\lambda;x,0),\\
\pm\Im\{\rho(\lambda)\}\ge 0,
\label{eq:kpm2-epsilon-rewrite}
\end{multline}
where
\begin{equation}
\mathbf{k}_0^{\pm,2}(\lambda;x,0):=(\mathscr{I}-\mathscr{V}^{\pm,2}[\psi_0(\cdot,0)-1])^{-1}\mathbf{e}^2(\lambda),\quad \pm\Im\{\rho(\lambda)\}\ge 0,\quad \rho(\lambda)\neq 0.
\label{eq:k0-pm-2}
\end{equation}
In both cases, the inhomogeneous terms $\mathbf{k}_0^{\pm,j}(\lambda;x,0)$, $j=1,2$, lie in $L^\infty(\mathbb{R})$  (the additional condition $\rho(\lambda)\neq 0$ is required because otherwise $\mathbf{e}^j(\lambda)$ does not exist), and if $\rho(\lambda)$ is bounded away from zero, their norms are bounded uniformly with respect to $\lambda$.  It then follows that if $\varepsilon \|(\mathscr{I}-\mathscr{V}^{\pm,1}[\psi_0(\cdot,0)-1])^{-1}\|\cdot \|\mathscr{V}^{\pm,1}[\psi_1]\|<1$, then \eqref{eq:kpm1-epsilon-rewrite} can be solved by a Neumann series convergent in $L^\infty(\mathbb{R})$, which takes the form of a \emph{power series} in $\varepsilon$.  A similar result holds for \eqref{eq:kpm2-epsilon-rewrite} under the condition $\varepsilon \|(\mathscr{I}-\mathscr{V}^{\pm,2}[\psi_0(\cdot,0)-1])^{-1}\|\cdot \|\mathscr{V}^{\pm,2}[\psi_1]\|<1$.
This implies that the unique solutions furnished by the sums of the convergent Neumann series are analytic functions of $\varepsilon$ to $L^\infty(\mathbb{R})$ at $\varepsilon=0$.

Note that the inhomogeneous terms defined in \eqref{eq:k0-pm-1} and \eqref{eq:k0-pm-2} if also $\rho(\lambda)\neq 0$ are simply the solutions of the integral equations \eqref{eq:kpm1-epsilon} and \eqref{eq:kpm2-epsilon} respectively for $\varepsilon=0$, and therefore are proportional via the exponential factors $\E^{\pm\I\rho(\lambda)x}$ to the Jost solutions of the unperturbed problem.  We therefore have series representations:
\begin{multline}
\mathbf{j}^{\pm,1}(\lambda;x,0;\varepsilon)=\sum_{n=0}^\infty\varepsilon^n\mathbf{j}^{\pm,1}_n(\lambda;x,0),\quad \mathbf{j}^{\pm,1}_n(\lambda;x,0)=\E^{-\I\rho(\lambda)x}\mathbf{k}_n^{\pm,1}(\lambda;x,0),\\
 \mp\Im\{\rho(\lambda)\}\ge 0,\quad \rho(\lambda)\neq 0,
\label{eq:series1}
\end{multline}
and
\begin{multline}
\mathbf{j}^{\pm,2}(\lambda;x,0;\varepsilon)=\sum_{n=0}^\infty\varepsilon^n\mathbf{j}^{\pm,2}_n(\lambda;x,0),\quad\mathbf{j}^{\pm,2}_n(\lambda;x,0)=\E^{\I\rho(\lambda)x}\mathbf{k}_n^{\pm,2}(\lambda;x,0),\\\pm\Im\{\rho(\lambda)\}\ge 0,\quad\rho(\lambda)\neq 0,
\label{eq:series2}
\end{multline}
both convergent for sufficiently small $|\varepsilon|$, with the convergence interpreted in a weighted $L^\infty(\mathbb{R})$ space corresponding to unweighted $L^\infty(\mathbb{R})$ convergence for $\mathbf{k}^{\pm,j}(\lambda;x,0;\varepsilon)$, $j=1,2$.

We may now solve explicitly for the first-order correction terms $\mathbf{j}^{\pm,j}_1(\lambda;x,0)$.  These are given at first by the expressions
\begin{multline}
\mathbf{k}_1^{\pm,1}(\lambda;x,0)=(\mathscr{I}-\mathscr{V}^{\pm,1}[\psi_0(\cdot,0)-1])^{-1}\circ\mathscr{V}^{\pm,1}[\psi_1]\mathbf{k}_0^{\pm,1}(\lambda;x,0),\\ \mp\Im\{\rho(\lambda)\}\ge 0,\quad \rho(\lambda)\neq 0
\end{multline}
and
\begin{multline}
\mathbf{k}_1^{\pm,2}(\lambda;x,0)=(\mathscr{I}-\mathscr{V}^{\pm,2}[\psi_0(\cdot,0)-1])^{-1}\circ\mathscr{V}^{\pm,2}[\psi_1]\mathbf{k}_0^{\pm,2}(\lambda;x,0),\\\pm\Im\{\rho(\lambda)\}\ge 0,\quad\rho(\lambda)\neq 0.
\end{multline}
Multiplying these through by $\E^{-\I\rho(\lambda)x}(\mathscr{I}-\mathscr{V}^{\pm,1}[\psi_0(\cdot,0)-1])$ and $\E^{\I\rho(\lambda)x}(\mathscr{I}-\mathscr{V}^{\pm,2}[\psi_0(\cdot,0)-1])$ respectively, eliminating $\mathbf{k}_n^{\pm,j}(\lambda;x,0)$ in favor of $\mathbf{j}_n^{\pm,j}(\lambda;x,0)$, $j=1,2$, and combining the resulting equations with their derivatives with respect to $x$ yields in all cases the same differential equation
\begin{multline}
\frac{\partial\mathbf{j}^{\pm,j}_1}{\partial x}(\lambda;x,0)=\begin{bmatrix}-\I\lambda & \psi_0(x,0)\\-\psi_0(x,0)^* & \I\lambda\end{bmatrix}
\mathbf{j}^{\pm,j}_1(\lambda;x,0) + \begin{bmatrix}0 & \psi_1(x)\\-\psi_1(x)^* & 0\end{bmatrix}\mathbf{j}_0^{\pm,j}(\lambda;x,0),\\ j=1,2.
\end{multline}
The general solution of this equation may be obtained from a fundamental matrix for the homogeneous problem and variation of parameters.
To obtain $\mathbf{j}^{-,1}_1(\lambda;x,0)$ and $\mathbf{j}^{+,2}_1(\lambda;x,0)$ for $|\lambda|=r$ with $\Im\{\lambda\}>0$, we use the fundamental matrix $\mathbf{U}_{0}^\mathrm{BC}(\lambda;x,0)=\mathbf{M}_{0}^\mathrm{BC}(\lambda;x,0)\E^{-\I\rho(\lambda)x\sigma_3}$ which is given explicitly in terms of the Jost solutions of the unperturbed problem by
\begin{equation}
\mathbf{U}_{0}^\mathrm{BC}(\lambda;x,0)=\left[ a_0(\lambda)^{-1}\mathbf{j}^{-,1}_0(\lambda;x,0);\,\mathbf{j}^{+,2}_0(\lambda;x,0)\right], \quad |\lambda|=r,\quad\Im\{\lambda\}>0,
\end{equation}
where $a_0(\lambda)$ is the Wronskian determinant that ensures that $\det(\mathbf{U}_{0}^\mathrm{BC}(\lambda;x,0))=1$; note that $a_0(\lambda)\neq 0$ by choice of sufficiently large radius $r$.  We observe also that $\mathbf{V}^+(\lambda;0)=\mathbf{U}_{0}^\mathrm{BC}(\lambda;0,0)$.  Thus, the variation of parameters formula gives, if $|\lambda|=r$ and $\Im\{\lambda\}>0$,
\begin{equation}
\mathbf{j}_1^{-,1}(\lambda;x,0)=\mathbf{U}_{0}^\mathrm{BC}(\lambda;x,0)\int_{-\infty}^x\mathbf{U}_{0}^\mathrm{BC}(\lambda;y,0)^{-1}\begin{bmatrix}0 & \psi_1(y)\\-\psi_1(y)^* & 0\end{bmatrix}\mathbf{j}_0^{-,1}(\lambda;y,0)\,\D y,
\end{equation}
where the lower limit of integration is fixed by the requirement that $\mathbf{k}_1^{-,1}(\lambda;x,0)\to \mathbf{0}$ as $x\to-\infty$, and
\begin{equation}
\mathbf{j}_1^{+,2}(\lambda;x,0)=\mathbf{U}_{0}^\mathrm{BC}(\lambda;x,0)\int_{+\infty}^x\mathbf{U}_{0}^\mathrm{BC}(\lambda;y,0)^{-1}\begin{bmatrix}0 & \psi_1(y)\\
-\psi_1(y)^* & 0\end{bmatrix}\mathbf{j}_0^{+,2}(\lambda;y,0)\,\D y
\end{equation}
where the lower limit of integration is fixed by the requirement that $\mathbf{k}_1^{+,2}(\lambda;x,0)\to\mathbf{0}$ as $x\to+\infty$.  With the squared eigenfunctions  $\mu^{[jk]}_1(\lambda;x,t):=U_{0,1j}^\mathrm{BC}(\lambda;x,t)U_{0,1k}^\mathrm{BC}(\lambda;x,t)$ and
$\mu^{[jk]}_3(\lambda;x,t):=U_{0,2j}^\mathrm{BC}(\lambda;x,t)U_{0,2k}^\mathrm{BC}(\lambda;x,t)$, these formulae become
\begin{equation}
\mathbf{j}_1^{-,1}(\lambda;x,0)=a_0(\lambda)\mathbf{U}_{0}^\mathrm{BC}(\lambda;x,0)\int_{-\infty}^x\begin{bmatrix}\psi_1(y)\mu_3^{[12]}(\lambda;y,0) + \psi_1(y)^*\mu_1^{[12]}(\lambda;y,0)\\
-\psi_1(y)\mu_3^{[11]}(\lambda;y,0)-\psi_1(y)^*\mu_1^{[11]}(\lambda;y,0)\end{bmatrix}\,\D y
\label{eq:j1m1}
\end{equation}
and
\begin{equation}
\mathbf{j}_1^{+,2}(\lambda;x,0)=\mathbf{U}_{0}^\mathrm{BC}(\lambda;x,0)\int_{+\infty}^x\begin{bmatrix}
\psi_1(y)\mu_3^{[22]}(\lambda;y,0) +\psi_1(y)^*\mu_1^{[22]}(\lambda;y,0)\\
-\psi_1(y)\mu_3^{[12]}(\lambda;y,0)-\psi_1(y)^*\mu_1^{[12]}(\lambda;y,0)\end{bmatrix}\,\D y.
\label{eq:j1p2}
\end{equation}
Since $a(\lambda;\varepsilon):=\det(\mathbf{j}^{-,1}(\lambda;x,t;\varepsilon),\mathbf{j}^{+,2}(\lambda;x,t;\varepsilon))$, from the convergent series \eqref{eq:series1}--\eqref{eq:series2} we get that
\begin{equation}
a(\lambda;\varepsilon)=a_0(\lambda)+\varepsilon a_1(\lambda)+O(\varepsilon^2),\quad\varepsilon\to 0
\end{equation}
holds uniformly for $|\lambda|=r$ with $\Im\{\lambda\}\ge 0$ as well as for $\lambda\in\mathbb{R}$, $|\lambda|>r$, where for any $x\in\mathbb{R}$,
\begin{equation}
\begin{split}
a_1(\lambda)&=\det\left[\mathbf{j}_0^{-,1}(\lambda;x,0);\,\mathbf{j}_1^{+,2}(\lambda;x,0)\right] + \det\left[\mathbf{j}_1^{-,1}(\lambda;x,0);\,\mathbf{j}_0^{+,2}(\lambda;x,0)\right]\\
& = a_0(\lambda)\int_{-\infty}^{+\infty}(\psi_1(y)\mu_3^{[12]}(\lambda;y,0)+\psi_1(y)^*\mu_1^{[12]}(\lambda;y,0))\,\D y.
\end{split}
\label{eq:a1}
\end{equation}
Now, the formula \eqref{eq:Vplus-Linearization} admits a convergent expansion in powers of $\varepsilon$ of the form
$\mathbf{V}^+(\lambda;\varepsilon)=\mathbf{V}_0^+(\lambda)+\varepsilon\mathbf{V}^+_1(\lambda) + O(\varepsilon^2)$,
and combining the explicit expansion of \eqref{eq:Vplus-Linearization} with \eqref{eq:j1m1}--\eqref{eq:j1p2} evaluated at $x=0$ with $\mathbf{U}_{0}^\mathrm{BC}(\lambda;0,0)=\mathbf{V}^+_0(\lambda)$, as well as \eqref{eq:a1} gives
\begin{equation}
\mathbf{V}_1^+(\lambda)=\left[a_0(\lambda)^{-1}\mathbf{j}_1^{-,1}(\lambda;0,0)-a_0(\lambda)^{-2}a_1(\lambda)\mathbf{j}_0^{-,1}(\lambda;0,0);\,\mathbf{j}_1^{+,2}(\lambda;0,0)\right]
=\mathbf{V}_0^+(\lambda)\mathbf{T}(\lambda)\sigma_3\sigma_1,
\label{eq:V1-plus}
\end{equation}
where $\mathbf{T}(\lambda)$ is a matrix of transforms given by
\begin{equation}
\mathbf{T}(\lambda):=\begin{bmatrix} \displaystyle -\int_0^{+\infty}\left(\psi_1\mu_3^{[22]}+\psi_1^*\mu_1^{[22]}\right)\,\D y &\displaystyle
\int_0^{+\infty}\left(\psi_1\mu_3^{[12]}+\psi_1^*\mu_1^{[12]}\right)\,\D y \\ \displaystyle \int_0^{+\infty}\left(\psi_1\mu_3^{[12]}+\psi_1^*\mu_1^{[12]}\right)\,\D y &\displaystyle
\int_{-\infty}^0\left(\psi_1\mu_3^{[11]}+\psi_1^*\mu_1^{[11]}\right)\,\D y  \end{bmatrix},
\end{equation}
and in the integrands we use the abbreviated notation $\psi_1=\psi_1(y)$, $\psi_1^*=\psi_1(y)^*$, and $\mu^{[jk]}_{1,3}=\mu^{[jk]}_{1,3}(\lambda;y,0)$.  From the exact symmetry $\mathbf{V}^-(\lambda;\varepsilon)=\mathbf{V}^+(\lambda^*;\varepsilon)^\dagger$ that holds for all $\varepsilon\in\mathbb{R}$, we deduce that the coefficients of the convergent power series in $\varepsilon$ of both sides are also related:  $\mathbf{V}_n^-(\lambda)=\mathbf{V}_n^+(\lambda^*)^\dagger$, and therefore if $|\lambda|=r$ and $\Im\{\lambda\}<0$,
\begin{equation}
\mathbf{V}_1^-(\lambda)=\mathbf{V}_1^+(\lambda^*)^\dagger
=\sigma_1\sigma_3\mathbf{T}(\lambda^*)^\dagger \mathbf{V}_0^+(\lambda^*)^\dagger =
\sigma_1\sigma_3\mathbf{T}(\lambda^*)^\dagger\mathbf{V}_0^-(\lambda).
\label{eq:V1-minus}
\end{equation}
Finally, to calculate the expansion of the matrix $\mathbf{V}^\mathbb{R}(\lambda;\epsilon)=\mathbf{V}_0^\mathbb{R}(\lambda)+\varepsilon\mathbf{V}_1^\mathbb{R}(\lambda)+O(\varepsilon^2)$ for $\lambda\in\Sigma_\mathrm{L}\cup\Sigma_\mathrm{R}$, we recall that the jump matrices $\mathbf{V}^\pm(\lambda;\varepsilon)$ make sense not only on the curves $\Sigma_\pm$ but also (via analytic continuation) in the part of the closed half-planes $\overline{\mathbb{C}_\pm}$ exterior to these curves, and that in terms of these the jump matrix on $\Sigma_\mathrm{L}\cup\Sigma_\mathrm{R}$ can be written in factored form as $\mathbf{V}^\mathbb{R}(\lambda;\varepsilon)=\mathbf{V}^-(\lambda;\varepsilon)\mathbf{V}^+(\lambda;\varepsilon)$.  Therefore, combining \eqref{eq:V1-plus} and \eqref{eq:V1-minus} shows that
\begin{equation}
\begin{split}
\mathbf{V}_1^\mathbb{R}(\lambda)&=\mathbf{V}_0^-(\lambda)\mathbf{V}_0^+(\lambda)\mathbf{T}(\lambda)\sigma_3\sigma_1+\sigma_1\sigma_3\mathbf{T}(\lambda)^\dagger\mathbf{V}_0^-(\lambda)\mathbf{V}_0^+(\lambda)  \\ &=
\mathbf{V}_0^\mathbb{R}(\lambda)\mathbf{T}(\lambda)\sigma_3\sigma_1 + \sigma_1\sigma_3\mathbf{T}(\lambda)^\dagger\mathbf{V}_0^\mathbb{R}(\lambda).
\end{split}
\label{eq:V1-R}
\end{equation}

Finally, we make the observation that, whether $\mathbf{T}(\lambda)$ is evaluated for $\lambda\in\Sigma_+$ or for $\lambda\in\Sigma_\mathrm{L}\cup\Sigma_\mathrm{R}$ (by analytic continuation), its elements can be expressed in terms of half-line $L^2$ inner products of $\psi_1$ with squared eigenfunction solutions of the linearized NLS equation \eqref{eq:lNLS} evaluated at $t=0$.  Indeed, taking the squared eigenfunctions $\mu^{[jk]}(\lambda;x,t)$ and $\nu^{[jk]}(\lambda;x,t)$ to be defined from the two columns $\mathbf{u}^{[1]}(\lambda;x,t)$ and $\mathbf{u}^{[2]}(\lambda;x,t)$ of the matrix $\mathbf{U}_{0+}(\lambda;x,t)=\mathbf{U}_0^\mathrm{BC}(\lambda;x,t)$ (evaluated as a limit from $\mathbb{C}_+$ for $\lambda\in\Sigma_\mathrm{L}\cup\Sigma_\mathrm{R}$) as indicated in \eqref{eq:psi-1-first} and \eqref{eq:nu-def}, we see that
\begin{equation}
\begin{split}
T_{11}(\lambda)&=\Im\left\{\int_0^{+\infty}\psi_1(y)\nu^{[22]}(\lambda;y,0)^*\,\D y\right\} + \I
\Im\left\{\int_0^{+\infty}\psi_1(y)\mu^{[22]}(\lambda;y,0)^*\,\D y\right\},\\
T_{12}(\lambda)=T_{21}(\lambda)&=-\Im\left\{\int_0^{+\infty}\psi_1(y)\nu^{[12]}(\lambda;y,0)^*\,\D y\right\} - \I\Im\left\{\int_0^{+\infty}\psi_1(y)\mu^{[12]}(\lambda;y,0)^*\,\D y\right\},\\
T_{22}(\lambda)&=-\Im\left\{\int_{-\infty}^0\psi_1(y)\nu^{[11]}(\lambda;y,0)^*\,\D y\right\} -\I
\Im\left\{\int_{-\infty}^0\psi_1(y)\mu^{[11]}(\lambda;y,0)^*\,\D y\right\}.
\end{split}
\label{eq:T-inner-products}
\end{equation}

\subsubsection{Linearizing the inverse transform}
The jump matrix for \rhref{rhp:M-NZBC} has been shown to have a convergent power series expansion in $\varepsilon$, and we here assume that the convergence holds in the sense of $L^\infty(\Sigma)\cap L^2(\Sigma)$.   Note that the jump matrix on the branch cut $\Sigma_\mathrm{c}$ of $\rho(\lambda)$ is independent of $\varepsilon$.  For $\varepsilon\ll 1$ we use the solution $\mathbf{M}_0(\lambda;x,t)$ for $\varepsilon=0$ as a global parametrix; we therefore define a new unknown $\mathbf{E}(\lambda;x,t)$ to take the place of $\mathbf{M}(\lambda;x,t)$ by
\begin{equation}
\mathbf{E}(\lambda;x,t):=\mathbf{M}(\lambda;x,t)\mathbf{M}_0(\lambda;x,t)^{-1}.
\label{eq:E-define}
\end{equation}
It is easy to check that since $\mathbf{M}(\lambda;x,t)$ and $\mathbf{M}_0(\lambda;x,t)$ both take continuous boundary values on $\Sigma_\mathrm{c}$ and satisfy the same jump condition there, $\mathbf{E}(\lambda;x,t)$ is analytic in the domain $\mathbb{C}\setminus\Sigma'$, where $\Sigma':=\Sigma_+\cup\Sigma_-\cup\Sigma_\mathrm{L}\cup\Sigma_\mathrm{R}=\Sigma\setminus\Sigma_\mathrm{c}$.  Since both factors in \eqref{eq:E-define} tend to the identity matrix $\mathbb{I}$ as $\lambda\to\infty$ in each half-plane, the same is true of $\mathbf{E}(\lambda;x,t)$.  Finally, we may compute the jump conditions satisfied by $\mathbf{E}(\lambda;x,t)$ on its jump contour:
\begin{equation}
\begin{split}
\mathbf{E}_+(\lambda;x,t)&=\mathbf{E}_-(\lambda;x,t)\mathbf{M}_{0-}(\lambda;x,t)\E^{-\I\rho(\lambda)(x+\lambda t)\sigma_3}\mathbf{V}(\lambda;\varepsilon)\mathbf{V}_0(\lambda)^{-1}\E^{\I\rho(\lambda)(x+\lambda t)\sigma_3}\mathbf{M}_{0-}(\lambda;x,t)^{-1}\\
&=\mathbf{E}_-(\lambda;x,t)\mathbf{U}_{0-}(\lambda;x,t)\mathbf{V}(\lambda;\varepsilon)\mathbf{V}_0(\lambda)^{-1}\mathbf{U}_{0-}(\lambda;x,t)^{-1},\quad \lambda\in\Sigma'.
\end{split}
\label{eq:E-jump}
\end{equation}
Now observe that since $\mathbf{V}(\lambda;\varepsilon)=\mathbf{V}_0(\lambda)+\varepsilon\mathbf{V}_1(\lambda)+O(\varepsilon^2)$, the jump matrix for $\mathbf{E}(\lambda;x,t)$ is a small (in $L^\infty(\Sigma)\cap L^2(\Sigma)$) perturbation of the identity matrix when $\varepsilon\ll 1$.  We may write it in the form
\begin{equation}
\mathbf{U}_{0-}(\lambda;x,t)\mathbf{V}(\lambda;\varepsilon)\mathbf{V}_0(\lambda)^{-1}\mathbf{U}_{0-}(\lambda;x,t)^{-1}=\mathbb{I}+\varepsilon\mathbf{W}(\lambda;x,t;\varepsilon),
\end{equation}
where
\begin{equation}
\quad\mathbf{W}(\lambda;x,t;\varepsilon)=\mathbf{U}_{0-}(\lambda;x,t)\mathbf{V}_1(\lambda)\mathbf{V}_0(\lambda)^{-1}\mathbf{U}_{0-}(\lambda;x,t)^{-1}+O(\varepsilon).
\label{eq:W-define}
\end{equation}
This means that $\mathbf{E}(\lambda;x,t)$ satisfies the conditions of a Riemann-Hilbert problem of \emph{small-norm type}.

To solve for $\mathbf{E}(\lambda;x,t)$, we subtract $\mathbf{E}_-(\lambda;x,t)$ from both sides of \eqref{eq:E-jump} and use the Plemelj formula:
\begin{equation}
\mathbf{E}(\lambda;x,t)=\mathbb{I}+\frac{\varepsilon}{2\pi\I}\int_{\Sigma'}\frac{\mathbf{E}_-(z;x,t)\mathbf{W}(z;x,t;\varepsilon)\, \D z}{z-\lambda}.
\label{eq:E-formula}
\end{equation}
Letting $\lambda$ tend to the jump contour from the right yields a closed singular integral equation for the boundary value $\mathbf{E}_-(\lambda;x,t)$; we write this equation in the form:
\begin{equation}
\mathbf{F}(\lambda;x,t)-\varepsilon\mathcal{C}_-[\mathbf{F}(\cdot;x,t)\mathbf{W}(\cdot;x,t;\varepsilon)](\lambda) =
\varepsilon\mathcal{C}_-[\mathbf{W}(\cdot;x,t;\varepsilon)](\lambda),
\label{eq:sing-int-eqn}
\end{equation}
where $\mathbf{F}(\lambda;x,t):=\mathbf{E}_-(\lambda;x,t)-\mathbb{I}$ and where the Cauchy projection operator $\mathcal{C}_-$ is defined by
\begin{equation}
\mathcal{C}_-[\mathbf{H}(\cdot)](\lambda):=\frac{1}{2\pi\I}\int_{\Sigma'}\frac{\mathbf{H}(z)\,\D z}{z-\lambda_-},\quad\lambda\in\Sigma',
\end{equation}
where the subscript on $\lambda$ indicates the nontangential boundary value taken from the right side of $\Sigma'$ by orientation.  Once \eqref{eq:sing-int-eqn} is solved, $\mathbf{E}_-(\lambda;x,t)=\mathbb{I}+\mathbf{F}(\lambda;x,t)$ holds on $\Sigma'$, and then \eqref{eq:E-formula}
gives the solution of the Riemann-Hilbert problem.
It is well-known that $\mathcal{C}_-$ is bounded on $L^2(\Sigma')$ with norm independent of $\varepsilon$ (the norm only has to do with the geometry of $\Sigma'$), so the fact that $\|\mathbf{W}(\cdot;x,t;\varepsilon)\|_2$ and $\|\mathbf{W}(\cdot;x,t;\varepsilon)\|_\infty$ are both bounded independent of $\varepsilon\ll 1$ means that \eqref{eq:sing-int-eqn} can be solved by Neumann series with the result that $\|\mathbf{F}(\cdot;x,t)\|_2=O(\varepsilon)$.  Using this result in \eqref{eq:E-formula} gives
\begin{equation}
\mathbf{E}(\lambda;x,t)=\mathbb{I}+\frac{\varepsilon}{2\pi\I}\int_\Sigma\frac{\mathbf{W}(z;x,t;0)\,\D z}{z-\lambda} +O(\varepsilon^2\lambda^{-1}).
\end{equation}
The solution $\psi(x,t;\varepsilon)$ of the focusing NLS equation corresponding to the jump matrix $\mathbf{V}(\lambda;x,t;\varepsilon)$ and hence to the initial data $\psi(x,0;\varepsilon)=\psi_0(x)+\varepsilon\psi_1(x)$ is (cf., \eqref{eq:NLS-recovery})
\begin{equation}
\begin{split}
\psi(x,t;\varepsilon)&=2\I\lim_{\lambda\to\infty}\lambda M_{12}(\lambda;x,t)\\
&=2\I\lim_{\lambda\to\infty}\lambda \left[\mathbf{E}(\lambda;x,t)\mathbf{M}_0(\lambda;x,t)\right]_{12}\\
&=2\I\lim_{\lambda\to\infty}\lambda \left[E_{11}(\lambda;x,t)M_{0,12}(\lambda;x,t) +
E_{12}(\lambda;x,t)M_{0,22}(\lambda;x,t)\right]\\
&=\psi_0(x,t)+2\I\lim_{\lambda\to\infty}\lambda E_{12}(\lambda;x,t)\\
&=\psi_0(x,t) + \varepsilon\psi_1(x,t)+ O(\varepsilon^2),
\end{split}
\end{equation}
where
\begin{equation}
\psi_1(x,t):=-\frac{1}{\pi}\int_{\Sigma'}W_{12}(\lambda;x,t;0)\,\D \lambda.
\label{eq:psi-1-integral}
\end{equation}

We now display this formula as a superposition of squared eigenfunction solutions of the lNLS equation \eqref{eq:lNLS} linearized about the solution $\psi_0(x,t)$ of the NLS equation.  First we consider $\lambda\in\Sigma_\mathrm{L}\cup\Sigma_\mathrm{R}\subset\mathbb{R}$.  Then, using \eqref{eq:V1-R} and $\mathbf{U}_{0+}(\lambda;x,t)=\mathbf{U}_{0-}(\lambda;x,t)\mathbf{V}_0^\mathbb{R}(\lambda)$,
\begin{equation}
\begin{split}
W_{12}(\lambda;x,t;0)&=\left[\mathbf{U}_{0-}(\lambda;x,t)\mathbf{V}^\mathbb{R}_1(\lambda)\mathbf{V}^\mathbb{R}_0(\lambda)^{-1}
\mathbf{U}_{0-}(\lambda;x,t)^{-1}\right]_{12}\\
&=\left[\mathbf{U}_{0+}(\lambda;x,t)\mathbf{T}(\lambda)\sigma_3\sigma_1
\mathbf{U}_{0+}(\lambda;x,t)^{-1}\right]_{12} \\
&\quad\quad\quad{}+
\left[\mathbf{U}_{0-}(\lambda;x,t)\sigma_1\sigma_3\mathbf{T}(\lambda)^\dagger\mathbf{U}_{0-}(\lambda;x,t)^{-1}\right]_{12}.
\end{split}
\end{equation}
Recall also the symmetry $\mathbf{M}(\lambda^*;x,t)^\dagger=\mathbf{M}(\lambda;x,t)$, which implies that
$\mathbf{U}_0(\lambda^*;x,t)^\dagger=\mathbf{U}_0(\lambda;x,t)$, so letting $\lambda$ tend to $\Sigma_\mathrm{L}\cup\Sigma_\mathrm{R}\subset\mathbb{R}$ from either side we see that $\mathbf{U}_{0-}(\lambda;x,t)=\mathbf{U}_{0+}(\lambda;x,t)^{-\dagger}$ holds for $\lambda\in\Sigma_\mathrm{L}\cup\Sigma_\mathrm{R}$.  Therefore
\begin{multline}
W_{12}(\lambda;x,t;0)=
\left[\mathbf{U}_{0+}(\lambda;x,t)\mathbf{T}(\lambda)\sigma_3\sigma_1
\mathbf{U}_{0+}(\lambda;x,t)^{-1}\right]_{12} \\
{}+
\left[\mathbf{U}_{0+}(\lambda;x,t)^{-\dagger}\sigma_1\sigma_3\mathbf{T}(\lambda)^\dagger\mathbf{U}_{0+}(\lambda;x,t)^\dagger\right]_{12}.
\label{eq:W12-real-first}
\end{multline}
Letting $\mathbf{u}^{[k]}(\lambda;x,t)$, $k=1,2$, denote the columns of the matrix $\mathbf{U}_{0+}(\lambda;x,t)$, and using the fact that $\mathbf{U}_{0+}(\lambda;x,t)$ has unit determinant,
\begin{equation}
\begin{split}
W_{12}(\lambda;x,t;0)&=\begin{bmatrix} u^{[1]}_1(\lambda;x,t) & u^{[2]}_1(\lambda;x,t)\end{bmatrix}\mathbf{T}(\lambda)\begin{bmatrix} u^{[1]}_1(\lambda;x,t)\\ u^{[2]}_1(\lambda;x,t)\end{bmatrix} \\
&\quad\quad{}-
\begin{bmatrix}u^{[1]}_2(\lambda;x,t)^* & u^{[2]}_2(\lambda;x,t)^*\end{bmatrix}\mathbf{T}(\lambda)^\dagger\begin{bmatrix} u_2^{[1]}(\lambda;x,t)^*\\u_2^{[2]}(\lambda;x,t)^*\end{bmatrix}\\
&=\begin{bmatrix} u^{[1]}_1(\lambda;x,t) & u^{[2]}_1(\lambda;x,t)\end{bmatrix}\mathbf{T}(\lambda)\begin{bmatrix} u^{[1]}_1(\lambda;x,t)\\ u^{[2]}_1(\lambda;x,t)\end{bmatrix} \\
&\quad\quad{}-
\begin{bmatrix}u^{[1]}_2(\lambda;x,t)^* & u^{[2]}_2(\lambda;x,t)^*\end{bmatrix}\mathbf{T}(\lambda)^*
\begin{bmatrix}u^{[1]}_2(\lambda;x,t)^* \\ u^{[2]}_2(\lambda;x,t)^*\end{bmatrix}.
\end{split}
\end{equation}
Writing $\mathbf{T}(\lambda)=\mathbf{A}(\lambda)+\I \mathbf{B}(\lambda)$ with the elements of the matrix functions $\mathbf{A}$ and $\mathbf{B}$ being real, we arrive at
\begin{equation}
W_{12}(\lambda;x,t;0)=\sum_{j,k=1}^2A_{jk}(\lambda)\mu^{[jk]}(\lambda;x,t) +\sum_{j,k=1}^2
B_{jk}(\lambda)\nu^{[jk]}(\lambda;x,t),
\end{equation}
where $\mu^{[jk]}(\lambda;x,t)$ and $\nu^{[jk]}(\lambda;x,t)$ are precisely the squared eigenfunction solutions of the linearized NLS equation \eqref{eq:lNLS} defined respectively by \eqref{eq:psi-1-first} and \eqref{eq:nu-def} in terms of the two columns $\mathbf{u}^{[1]}(\lambda;x,t)$ and $\mathbf{u}^{[2]}(\lambda;x,t)$ of the boundary value $\mathbf{U}_{0+}(\lambda;x,t)$ taken by the matrix $\mathbf{U}_0(\lambda;x,t)$ from the left on $\Sigma_\mathrm{L}\cup\Sigma_\mathrm{R}$.
Since $\mu^{[jk]}(\lambda;x,t)$ and $\nu^{[jk]}(\lambda;x,t)$ satisfy \eqref{eq:lNLS} for each $\lambda\in\Sigma_\mathrm{L}\cup\Sigma_\mathrm{R}$, and since the space of solutions is closed under taking real linear combinations, it follows that the contributions from $\lambda\in\Sigma_\mathrm{L}\cup\Sigma_\mathrm{R}\subset\mathbb{R}$ to $\psi_1(x,t)$ given by the formula \eqref{eq:psi-1-integral} formally satisfy \eqref{eq:lNLS}.

Next consider the contributions to \eqref{eq:psi-1-integral} from $\Sigma_+\cup\Sigma_-$.  We parametrize $\Sigma_+$ by $\lambda=\Lambda(u)$, $\underline{u}<u<\overline{u}$, and by Schwarz symmetry we similarly parametrize $\Sigma_-$ by $\lambda=\Lambda(u)^*$, $\underline{u}<u<\overline{u}$.  Therefore
\begin{equation}
\int_{\Sigma_+\cup\Sigma_-}W_{12}(\lambda;x,t;0)\,\D\lambda = \int_{\underline{u}}^{\overline{u}}\left[W_{12}(\Lambda(u);x,t;0)\Lambda'(u) + W_{12}(\Lambda(u)^*;x,t;0)\Lambda'(u)^*\right]\,\D u.
\end{equation}
Substituting from \eqref{eq:V1-plus}, \eqref{eq:V1-minus}, and \eqref{eq:W-define}, and using $\mathbf{U}_{0+}(\lambda;x,t)=\mathbf{U}_{0-}(\lambda;x,t)\mathbf{V}_0(\lambda)$ gives
\begin{multline}
W_{12}(\Lambda(u);x,t;0)\Lambda'(u) + W_{12}(\Lambda(u)^*;x,t;0)\Lambda'(u)^*\\
\begin{aligned}
&= \left[\mathbf{U}_{0+}(\Lambda(u);x,t)\widetilde{\mathbf{T}}(u)\sigma_3\sigma_1\mathbf{U}_{0+}(\Lambda(u);x,t)^{-1}\right]_{12} \\
&\quad\quad\quad{}+
\left[\mathbf{U}_{0-}(\Lambda(u)^*;x,t)\sigma_1\sigma_3\widetilde{\mathbf{T}}(u)^\dagger\mathbf{U}_{0-}(\Lambda(u)^*;x,t)^{-1}\right]_{12}\\
&=\left[\mathbf{U}_{0+}(\Lambda(u);x,t)\widetilde{\mathbf{T}}(u)\sigma_3\sigma_1\mathbf{U}_{0+}(\Lambda(u);x,t)^{-1}\right]_{12} \\
&\quad\quad\quad{}+
\left[\mathbf{U}_{0+}(\Lambda(u);x,t)^{-\dagger}\sigma_1\sigma_3\widetilde{\mathbf{T}}(u)^\dagger\mathbf{U}_{0+}(\Lambda(u);x,t)^\dagger\right]_{12},
\end{aligned}
\end{multline}
where $\widetilde{\mathbf{T}}(u):=\mathbf{T}(\Lambda(u))\Lambda'(u)$, and where in the second line we used the Schwarz symmetry $\mathbf{U}(\lambda^*;x,t)=\mathbf{U}(\lambda;x,t)^{-\dagger}$, which also has the effect of switching the boundary value from the right of $\Sigma_-$ to the left of $\Sigma_+$.  Comparing now with \eqref{eq:W12-real-first} we see that the rest of the argument employed in the case of $\lambda\in\Sigma_\mathrm{L}\cup\Sigma_\mathrm{R}$ goes through essentially unchanged, with the result that
\begin{multline}
W_{12}(\Lambda(u);x,t;0)\Lambda'(u) + W_{12}(\Lambda(u)^*;x,t;0)\Lambda'(u)^* \\=
\sum_{j,k=1}^2A_{jk}^+(u)\mu^{[jk]}(\Lambda(u);x,t)+
\sum_{j,k=1}^2B_{jk}^+(u)\nu^{[jk]}(\Lambda(u);x,t)
\end{multline}
where $\widetilde{\mathbf{T}}(u)=\mathbf{A}^+(u)+\I \mathbf{B}^+(u)$ with $\mathbf{A}^+(\cdot)$ and $\mathbf{B}^+(\cdot)$ being real matrix-valued functions on $(0,1)$, and where $\mu^{[jk]}(\lambda;x,t)$ and $\nu^{[jk]}(\lambda;x,t)$ are again solutions of \eqref{eq:lNLS} constructed via \eqref{eq:psi-1-first} and \eqref{eq:nu-def} respectively, but now using the columns $\mathbf{u}^{[1]}(\lambda;x,t)$ and $\mathbf{u}^{[2]}(\lambda;x,t)$ of the boundary value $\mathbf{U}_{0+}(\lambda;x,t)$ taken from the left by the matrix $\mathbf{U}_0(\lambda;x,t)$ on the arc $\Sigma_+$.
It therefore follows that again the contributions to \eqref{eq:psi-1-integral} from $\Sigma_+\cup\Sigma_-$ taken together formally satisfy the focusing linearized NLS equation \eqref{eq:lNLS}.

\subsubsection{Summary}
The Cauchy problem for the linearized NLS equation \eqref{eq:lNLS} (based on the solution $\psi_0(x,t)$ of the NLS equation \eqref{eq:NLS} for which $\psi_0-1\in L^1(\mathbb{R})$) with initial data $\psi_1(x)$ having sufficiently rapid decay can be obtained formally as follows.
\begin{itemize}
\item For $\lambda\in\Sigma_\mathrm{L}\cup\Sigma_+\cup\Sigma_\mathrm{R}$, let $\mathbf{M}_{0+}(\lambda;x,t)$ denote the boundary value taken from the left of the simultaneous solution of the direct/inverse problem for the solution $\psi_0(x,t)$ of \eqref{eq:NLS}.  Taking $(\mathbf{u}^{[1]}(\lambda;x,t),\mathbf{u}^{[2]}(\lambda;x,t))=\mathbf{U}_{0+}(\lambda;x,t)=\mathbf{M}_{0+}(\lambda;x,t)\E^{-\I\rho(\lambda)(x+\lambda t)\sigma_3}$, construct from the columns $\mathbf{u}^{[j]}(\lambda;x,t)$, $j=1,2$, the squared eigenfunction solutions $\mu^{[11]}(\lambda;x,t)$, $\mu^{[12]}(\lambda;x,t)=\mu^{[21]}(\lambda;x,t)$, and $\mu^{[22]}(\lambda;x,t)$ from \eqref{eq:psi-1-first} and
  $\nu^{[11]}(\lambda;x,t)$, $\nu^{[12]}(\lambda;x,t)=\nu^{[21]}(\lambda;x,t)$, and $\nu^{[22]}(\lambda;x,t)$ from \eqref{eq:nu-def} for $\lambda\in\Sigma_\mathrm{L}\cup\Sigma_+\cup\Sigma_{\mathrm{L}}$.
\item Evaluating the squared eigenfunction solutions at $t=0$, compute the inner products with the Cauchy data $\psi_1$ for the linearized problem necessary to obtain the transform matrix $\mathbf{T}(\lambda)$ from \eqref{eq:T-inner-products} for $\lambda\in\Sigma_\mathrm{L}\cup\Sigma_+\cup\Sigma_\mathrm{R}$.
\item
The solution $\psi_1(x,t)$ of the linearized Cauchy problem is then $\psi_1(x,t)=\psi_1^\mathbb{R}(x,t)+\psi_1^+(x,t)$, where we have the explicit integral representations
\begin{equation}
\begin{split}
\psi_1^\mathbb{R}(x,t)&:=-\frac{1}{\pi}\sum_{j,k=1}^2
\int_{\Sigma_\mathrm{L}\cup\Sigma_\mathrm{R}}\left(\Re\{T_{jk}(\lambda)\}\mu^{[jk]}(\lambda;x,t)+
\Im\{T_{jk}(\lambda)\}\nu^{[jk]}(\lambda;x,t)\right)\,\D\lambda \\
\psi_1^+(x,t)&:=-\frac{1}{\pi}\sum_{j,k=1}^2
\int_{\underline{u}}^{\overline{u}}\left(\Re\{T_{jk}(\Lambda(u))\Lambda'(u)\}\mu^{[jk]}(\Lambda(u);x,t)\right.\\
&\quad\quad\quad\quad\left.{}+ \Im\{T_{jk}(\Lambda(u))\Lambda'(u)\}\nu^{[jk]}(\Lambda(u);x,t)\right)\,\D u
\end{split}
\label{eq:psi1-pieces}
\end{equation}
with $\Lambda:(\underline{u},\overline{u})\to\Sigma_+$ being a piecewise-smooth parametrization of $\Sigma_+$.
\end{itemize}

\subsubsection{Simplification of $\psi_1^\mathbb{R}(x,t)$ using Schwarz symmetry}
When $\lambda\in\mathbb{R}$, the six squared eigenfunctions $\mu^{[11]}(\lambda;x,t)$, $\mu^{[12]}(\lambda;x,t)=\mu^{[21]}(\lambda;x,t)$, $\mu^{[22]}(\lambda;x,t)$, $\nu^{[11]}(\lambda;x,t)$, $\nu^{[12]}(\lambda;x,t)=\nu^{[21]}(\lambda;x,t)$, and $\nu^{[22]}(\lambda;x,t)$, are not linearly independent over the real numbers, leading to some simplifications in the formula for $\psi_1^{\mathbb{R}}(x,t)$.  To obtain the relations among the solutions, we eliminate $\mathbf{U}_{0-}(\lambda;x,t)$ between the jump condition $\mathbf{U}_{0+}(\lambda;x,t)=\mathbf{U}_{0-}(\lambda;x,t)\mathbf{V}_0^\mathbb{R}(\lambda)$ and the Schwarz symmetry condition $\mathbf{U}_{0+}(\lambda;x,t)^\dagger=\mathbf{U}_{0-}(\lambda;x,t)^{-1}$, both of which hold for $\lambda\in\Sigma_\mathrm{L}\cup\Sigma_\mathrm{R}$. Expressing $\mathbf{V}_0^\mathbb{R}(\lambda)$ in terms of the reflection coefficient $R_0(\lambda)$ for the unperturbed initial condition $\psi_0(x,0)$ by
\begin{equation}
\mathbf{V}_0^\mathbb{R}(\lambda)=\begin{bmatrix}1+|R_0(\lambda)|^2 & R_0(\lambda)^*\\R_0(\lambda) & 1
\end{bmatrix},
\end{equation}
we systematically write all six of the squared eigenfunction solutions for a given $\lambda\in\Sigma_\mathrm{L}\cup\Sigma_\mathrm{R}$ in terms of the others, obtaining a $6\times 6$ real-linear homogeneous system.  A tedious but straightforward row reduction shows that the rank of this system is exactly $3$, containing only the following three independent relations:
\begin{equation}
\begin{split}
\mu^{[11]}(\lambda;x,t)&=2\Re\{R_0(\lambda)\}\mu^{[12]}(\lambda;x,t)-(1+|R_0(\lambda)|^2)\mu^{[22]}(\lambda;x,t)\\
\nu^{[11]}(\lambda;x,t)&=-2\Im\{R_0(\lambda)\}\mu^{[12]}(\lambda;x,t)+(1+|R_0(\lambda)|^2)\nu^{[22]}(\lambda;x,t)\\
\nu^{[12]}(\lambda;x,t)&=-\Im\{R_0(\lambda)\}\mu^{[22]}(\lambda;x,t)+\Re\{R_0(\lambda)\}\nu^{[22]}(\lambda;x,t).
\end{split}
\label{eq:real-relations}
\end{equation}
It follows that for $\lambda\in\Sigma_\mathrm{L}\cup\Sigma_\mathrm{R}$,
\begin{multline}
\sum_{j,k=1}^2\left(\Re\{T_{jk}(\lambda)\}\mu^{[jk]}(\lambda;x,t)+\Im\{T_{jk}(\lambda)\}\nu^{[jk]}(\lambda;x,t)\right)\\
=\left[(1+|R_0(\lambda)|^2)\Im\{T_{11}(\lambda)\} + 2\Re\{R_0(\lambda)\}\Im\{T_{12}(\lambda)\} + \Im\{T_{22}(\lambda)\}\right]\nu^{[22]}(\lambda;x,t)\\
{}+\left[2\Re\{R_0(\lambda)\}\Re\{T_{11}(\lambda)\}-2\Im\{R_0(\lambda)\}\Im\{T_{11}(\lambda)\} + 2\Re\{T_{12}(\lambda)\}\right]\mu^{[12]}(\lambda;x,t)\\
{}+\left[-(1+|R_0(\lambda)|^2)\Re\{T_{11}(\lambda)\}-2\Im\{R_0(\lambda)\}\Im\{T_{12}(\lambda)\}+\Re\{T_{22}(\lambda)\}\right]\mu^{[22]}(\lambda;x,t),
\end{multline}
where we have also used $T_{21}=T_{12}$, $\mu^{[jk]}=\mu^{[kj]}$, and $\nu^{[jk]}=\nu^{[kj]}$.  Now combining the relations \eqref{eq:real-relations} with \eqref{eq:T-inner-products}, we observe the following identities:
\begin{multline}
(1+|R_0(\lambda)|^2)\Im\{T_{11}(\lambda)\} + 2\Re\{R_0(\lambda)\}\Im\{T_{12}(\lambda)\} + \Im\{T_{22}(\lambda)\}\\
= -\Im\left\{\int_\mathbb{R}\psi_1(y)\mu^{[11]}(\lambda;y,0)^*\,\D y\right\},
\end{multline}
\begin{equation}
2\Re\{R_0(\lambda)\}\Re\{T_{11}(\lambda)\}-2\Im\{R_0(\lambda)\}\Im\{T_{11}(\lambda)\} + 2\Re\{T_{12}(\lambda)\}=0,
\end{equation}
and
\begin{multline}
-(1+|R_0(\lambda)|^2)\Re\{T_{11}(\lambda)\}-2\Im\{R_0(\lambda)\}\Im\{T_{12}(\lambda)\}+\Re\{T_{22}(\lambda)\}\\
=-\Im\left\{\int_\mathbb{R}\psi_1(y)\nu^{[11]}(\lambda;y,0)^*\,\D y\right\}.
\end{multline}
Therefore, from the definition \eqref{eq:psi1-pieces} we see that $\psi_1^\mathbb{R}(x,t)$ takes a simpler form:
\begin{equation}
\psi_1^\mathbb{R}(x,t)=\int_{\Sigma_\mathrm{L}\cup\Sigma_\mathrm{R}}\left(C_\mu(\lambda)\mu^{[22]}(\lambda;x,t) + C_\nu(\lambda)\nu^{[22]}(\lambda;x,t)\right)\,\D\lambda,
\end{equation}
where the real-valued coefficients $C_\mu(\lambda)$ and $C_\nu(\lambda)$ are now given by explicit \emph{whole-line} spectral transforms:
\begin{equation}
\begin{split}
C_\mu(\lambda)&:=\frac{1}{\pi}\Im\left\{\int_\mathbb{R}\psi_1(y)\nu^{[11]}(\lambda;y,0)^*\,\D y\right\}\\
C_\nu(\lambda)&:=\frac{1}{\pi}\Im\left\{\int_\mathbb{R}\psi_1(y)\mu^{[11]}(\lambda;y,0)^*\,\D y\right\}.
\end{split}
\end{equation}

\subsection{Solution of the linearized equation for the Peregrine solution}
Consider the case that the base solution of the NLS equation \eqref{eq:NLS} about which we will linearize is the Peregrine breather
\begin{equation}
     \psi_0(x,t)=\psi_{P}(x,t):=1-4\frac{1+2\I (t-t_0)}{1+4(x-x_0)^2+4(t-t_0)^2}
     \label{eq:Peregrine-recall}
\end{equation}
centered at $(x_0,t_0)$ in the $(x,t)$-plane.  Without loss of generality, we will take $x_0=0$ but let $t_0$ be arbitrary.  We write $\tilde{t}:=t-t_0$ below.
As we have seen in Section~\ref{s:OneDarboux} this solution can be obtained in the setting of the robust IST by applying a standard Darboux transformation with suitable parameters to the Riemann-Hilbert matrix that represents the background state $\psi(x,t)\equiv 1$.  The Darboux transformation also supplies the matrix $\mathbf{U}_0(\lambda;x,t)=\mathbf{M}_0(\lambda;x,t)\E^{-\I\rho(\lambda)(x+\lambda t)\sigma_3}$ associated with the solution \eqref{eq:Peregrine-recall} for $x_0=0$.
For $\lambda$ in the domain $D_+$, the matrix $\mathbf{U}_0(\lambda;x,t)$ takes the form
\begin{multline}
\mathbf{U}_0(\lambda;x,t)\\
=\frac{n(\lambda)}{\rho(\lambda)(1+4x^2+4\tilde{t}^2)}
\left[
\begin{matrix}
2\I\lambda(2\tilde{t}-\I)+\rho(\lambda)(2\tilde{t}-\I)^2+4x(\rho(\lambda)x-\I) \\
4\rho(\lambda)(\lambda-\rho(\lambda))\tilde{t}^2-4\I\tilde{t}+(\lambda-\rho(\lambda))(\rho(\lambda)+4x(\rho(\lambda)x-\I))+2
\end{matrix}\right.\\
\left.\begin{matrix}
4\rho(\lambda)(\lambda-\rho(\lambda))\tilde{t}^2+4\I\tilde{t}+(\lambda-\rho(\lambda))(\rho(\lambda)+2x(\rho(\lambda)x+\I))+2 \\ -2\I\lambda(2\tilde{t}+\I)+\rho(\lambda)(2\tilde{t}+\I)^2+4x(\rho(\lambda)x+\I)
\end{matrix}\right]
\E^{\I\rho(\lambda)(x+\lambda t)\sigma_3}.
\label{eq:U0-Peregrine}
\end{multline}
We observe that this formula admits analytic continuation to the whole complex plane with the branch cut $\Sigma_\mathrm{c}$ of $\rho$ and $n$ omitted.  This implies that, provided $\varepsilon\psi_1(x)$ is such that it introduces no further singularities in the solution for the perturbed potential $\psi_0(x)+\varepsilon\psi_1(x)$, we may replace the contour $\Sigma_+\cup\Sigma_-$ with a ``dogbone'' contour that lies against the two sides of the branch cut $\Sigma_\mathrm{c}$ except near the endpoints $\lambda=\pm\I$ where it is augmented by small circles of radius $\delta\ll 1$.

Such a deformation of $\Sigma_\pm$ also results in $\Sigma_\mathrm{L}$ being extended to $(-\infty,0)$ and $\Sigma_\mathrm{R}$ being extended to $(0,+\infty)$.  The contribution $\psi_1^\mathbb{R}(x,t)$ therefore takes the form
\begin{equation}
\psi_1^\mathbb{R}(x,t)=\int_\mathbb{R} \left(C_\mu(\lambda)\mu^{[22]}(\lambda;x,t)+C_\nu(\lambda)\nu^{[22]}(\lambda;x,t)\right)\,\D\lambda,
\end{equation}
where, applying \eqref{eq:real-relations} with $R_0(\lambda)=0$ for all $\lambda\in\mathbb{R}\setminus\{0\}$,
\begin{equation}
\begin{split}
C_\mu(\lambda)&=\frac{1}{\pi}\Im\left\{\int_\mathbb{R}\psi_1(y)\nu^{[22]}(\lambda;y,0)^*\,\D y\right\}\\
C_\nu(\lambda)&=-\frac{1}{\pi}\Im\left\{\int_\mathbb{R}\psi_1(y)\mu^{[22]}(\lambda;y,0)^*\,\D y\right\}.
\end{split}
\label{eq:C-mu-nu}
\end{equation}
In this situation, we therefore have a very compelling representation of $\psi_1^\mathbb{R}(x,t)$ involving $L^2(\mathbb{R})$ projections onto the functions $\mu^{[22]}(\lambda;x,0)$ and $\nu^{[22]}(\lambda;x,0)$ followed by superposition with the very same solutions now evaluated at $t\neq 0$.

Now $\Sigma_+$ has been deformed to consist of three parts joined end-to-end:  the vertical contour connecting $\lambda=0$ with $\lambda=\I (1-\delta)$ and lying on the left side of $\Sigma_\mathrm{c}$, followed by the negatively-oriented circle of radius $\delta$ centered at $\lambda=\I$, followed by the vertical contour connecting $\lambda=\I (1-\delta)$ with $\lambda=0$ and lying on the right side of $\Sigma_\mathrm{c}$.  Let us call the contributions to $\psi_1^+(x,t)$ arising from integration on the circle $\psi_1^\circ(x,t)$ and set $\psi_1^\|(x,t):=\psi_1^+(x,t)-\psi_1^\circ(x,t)$.  We also adopt the subscript notation $+$/$-$ for the boundary values of $\mathbf{U}_0(\lambda;x,t)$ and the corresponding squared eigenfunctions taken on $\Sigma_\mathrm{c}$ from the left/right half-plane.  Then we have the jump condition $\mathbf{U}_{0+}(\I s;x,t)=\mathbf{U}_{0-}(\I s;x,t)\I\sigma_1$ for $0<s<1$, and it follows that the squared eigenfunctions satisfy
\begin{equation}
\begin{aligned}
\mu_\pm^{[11]}(\I s;x,t)&=-\mu_\mp^{[22]}(\I s;x,t);\quad\mu_+^{[12]}(\I s;x,t)&=-\mu_-^{[12]}(\I s;x,t);\\
\nu_\pm^{[11]}(\I s;x,t)&=-\nu_\mp^{[22]}(\I s;x,t);\quad \nu_+^{[12]}(\I s;x,t) &= -\nu_-^{[12]}(\I s;x,t)
\end{aligned}
\label{eq:squared-ef-cut}
\end{equation}
for $0<s<1$.  Taking $\Lambda(u)=\I u$, $0<u<1-\delta$ for the parametrization of the upward contribution to $\psi_1^\|(x,t)$ and $\Lambda(u)=\I(1-\delta-u)$, $0<u<1-\delta$ for the parametrization of the downward contribution to $\psi_1^\|(x,t)$, combining \eqref{eq:T-inner-products}, \eqref{eq:psi1-pieces}, and \eqref{eq:squared-ef-cut} gives
\begin{multline}
\psi_1^\|(x,t)=\int_0^{1-\delta}\left(D_\mu^{[11]}(u)\mu_+^{[11]}(\I u;x,t) + D_\mu^{[22]}(u)\mu_+^{[22]}(\I u;x,t) \right.\\
{}\left.+
D_\nu^{[11]}(u)\nu_+^{[11]}(\I u;x,t) + D_\nu^{[22]}(u)\nu_+^{[22]}(\I u;x,t)\right)\,\D u,
\end{multline}
where for $0<u<1-\delta$,
\begin{equation}
\begin{split}
D_\mu^{[11]}(u)&:=\frac{1}{\pi}\Im\left\{\int_\mathbb{R}\psi_1(y)\mu_+^{[22]}(\I u;y,0)^*\,\D y\right\}\\
D_\mu^{[22]}(u)&:=-\frac{1}{\pi}\Im\left\{\int_\mathbb{R}\psi_1(y)\mu_+^{[11]}(\I u;y,0)^*\,\D y\right\}\\
D_\nu^{[11]}(u)&:=-\frac{1}{\pi}\Im\left\{\int_\mathbb{R}\psi_1(y)\nu_+^{[22]}(\I u;y,0)^*\,\D y\right\}\\
D_\nu^{[22]}(u)&:=\frac{1}{\pi}\Im\left\{\int_\mathbb{R}\psi_1(y)\nu_+^{[11]}(\I u;y,0)^*\,\D y\right\}.
\end{split}
\end{equation}
Therefore, the contribution $\psi_1^\|(x,t)$ is represented also as a real linear combination of squared eigenfunction solutions of the linearized problem with spectral transform coefficients written in terms of $L^2(\mathbb{R})$ inner products of the initial condition $\psi_1$ with the same solutions at $t=0$.  The only part of the solution formula $\psi_1(x,t)=\psi_1^\mathbb{R}(x,t)+\psi_1^\|(x,t)+\psi_1^\circ(x,t)$ left in terms of transform integrals of $\psi_1$ integrated against the squared eigenfunctions on the half line is $\psi_1^\circ(x,t)$.  Taking the limit $\delta\downarrow 0$ is generally not feasible, because from the factor $n(\lambda)/\rho(\lambda)$ in \eqref{eq:U0-Peregrine} we see that the quadratic forms $\mu^{[jk]}(\lambda;x,t)$ and $\nu^{[jk]}(\lambda;x,t)$ formed from the elements of the matrix $\mathbf{U}_0(\lambda;x,t)$ blow up proportional to $|\lambda-\I|^{-3/2}$, and hence $\psi_1^\circ(x,t)=O(\delta^{-2})$ as $\delta\downarrow 0$.

We can furthermore observe that, due to the simple exponential dependence of the columns of the matrix $\mathbf{U}_0(\lambda;x,t)$ given in \eqref{eq:U0-Peregrine}, for fixed $t\in\mathbb{R}$ the quadratic forms $\mu^{[jk]}(\lambda;x,t)$ and $\nu^{[jk]}(\lambda;x,t)$ are uniformly bounded for $x\in\mathbb{R}$
whenever $\rho(\lambda)\in\mathbb{R}$.  The same quadratic forms are uniformly bounded for $(x,t)\in\mathbb{R}^2$ whenever both $\rho(\lambda)\in\mathbb{R}$ and $\lambda\in\mathbb{R}$.  On the other hand, if $\rho(\lambda)\in\mathbb{R}$ but $\Im\{\lambda\}\neq 0$, then the quadratic forms are bounded with respect to $x$ but grow exponentially in $t$.  Moreover, the above bounds also hold uniformly as $\lambda\to\infty$.  These results imply that under reasonable conditions the component $\psi_1^\mathbb{R}(x,t)$ of the solution $\psi_1(x,t)$ of the linearized NLS equation \eqref{eq:lNLS} is uniformly bounded for all time.
\begin{theorem}
Suppose that the initial condition $\psi_1(\cdot)$ is such that $C_\mu(\cdot)$ and $C_\nu(\cdot)$ defined by \eqref{eq:C-mu-nu} are in $L^1(\mathbb{R})$.  Then $\psi_1^\mathbb{R}(x,t)$ is uniformly bounded for $(x,t)\in\mathbb{R}^2$.
\label{theorem:partial-stability}
\end{theorem}
\begin{proof}
This is a straightforward estimate, using the fact that $\lambda\in\mathbb{R}$ implies $\rho(\lambda)\in\mathbb{R}$:
\begin{equation}
|\psi_1^\mathbb{R}(x,t)|\le \int_\mathbb{R}\left(|C_\mu(\lambda)||\mu^{[22]}(\lambda;x,t)| + |C_\nu(\lambda)||\nu^{[22]}(\lambda;x,t)|\right)\,\D\lambda \le K\int_\mathbb{R}\left(|C_\mu(\lambda)|+|C_\nu(\lambda)|\right)\,\D\lambda,
\end{equation}
where $K$ is a uniform upper bound for $|\mu^{[22]}(\lambda;x,t)|$ and $|\nu^{[22]}(\lambda;x,t)|$.
\end{proof}
On the other hand, the component $\psi_1^\|(x,t)$ can grow exponentially in $t$.  \emph{This is for exactly the same reason that the solution of the linearized problem for the background solution $\psi_0(x,t)\equiv 1$ typically grows exponentially.}  The squared eigenfunctions $\mu_+^{[11]}(\I u;x,t)$, $\mu_+^{[22]}(\I u;x,t)$, $\nu_+^{[11]}(\I u;x,t)$, and $\nu_+^{[22]}(\I u;x,t)$, while bounded in $x$ for $0<u<1$ because $\rho_+(\I u)=-\sqrt{1-u^2}$, exhibit exponential growth in $t$ because $\rho_+(\I u)\lambda=-\I u\sqrt{1-u^2}$; the exponential factors that appear are precisely the same ones as in the case of linearization about $\psi_0(x,t)\equiv 1$ for which the analogue of \eqref{eq:U0-Peregrine} is simply $\mathbf{U}_0(\lambda;x,t)=\mathbf{E}(\lambda)\E^{\I\rho(\lambda)(x+\lambda t)\sigma_3}$.
The linearization about the Peregrine solution $\psi_0(x,t)=\psi_\mathrm{P}(x,t)$ therefore predicts  instability of exactly the same sort as in the case $\psi_0(x,t)\equiv 1$, namely the (modulational) instability of the background field on which the Peregrine solution rests.  However, Theorem~\ref{theorem:partial-stability} predicts linearized stability for suitable perturbations of the Peregrine solution $\psi_0(x,t)=\psi_\mathrm{P}(x,t)$ for which the components $\psi_1^\|(x,t)$ and $\psi_1^\circ(x,t)$ vanish identically (due to the vanishing of the corresponding integral transforms in the integrand).

\appendix

\section{A Priori Estimates of Jost Solutions}
\begin{lemma}[A priori estimates]
\label{lemma:Neumann}
Suppose that $\psi(x)$ is defined for all $x\in\mathbb{R}$ with $\Delta\psi = \psi - 1 \in L^{1}(\mathbb{R})$ and let $|\lambda|\geq s>1$. Then we have the following a priori bounds for the indicated columns of $\mathbf{K}^{\pm}(\lambda;x)$ given in \eqref{eq:Volterra-NZBC}:
\begin{align}
\sup_{x\in\mathbb{R}}\| \mathbf{k}^{-,1}(\lambda;x)\|_{\ell^1} &\leq \left(1+\frac{1}{\sqrt{s^2-1}}\right)\E^{c(s)\|\Delta\psi\|_1},\\
\sup_{x\in\mathbb{R}}\| \mathbf{k}^{+,2}(\lambda;x)\|_{\ell^1} &\leq \left(1+\frac{1}{\sqrt{s^2-1}}\right)\E^{c(s)\|\Delta\psi\|_1}.
\end{align}
where $\| \mathbf{v} \|_{\ell^1}$ denotes the $\ell^1$-norm of a vector $\mathbf{v}\in\mathbb{C}^2$ and
\begin{equation}
c(s)=\frac{1}{2}\left(1+\frac{s}{\sqrt{s^2-1}}\right)+\frac{1}{\sqrt{s^2-1}}+\frac{1}{2(s^2-1)}.
\end{equation}
\end{lemma}
\begin{proof}
Fix $|\lambda|\geq s$ for some $s>1$ with $\Im \{\lambda \}\geq 0$. Then from \eqref{eq:Volterra-NZBC} it is seen that
\begin{align}
\label{eq:kminus-col-1-E}
\mathbf{k}^{-,1}(\lambda;x) &= \mathbf{e}^{1}(\lambda)+ \int\limits_{-\infty}^x \mathbf{E}(\lambda)\begin{bmatrix}1 & 0 \\ 0 & \E^{2\I\rho(\lambda)(x-y)} \end{bmatrix}\mathbf{E}(\lambda)^{-1}\Delta\mathbf{\Psi}(y)\mathbf{k}^{-,1}(\lambda;y)\,\D y \\
\mathbf{k}^{+,2}(\lambda;x) &= \mathbf{e}^{2}(\lambda)+ \int\limits_{+\infty}^x \mathbf{E}(\lambda)\begin{bmatrix} \E^{-2\I\rho(\lambda)(x-y)} & 0 \\ 0 &1 \end{bmatrix}\mathbf{E}(\lambda)^{-1}\Delta\mathbf{\Psi}(y)\mathbf{k}^{+,2}(\lambda;y)\,\D y
\label{eq:kplus-col-2-E}
\end{align}
To solve the integral equation \eqref{eq:kminus-col-1-E}, we set $\mathbf{C}^{[1]}(\lambda;x,y)= \mathbf{E}(\lambda)\diag\big(1,\E^{2\I\rho(\lambda)(x-y)}\big)\mathbf{E}(\lambda)^{-1}\Delta\mathbf{\Psi}(y)$ and introduce the Neumann series
\begin{equation}
\mathbf{k}^{-,1}(\lambda;x) = \sum\limits_{n=0}^{\infty}\boldsymbol {\omega}^{[1]}_n(\lambda;x,t),
\end{equation}
where 
\begin{equation}
\boldsymbol{\omega}^{[1]}_0\defeq \mathbf{e}^{1}(\lambda),\quad \boldsymbol{\omega}^{[1]}_{n+1}(\lambda;x)\defeq \int\limits_{-\infty}^{x} \mathbf{C^{[1]}}(\lambda;x,y)\boldsymbol{\omega}^{[1]}_{n}(\lambda;y)\,\D y.
\end{equation}
First, observe that if $\lambda=|\lambda|\E^{\I \theta}$, $\theta\in\mathbb{R}$, then $|\rho(\lambda)|=(|\lambda|^4 + 2 |\lambda|^2\cos(2\theta) + 1)^{1/4}$ which implies that
\begin{equation}
\sqrt{|\lambda|^2 - 1}\leq|\rho(\lambda)|\leq \sqrt{|\lambda|^2 + 1}.
\label{eq:rho-bound}
\end{equation}
Then since $x\mapsto x/\sqrt{x^2-1}$ is a monotone decreasing map for $x>1$ we have
\begin{equation}
\frac{|\lambda|}{|\rho(\lambda)|} \leq \frac{s}{\sqrt{s^2 -1}}.
\end{equation}
Now let $\| \cdot \|$ denote the (subordinate) matrix norm induced by the $\ell^1$-norm on $\mathbb{C}^2$ and observe that
\begin{equation}
\begin{aligned}
\|\mathbf{E}(\lambda)\|\|\mathbf{E}(\lambda)^{-1}\| &= \frac{|\lambda + \rho(\lambda)|}{2|\rho(\lambda)|}(1+|\lambda-\rho(\lambda)|)^2=\frac{1}{|\rho(\lambda)|}+\frac{|\lambda-\rho(\lambda)|+|\lambda+\rho(\lambda)|}{2|\rho(\lambda)|}\\ &\leq \frac{1}{\rho(\lambda)} + \frac{|\lambda|}{2|\rho(\lambda)|} + \frac{1}{2} + \frac{1}{2|\rho(\lambda)|^2} \\&\leq \frac{1}{2}\left(1+\frac{s}{\sqrt{s^2-1}}\right)+\frac{1}{\sqrt{s^2-1}}+\frac{1}{2(s^2-1)}\eqdef c(s).
\end{aligned}
\label{eq:c-definition}
\end{equation}
This uniform bound implies
\begin{equation}
\|\mathbf{C}^{[1]}(\lambda;x,y)\| \leq c(s)\max\{1,\E^{-2\Im\{ \rho(\lambda)\}(x-y)} \}|\Delta\psi(y)|,
\end{equation}
which for $-\infty<y<x$ becomes
\begin{equation}
\|\mathbf{C}^{[1]}(\lambda;x,y)\| \leq c(s)|\Delta\psi(y)|.
\end{equation}
Then for the $n^\mathrm{th}$ iterate in the Neumann series we have
\begin{multline}
\| \boldsymbol{\omega}^{[1]}_n(\lambda;x) \|_{\ell^1}\\
\begin{aligned}
&\leq \|\mathbf{e}^{1}(\lambda)\|_{\ell^1} \int\limits_{\pm \infty}^{x} \int\limits_{\pm \infty}^{x_1} \cdots \int\limits_{\pm \infty}^{x_{n-1}}\|\mathbf{C}^{[1]}(\lambda;x,x_1) \| \cdots \|\mathbf{C}^{[1]}(\lambda;x_{n-1},x_n) \|\,\D x_n\,\, \D x_{n-1}\cdots \D x_1\\
& \leq \|\mathbf{e}^{1}(\lambda)\|_{\ell^1} \int\limits_{\pm \infty}^{x} \int\limits_{\pm \infty}^{x_1} \cdots \int\limits_{\pm \infty}^{x_{n-1}} |\Delta\psi(x_1)||\Delta\psi(x_2)|\cdots |\Delta\psi(x_n)|\cdot c(s)^n \,
\D x_n \cdots \D x_1.
\end{aligned}
\end{multline}
Setting
\begin{equation}
\tau(x_n)=c(s)\int\limits_{\pm\infty}^{x_n} |\Delta\psi(x_{n-1})|\, \D x_{n-1}, \quad x_0=x
\end{equation}
and $\tau_n \defeq \tau(x_n)$ gives
\begin{equation}
\begin{split}
\| \boldsymbol{\omega}^{[1]}_n(\lambda;x) \|_{\ell^1}&\leq  \|\mathbf{e}^{1}(\lambda)\|_{\ell^1} \int\limits_{\pm \infty}^{\tau(x)} \int\limits_{\pm \infty}^{\tau_1} \cdots \int\limits_{\pm \infty}^{\tau_{n-1}}  \,d\tau_n\,d\tau_{n-1}\cdots d\tau_1\\
&= \|\mathbf{e}^{1}(\lambda)\|_{\ell^1}\frac{\tau(x)^n}{n!}\leq \|\mathbf{e}^{1}(\lambda)\|_{\ell^1}\frac{c(s)^n\|\Delta\psi \|^n_1}{n!}.
\end{split}
\label{eq:omega-bound}
\end{equation}
Therefore the Neumann series converges absolutely and uniformly in $x$ for $|\lambda|>s$ restricted to $\Im\{\lambda \}\geq 0$, and we obtain
\begin{equation}
\sup_{x\in\mathbb{R}}\|\mathbf{k}^{-,1}(\lambda;x)\|_{\ell^1} \leq \left(1+\frac{1}{\sqrt{s^2-1}}\right)\E^{c(s)\|\Delta\psi\|_1}.
\label{eq:kminus-col-1-apriori-bound}
\end{equation}
Similarly, because $\| \mathbf{e}^2(\lambda) \|_{\ell^1} = \| \mathbf{e}^1(\lambda) \|_{\ell^1}$ and $\| \diag(\E^{-2\I\rho(x-y)},1) \|\leq 1$ on $x<y<+\infty$, we have
\begin{equation}
\sup_{x\in\mathbb{R}}\|\mathbf{k}^{+,2}(\lambda;x)\|_{\ell^1} \leq \left(1+\frac{1}{\sqrt{s^2-1}}\right)\E^{c(s)\|\Delta\psi\|_1}.
\label{eq:kplus-col-2-apriori-bound}
\end{equation}
\end{proof}

\section{Asymptotic Properties of $a(\lambda)$ for Large $\lambda$}
\label{A:lemma-radius-proof}
\begin{proof}[Proof of Lemma~\ref{lemma:radius}]
Since $a(\lambda;t)=a(\lambda;0)\eqdef a(\lambda)$, we omit the time-dependence of the quantities involved in the proof. We suppose that $\psi(x;0)$ is not identically $1$ since otherwise $a(\lambda)\equiv 1$, and we also a priori assume that $|\lambda|\geq2$, and hence $c(s)$ defined in \eqref{eq:c-definition} satisfies
\begin{equation}
c(s)<2, \quad\text{for~}s>2.
\label{eq:c-bound}
\end{equation}
With his assumption Lemma~\ref{lemma:Neumann} implies the following a priori bounds 
\begin{equation}
\sup_{x\in\mathbb{R}}\|\mathbf{k}^{-,1}(\lambda;x)\|_{\ell^1} \leq 2 \E^{2\|\Delta\psi\|_1},\quad \sup_{x\in\mathbb{R}}\|\mathbf{k}^{+,2}(\lambda;x)\|_{\ell^1} \leq 2 \E^{2\|\Delta\psi\|_1}\,.
\label{eq:columns-a-priori}
\end{equation}
Note that $\mathbf{E}(\lambda)=\mathbb{I}+\I(\lambda-\rho(\lambda))\sigma_1$. 
We write $\mathbf{k}^{-,1}(\lambda;x)\eqdef [u_1(\lambda;x);\,\, v_1(\lambda;x)\E^{2\I\rho(\lambda)x}]^\top$. Then the second component of \eqref{eq:kminus-col-1-E} reads
\begin{multline}
\E^{2\I\rho(\lambda)x} v_1(\lambda;x)= \I(\lambda-\rho(\lambda)) +n(\lambda)^2 \Biggl\{- \int\limits_{-\infty}^x \E^{2\I\rho(\lambda)(x-y)}\Delta\psi(y)^* u_1(\lambda;y) \,\D y \biggr. \\ + \I(\lambda-\rho(\lambda)) \int\limits_{-\infty}^x \left(1-\E^{2\I\rho(\lambda)(x-y)}\right)\Delta\psi(y)  \E^{2\I\rho(\lambda)y}v_1(\lambda;y)\,\D y 
-\ (\lambda-\rho(\lambda))^2 \int\limits_{-\infty}^x\Delta\psi(y)^*u_1(\lambda;y)\biggl. \Biggr\}.
\label{eq:v1-1}
\end{multline}
Note that since $\Delta\psi$, $\psi'\in L^{1}(\mathbb{R})$, we have $\psi, \Delta\psi \in L^{\infty}(\mathbb{R})$. Using this together with the boundedness of $u_1(\lambda;y)$ for $-\infty<y\leq x$, $\Im\rho(\lambda)>0$, integrating by parts in the first integral and using the differential equation \eqref{eq:Lax-x} gives
\begin{multline}
\E^{2\I\rho(\lambda)x} v_1(\lambda;x)={n(\lambda)^2} \I(\lambda-\rho(\lambda)) +n(\lambda)^2 \Biggl\{\frac{\Delta\psi(x)^* u_1(\lambda;x)}{2\I\rho(\lambda)} \Biggr. + \frac{1}{2\I\rho(\lambda)}\mathscr{T}_1[\Delta\psi]\E^{2\I\rho(\lambda)x}v_1(\lambda;x)\\{}
- \frac{1}{2\I\rho(\lambda)}\int\limits_{-\infty}^x \E^{2\I\rho(\lambda)(x-y)}\left(\psi'(y)^*+\I(\rho(\lambda)-\lambda)\Delta\psi(y)^* \right) u_1(\lambda;y)\, \D y \\ - (\lambda-\rho(\lambda))^2 \int\limits_{-\infty}^x\Delta\psi(y)^*u_1(\lambda;y)\, \D y \Biggr. \Biggl. \Biggr\},
\label{eq:v1-2}
\end{multline}
where $\mathscr{T}_1[\Delta\psi]$ denotes the following Volterra integral operator on $L^{\infty}(\mathbb{R})$:
\begin{multline}
\mathscr{T}_1[\Delta\psi]h(x)\defeq-\int\limits_{-\infty}^x \E^{2\I\rho(\lambda)(x-y)} \Delta\psi(y)^*\psi(y)h(y)\,\D y \\ - 2\rho(\lambda)(\lambda-\rho(\lambda))\int\limits_{-\infty}^{x}\left(1-\E^{2\I\rho(\lambda)(x-y)}\right)\Delta\psi(y)h(y)\, \D y.
\end{multline}
The assumption $|\lambda|\geq 2$ implies that $|\lambda-\rho(\lambda)||\rho(\lambda)|<1$ and hence the bound
\begin{equation}
\| \mathscr{T}_1[\Delta\psi]\|_\infty \leq \|\Delta\psi\|_1 (\|\psi\|_\infty + 4)
\end{equation}
on the operator norm of $\mathscr{T}_1[\Delta\psi]$ on $L^{\infty}(\mathbb{R})$. Thus, with the uniform bounds \eqref{eq:columns-a-priori}, if $2|\rho(\lambda)|> \|\Delta\psi\|_1 (\|\psi\|_\infty + 4)$ then we have
\begin{equation}
\begin{aligned}
\| \E^{2\I\rho(\lambda)\cdot} v_1(\lambda;\cdot) \|_\infty &\leq \left\|\left(\mathscr{I}-\tfrac{1}{2\I\rho(\lambda)}\mathscr{T}_1[\Delta\psi]\right)^{-1} \right\|_{\infty}\\
&\quad\quad\quad\quad{}\cdot\frac{1}{|\rho(\lambda)|}\left(2+ \left[\|\Delta\psi\|_\infty + \|\psi'\|_1 + \frac{2\|\Delta\psi\|_1}{|\rho(\lambda)|} \right]2\E^{2\|\Delta\psi\|_1}\right)\\
&\leq \left(1-\frac{\|\Delta\psi\|_1 (\|\psi\|_\infty + 4)}{2|\rho(\lambda)|}\right)^{-1}\\
&\quad\quad\quad\quad{}\cdot\frac{1}{|\rho(\lambda)|}\left(2+ \left[\|\Delta\psi\|_\infty + \|\psi'\|_1 + \frac{2\|\Delta\psi\|_1}{|\rho(\lambda)|} \right]2\E^{2\|\Delta\psi\|_1}\right),
\end{aligned}
\label{eq:v1-sup-bound}
\end{equation}
where we have again used the bound $|\lambda-\rho(\lambda)||\rho(\lambda)|<1$, together with $|\lambda-\rho(\lambda)|^2|\rho(\lambda)|^2< 1/2$, and $|n(\lambda)^2| < 2$ for $|\lambda|\geq 2$. Similarly, the first component of \eqref{eq:kminus-col-1-E} reads
\begin{multline}
u_1(\lambda;x) = {n(\lambda)^2} + n(\lambda)^2 \Biggl\lbrace \int\limits_{-\infty}^x \Delta\psi(x)\E^{2\I\rho(\lambda)y} v_1(y)\, \D y \\
-\I(\lambda-\rho(\lambda)) \int\limits_{-\infty}^x \Delta\psi(y)^*(\E^{2\I\rho(\lambda)(x-y)} - 1)u_1(\lambda;y)\, \D y \Biggr. \\
+(\lambda-\rho(\lambda))^2\int\limits_{-\infty}^{x}\E^{2\I\rho(\lambda)(x-y)}\Delta\psi(y) \E^{2\I\rho(\lambda)y}v_1(\lambda;y)\,\D y
\Biggl. \Biggr \rbrace ,
\end{multline}
from which we directly obtain the bound
\begin{multline}
\|u_1(\lambda;\cdot) - 1 \|_\infty \leq \frac{1}{2|\rho(\lambda)|^2} \\{}+ 2\left(\|\E^{2\I\rho(\lambda)\cdot}v_1(\lambda;\cdot) \|_{\infty}\left(1+ \frac{1}{2|\rho(\lambda)|^2} \right) \|\Delta\psi\|_1 +\frac{2}{|\rho(\lambda)|}\|\Delta\psi\|_12\E^{2\|\Delta\psi\|_1}\right)\,.
\label{eq:u1-sup-bound}
\end{multline}
We now proceed with getting the analogous estimates for the elements of the other column
\begin{equation}
\mathbf{k}^{+,2}(\lambda;x)\eqdef [\E^{-2\I\rho(\lambda)x}u_2(x);\,\, v_2(x)]^\top.
\end{equation}
The first component of \eqref{eq:kplus-col-2-E} reads:
\begin{multline}
\E^{-2\I\rho(\lambda)x}u_2(x)=n(\lambda)^2\I(\lambda-\rho(\lambda)) + n(\lambda)^2 \Biggl\{\int\limits_{+\infty}^x \E^{-2\I\rho(\lambda)(x-y)}\Delta\psi(y)v_{1}(\lambda;y)\,\D y\Biggr.\\
-\I(\lambda-\rho(\lambda))\int\limits_{+\infty}^{x} (1-\E^{-2\I\rho(\lambda)(x-y)})\Delta\psi(y)^* \E^{-2\I\rho(\lambda)x}u_1(\lambda;y)\,\D y \\{}+(\lambda-\rho(\lambda))^2 \int\limits_{+\infty}^{x}\Delta\psi(y)v_2(\lambda;y)\,\D y\Biggl. \Biggl\}.
\end{multline}
As before, integrating by parts in the first integral and using the differential equation \eqref{eq:Lax-x} to eliminate $v_{1y}(\lambda;y)$ gives
\begin{multline}
\E^{-2\I\rho(\lambda)x}u_2(x)=n(\lambda)^2\I(\lambda-\rho(\lambda)) + n(\lambda)^2 \Biggl\{\frac{\Delta\psi(x)v_2(\lambda;x)}{2\I\rho(\lambda)} + \frac{1}{2\I\rho(\lambda)}\mathscr{T}_2[\Delta\psi]\E^{2\I\rho(x)u_2(\lambda;x)}\\
 - \frac{1}{2\I\rho(\lambda)}\int\limits_{+\infty}^x \E^{-2\I\rho(\lambda)(x-y)}\left(\psi'(y) + \Delta\psi(y)\I(\lambda-\rho(\lambda))\right)v_2(\lambda;y)\, \D y\Biggr. \\
 +(\lambda-\rho(\lambda))^2 \int\limits_{+\infty}^{x}\Delta\psi(y)v_2(\lambda;y)\,\D y\Biggl. \Biggl\},
\end{multline}
where
\begin{multline}
\mathscr{T}_2[\Delta\psi]h(x)\defeq \int\limits_{+\infty}^x \E^{-2\I\rho(\lambda)(x-y)}\Delta\psi(y)\psi(y)^*h(y)\,\D y \\+ 2\rho(\lambda)(\lambda-\rho(\lambda))\int\limits_{+\infty}^x \left( 1 - \E^{-2\I\rho(\lambda)(x-y)}\right)\Delta\psi(y)^* h(y)\,\D y.
\end{multline}
Again the assumption $|\lambda|\geq 2$ the bound
\begin{equation}
\| \mathscr{T}_2[\Delta\psi]\|_\infty \leq \|\Delta\psi\|_1 (\|\psi\|_\infty + 4)
\end{equation}
on the operator norm of $\mathscr{T}_2[\Delta\psi]$ on $L^{\infty}(\mathbb{R})$. Thus, if $2|\rho(\lambda)|> \|\Delta\psi\|_1 (\|\psi\|_\infty + 4)$ then we obtain
\begin{multline}
\| \E^{-2\I\rho(\lambda)\cdot} u_2(\lambda;\cdot) \|_\infty
\leq \left(1-\frac{\|\Delta\psi\|_1 (\|\psi\|_\infty + 4)}{2|\rho(\lambda)|}\right)^{-1}\\{}\cdot\frac{1}{|\rho(\lambda)|}\left(2+ \left[\|\Delta\psi\|_\infty + \|\psi'\|_1 + \frac{2\|\Delta\psi\|_1}{|\rho(\lambda)|} \right]2\E^{2\|\Delta\psi\|_1}\right),
\label{eq:u2-sup-bound}
\end{multline}
for $|\lambda|\geq 2$. Similarly, for the second component of \eqref{eq:kplus-col-2-E} we have
\begin{multline}
v_2(\lambda;x) = n(\lambda)^2 + n(\lambda)^2\Biggl \{\int\limits_{+\infty}^{x}-\Delta\psi(y)^* \E^{-2\I\rho(\lambda)y}u_1(\lambda;y)\, \D y \\
+\I(\lambda-\rho(\lambda))\int\limits_{+\infty}^{x}(\E^{-2\I\rho(\lambda)(x-y)} - 1)\Delta\psi(y)v_2(\lambda;y) \\ - (\lambda-\rho(\lambda))^2\int\limits_{+\infty}^{x}\E^{-2\I\rho(\lambda)(x-y)}\Delta\psi(y)^*\E^{-2\I\rho(\lambda)y}u_2(\lambda;y)\,\D y
\end{multline}
from which we can directly obtain the estimate
\begin{multline}
\|v_2(\lambda;\cdot) - 1 \|_\infty \leq \frac{1}{2|\rho(\lambda)|^2} \\+ 2\left(\|\E^{-2\I\rho(\lambda)\cdot}u_1(\lambda;\cdot) \|_{\infty}\left( \frac{2|\rho(\lambda)|^2+1}{2|\rho(\lambda)|^2} \right) \|\Delta\psi\|_1 +\frac{2}{|\rho(\lambda)|}\|\Delta\psi\|_12\E^{2\|\Delta\psi\|_1}\right)\,.
\label{eq:v2-sup-bound}
\end{multline}
For convenience we denote the (identical) bound in the estimates \eqref{eq:u1-sup-bound} and \eqref{eq:v2-sup-bound} by $m(\lambda)$, namely
\begin{equation}
\max\{ \| \E^{-2\I\rho(\lambda)\cdot} u_2(\lambda;\cdot) \|_\infty, \| \E^{2\I\rho(\lambda)\cdot} v_1(\lambda;\cdot) \|_\infty\}\leq m(\lambda),
\label{eq:common-bound-off-diag}
\end{equation}
where
\begin{equation}
m(\lambda)\defeq \left(1-\frac{\|\Delta\psi\|_1 (\|\psi\|_\infty + 4)}{2|\rho(\lambda)|}\right)^{-1}\cdot\frac{1}{|\rho(\lambda)|}\left(2+ \left[\|\Delta\psi\|_\infty + \|\psi'\|_1 + \frac{2\|\Delta\psi\|_1}{|\rho(\lambda)|} \right]2\E^{2\|\Delta\psi\|_1}\right),
\end{equation}
which is well-defined for $2|\rho(\lambda)|>\|\Delta\psi\|_1 (\|\psi\|_\infty + 4)$ and $m(\lambda)>0$, with $m(\lambda)=O(\lambda^{-1})$ as $\lambda\to\infty$ since $\rho(\lambda)=O(\lambda)$ as $\lambda\to\infty$. Then from \eqref{eq:u1-sup-bound} and \eqref{eq:v2-sup-bound} we obtain:
\begin{multline}
\max\{ \|u_1(\lambda;\cdot)- 1\|_\infty, \| v_2(\lambda;\cdot-1)\|_\infty\} \leq 
\frac{1}{2|\rho(\lambda)|^2} + 2{m(\lambda)\|\Delta\psi\|_1} + \frac{m(\lambda)\|\Delta\psi\|_1}{|\rho(\lambda)|^2}\\
+\frac{8\|\Delta\psi\|_1\E^{2\|\Delta\psi\|_1}}{|\rho(\lambda)|}\eqdef q(\lambda),
\label{eq:common-bound-diag}
\end{multline}
where clearly $q(\lambda)>0$ and $q(\lambda)=O(\lambda^{-1})$ as $\lambda\to\infty$. Then from the definition \eqref{eq:a-b-NZBC} of $a(\lambda)$ and the estimates \eqref{eq:common-bound-off-diag} and \eqref{eq:common-bound-diag} it follows that
\begin{equation}
|a(\lambda)-1| \leq 2q(\lambda)+q(\lambda)^2+m(\lambda)^2 = O(\lambda),\quad \lambda\to \infty, \quad\Im\{ \rho(\lambda)\}>0,
\label{eq:a-asymptotics}
\end{equation}
which proves the first claim. To guarantee that $|a(\lambda)|>\frac{1}{2}$ outside a disk centered at the origin in the $\lambda$-plane, it suffices to ensure that $|q(\lambda)|\leq\frac{1}{8}$ and $m(\lambda)\leq\frac{1}{4}$ by the inequality in \eqref{eq:a-asymptotics}. First, from \eqref{eq:common-bound-diag} it is seen that if 
\begin{equation}
|\rho(\lambda)|\geq\max\left\{ 4, 256\| \Delta\psi\|_1\E^{2\| \Delta\psi \|_1}\right\}
\label{eq:condition-1}
\end{equation}
and
\begin{equation}
m(\lambda)\leq \min\left\{\frac{1}{4}, \frac{1}{64\| \Delta\psi \|_1} \right\},
\label{eq:condition-2}
\end{equation}
then $q(\lambda)\leq\frac{1}{8}$. Note also that we have the condition 
\begin{equation}
|\rho(\lambda)|> \frac{1}{2}\|\Delta\psi\|_1 (\|\psi\|_\infty + 4)
\label{eq:condition-3}
\end{equation}
for $m(\lambda)$ to be well-defined and it remains to obtain a lower bound on $|\rho(\lambda)|$ to guarantee \eqref{eq:condition-2}. We first strengthen \eqref{eq:condition-3} and demand
\begin{equation}
|\rho(\lambda)|\geq 2\|\Delta\psi\|_1 (\|\psi'\|_1 + 5)\geq 2\|\Delta\psi\|_1 (\|\psi\|_\infty + 4).
\label{eq:condition-4}
\end{equation}
Then using \eqref{eq:condition-4} we obtain:
\begin{equation}
m(\lambda)\leq\frac{4}{3|\rho(\lambda)|}\left(2+\frac{1}{64}+ 2\left(\|\Delta\psi\|_\infty + \|\psi'\|_1  \right)\E^{2\|\Delta\psi\|_1}\right).
\end{equation}
Thus, to ensure \eqref{eq:condition-2}, one needs to choose $\lambda$ such that
\begin{equation}
|\rho(\lambda)|\geq\frac{1}{3}\max\left\{ 16, 256\|\Delta\psi\|_1  \right\}\cdot \left(2+\frac{1}{64}+ 4 \|\psi'\|_1 \E^{2\|\Delta\psi\|_1}\right)
\label{eq:condition-5}
\end{equation}
since $\|\Delta\psi\|_\infty \leq \|\psi' \|_1$.
Now, to guarantee all of the conditions \eqref{eq:condition-1}, \eqref{eq:condition-4}, and \eqref{eq:condition-5} simultaneously, $\rho(\lambda)$ needs to satisfy
\begin{equation}
|\rho(\lambda)|\geq\max\left\{4,256\|\Delta\psi\|_1\E^{2\|\Delta\psi\|_1}, 2\|\Delta\psi\|_1 (\|\psi'\|_1 + 5),\frac{1}{3}h(\lambda), 256\|\Delta\psi\|_1 h(\lambda)\right\},
\label{eq:condition-6}
\end{equation}
where
\begin{equation}
h(\lambda)\defeq \left(2+2^{-6}+ 4\|\psi'\|_1 \E^{2\|\Delta\psi\|_1}\right).
\end{equation}
Note that the initial assumption $|\lambda|\geq 2$ is absorbed in this condition. Using \eqref{eq:rho-bound}, it is enough to choose:
\begin{equation}
|\lambda| \geq \sqrt{1 + \max\left\{4,256\|\Delta\psi\|_1\E^{2\|\Delta\psi\|_1}, 2\|\Delta\psi\|_1 (\|\psi'\|_1 + 5),\frac{1}{3}h(\lambda), 256\|\Delta\psi\|_1 h(\lambda)\right\}^2}\defeq r[\Delta\psi]\,.
\label{eq:lambda-lower-bound}
\end{equation}
to guarantee \eqref{eq:condition-6}. With this choice $m(\lambda)\leq\frac{1}{4}$ and $q(\lambda)\leq\frac{1}{8}$, and hence
\begin{equation}
|a(\lambda)-1|\leq \frac{21}{64}<\frac{1}{2}\,,
\end{equation}
which implies the second claim.
\end{proof}

\bibliographystyle{siam}
 \bibliography{BilmanMiller-NewIST}

\begin{thebibliography}{10}

\bibitem{Ablowitz1974b}
{\sc M.~J. Ablowitz, D.~J. Kaup, A.~C. Newell, and H.~Segur}, {\em {The inverse
  scattering transform-Fourier analysis for nonlinear problems}}, Stud. Appl.
  Math., 53 (1974), pp.~249--315.

\bibitem{AkhmedievAS09}
{\sc N.~Akhmediev, A.~Ankiewicz, and J.~M. Soto-Crespo}, {\em {Rogue waves and
  rational solutions of the nonlinear Schr{\"{o}}dinger equation}}, Phys. Rev.
  E, 80 (2009), p.~026601.

\bibitem{Akhmediev}
{\sc N.~N. Akhmediev and V.~I. Korneev}, {\em {Modulation instability and
  periodic solutions of the nonlinear Schr{\"{o}}dinger equation}}, Teoret.
  Mat. Fiz., 69 (1986), pp.~189--194.

\bibitem{Ankiewicz2010}
{\sc A.~Ankiewicz, P.~A. Clarkson, and N.~Akhmediev}, {\em {Rogue waves,
  rational solutions, the patterns of their zeros and integral relations}}, J.
  Phys. A Math. Theor., 43 (2010), p.~122002.

\bibitem{BealsCoifman}
{\sc R.~Beals and R.~R. Coifman}, {\em {Scattering and inverse scattering for
  first order systems}}, Commun. Pure Appl. Math., 37 (1984), pp.~39--90.

\bibitem{Bertola2013}
{\sc M.~Bertola and A.~Tovbis}, {\em {Universality for the Focusing Nonlinear
  Schr{\"{o}}dinger Equation at the Gradient Catastrophe Point: Rational
  Breathers and Poles of the Tritronqu{\'{e}}e Solution to Painlev{\'{e}} I}},
  Commun. Pure Appl. Math., 66 (2013), pp.~678--752.

\bibitem{Biondini2014}
{\sc G.~Biondini and G.~Kova{\v{c}}i{\v{c}}}, {\em {Inverse scattering
  transform for the focusing nonlinear Schr{\"{o}}dinger equation with nonzero
  boundary conditions}}, J. Math. Phys., 55 (2014), p.~031506.

\bibitem{BiondiniM2017}
{\sc G.~Biondini and D.~Mantzavinos}, {\em {Long-Time Asymptotics for the
  Focusing Nonlinear Schr{\"{o}}dinger Equation with Nonzero Boundary
  Conditions at Infinity and Asymptotic Stage of Modulational Instability}},
  Commun. Pure Appl. Math.,  (2017).

\bibitem{Chabchoub2014}
{\sc A.~Chabchoub and M.~Fink}, {\em {Time-Reversal Generation of Rogue
  Waves}}, Phys. Rev. Lett., 112 (2014), p.~124101.

\bibitem{Chabchoub2011}
{\sc A.~Chabchoub, N.~P. Hoffmann, and N.~Akhmediev}, {\em {Rogue Wave
  Observation in a Water Wave Tank}}, Phys. Rev. Lett., 106 (2011), p.~204502.

\bibitem{CoddingtonLevinson}
{\sc E.~A. Coddington and N.~Levinson}, {\em {Theory of ordinary differential
  equations}}, McGraw-Hill Book Company, Inc., New York-Toronto-London, 1955.

\bibitem{Cousins2016a}
{\sc W.~Cousins and T.~P. Sapsis}, {\em {Reduced order prediction of rare
  events in unidirectional nonlinear water waves}}, J. Fluid Mech., 790 (2016),
  pp.~368--388.

\bibitem{Deift1991b}
{\sc P.~Deift and X.~Zhou}, {\em {Direct and inverse scattering on the line
  with arbitrary singularities}}, Commun. Pure Appl. Math., 44 (1991),
  pp.~485--533.

\bibitem{Deift1993a}
{\sc P.~Deift and X.~Zhou}, {\em {A Steepest Descent Method for Oscillatory
  Riemann--Hilbert Problems. Asymptotics for the MKdV Equation}}, Ann. Math.,
  137 (1993), p.~295.

\bibitem{Demontis2013}
{\sc F.~Demontis, B.~Prinari, C.~van~der Mee, and F.~Vitale}, {\em {The Inverse
  Scattering Transform for the Defocusing Nonlinear Schr{\"{o}}dinger Equations
  with Nonzero Boundary Conditions}}, Stud. Appl. Math., 131 (2013), pp.~1--40.

\bibitem{Dubard2010b}
{\sc P.~Dubard, P.~Gaillard, C.~Klein, and V.~Matveev}, {\em {On multi-rogue
  wave solutions of the NLS equation and positon solutions of the KdV
  equation}}, Eur. Phys. J. Spec. Top., 185 (2010), pp.~247--258.

\bibitem{FaddeevBook}
{\sc L.~D. Faddeev and L.~A. Takhtajan}, {\em {Hamiltonian methods in the
  theory of solitons}}, Springer-Verlag, Berlin, 1987.

\bibitem{FokasUnifiedArticle}
{\sc A.~S. Fokas}, {\em {A unified transform method for solving linear and
  certain nonlinear PDEs}}, Proc. R. Soc. A Math. Phys. Eng. Sci., 453 (1997),
  pp.~1411--1443.

\bibitem{FokasUnified}
{\sc A.~S. Fokas}, {\em {A Unified Approach to Boundary Value Problems}},
  Society for Industrial and Applied Mathematics, Philadelphia, PA, jan 2008.

\bibitem{Gaillard2013}
{\sc P.~Gaillard}, {\em {Degenerate determinant representation of solutions of
  the nonlinear Schr{\"{o}}dinger equation, higher order Peregrine breathers
  and multi-rogue waves}}, J. Math. Phys., 54 (2013), p.~013504.

\bibitem{Gaillard2014a}
{\sc P.~Gaillard}, {\em {Higher order Peregrine breathers, their deformations
  and multi-rogue waves}}, J. Phys. Conf. Ser., 482 (2014), p.~012016.

\bibitem{Guo2012a}
{\sc B.~Guo, L.~Ling, and Q.~P. Liu}, {\em {Nonlinear Schr{\"{o}}dinger
  equation: Generalized Darboux transformation and rogue wave solutions}},
  Phys. Rev. E - Stat. Nonlinear, Soft Matter Phys., 85 (2012), p.~026607.

\bibitem{GuoLZLY2017-book}
{\sc B.~Guo, L.~Tian, Z.~Yan, L.~Ling, and Y.-F. Wang}, {\em {Rogue waves.
  Mathematical theory and applications in physics}}, De Gruyter, Berlin, 2017.

\bibitem{Hammani2011}
{\sc K.~Hammani, B.~Kibler, C.~Finot, P.~Morin, J.~Fatome, J.~M. Dudley, and
  G.~Millot}, {\em {Peregrine soliton generation and breakup in standard
  telecommunications fiber}}, Opt. Lett., 36 (2011), p.~112.

\bibitem{HeZWPF13}
{\sc J.~S. He, H.~R. Zhang, L.~H. Wang, K.~Porsezian, and A.~S. Fokas}, {\em
  {Generating mechanism for higher-order rogue waves}}, Phys. Rev. E, 87
  (2013), p.~052914.

\bibitem{Kaup1976b}
{\sc D.~J. Kaup}, {\em {Closure of the squared Zakharov-Shabat eigenstates}},
  J. Math. Anal. Appl., 54 (1976), pp.~849--864.

\bibitem{Kharif2009}
{\sc C.~Kharif, E.~Pelinovsky, and A.~Slunyaev}, {\em {Rogue Waves in the
  Ocean}}, Advances in Geophysical and Environmental Mechanics and Mathematics,
  Springer Berlin Heidelberg, Berlin, Heidelberg, 2009.

\bibitem{Kibler2010}
{\sc B.~Kibler, J.~Fatome, C.~Finot, G.~Millot, F.~Dias, G.~Genty,
  N.~Akhmediev, and J.~M. Dudley}, {\em {The Peregrine soliton in nonlinear
  fibre optics}}, Nat. Phys., 6 (2010), pp.~790--795.

\bibitem{Kuznetsov1977}
{\sc E.~A. Kuznetsov}, {\em {Solitons in parametrically unstable media}}, Dokl.
  Akad. Nauk SSR, 22 (1977), pp.~507--508.

\bibitem{MillerL2007}
{\sc G.~D. Lyng and P.~D. Miller}, {\em {The $N$-soliton of the focusing
  nonlinear Schr{\"{o}}dinger equation for $N$ large}}, Commun. Pure Appl. Math.,
  60 (2007), pp.~951--1026.

\bibitem{Ma1979b}
{\sc Y.-C. Ma}, {\em {The Perturbed Plane-Wave Solutions of the Cubic
  Schr{\"{o}}dinger Equation}}, Stud. Appl. Math., 60 (1979), pp.~43--58.

\bibitem{MatveevS1991-book}
{\sc V.~B. Matveev and M.~A. Salle}, {\em {Darboux transformations and
  solitons}}, Springer-Verlag, Berlin, 1991.

\bibitem{Peregrine1983}
{\sc D.~H. Peregrine}, {\em {Water waves, nonlinear Schr{\"{o}}dinger equations
  and their solutions}}, J. Aust. Math. Soc. Ser. B. Appl. Math., 25 (1983),
  p.~16.

\bibitem{Suret2016}
{\sc P.~Suret, R.~E. Koussaifi, A.~Tikan, C.~Evain, S.~Randoux, C.~Szwaj, and
  S.~Bielawski}, {\em {Single-shot observation of optical rogue waves in
  integrable turbulence using time microscopy}}, Nat. Commun., 7 (2016),
  p.~13136.

\bibitem{Tajiri1998a}
{\sc M.~Tajiri and Y.~Watanabe}, {\em {Breather solutions to the focusing
  nonlinear Schr{\"{o}}dinger equation}}, Phys. Rev. E, 57 (1998),
  pp.~3510--3519.

\bibitem{Tikan2017}
{\sc A.~Tikan, C.~Billet, G.~El, A.~Tovbis, M.~Bertola, T.~Sylvestre,
  F.~Gustave, S.~Randoux, G.~Genty, P.~Suret, and J.~M. Dudley}, {\em
  {Universality of the Peregrine Soliton in the Focusing Dynamics of the Cubic
  Nonlinear Schr{\"{o}}dinger Equation}}, Phys. Rev. Lett., 119 (2017),
  p.~033901.

\bibitem{Walczak2015}
{\sc P.~Walczak, S.~Randoux, and P.~Suret}, {\em {Optical Rogue Waves in
  Integrable Turbulence}}, Phys. Rev. Lett., 114 (2015), p.~143903.

\bibitem{WangYWH17}
{\sc L.~Wang, C.~Yang, J.~Wang, and J.~He}, {\em {The height of an $n$-th-order
  fundamental rogue wave for the nonlinear Schr{\"{o}}dinger equation}}, Phys.
  Lett. A, 381 (2017), pp.~1714--1718.

\bibitem{Shabat1972}
{\sc V.~E. Zakharov and A.~B. Shabat}, {\em {A scheme for integrating the
  nonlinear equations of mathematical physics by the method of the inverse
  scattering problem. I}}, Funct. Anal. Its Appl., 8 (1974), pp.~226--235.

\bibitem{Zhou1989d}
{\sc X.~Zhou}, {\em {Direct and inverse scattering transforms with arbitrary
  spectral singularities}}, Commun. Pure Appl. Math., 42 (1989), pp.~895--938.

\bibitem{Zhou1989e}
\leavevmode\vrule height 2pt depth -1.6pt width 23pt, {\em {The Riemann-Hilbert
  Problem and Inverse Scattering}}, SIAM J. Math. Anal., 20 (1989),
  pp.~966--986.

\end{thebibliography}

\end{document}